\documentclass[onecolumn,letterpaper,12pts]{quantumarticle}
\pdfoutput=1

\usepackage{arabtex}    
\usepackage{authblk}
\usepackage{blindtext}
\usepackage[utf8]{inputenc}
\usepackage[english]{babel}
\usepackage[T1]{fontenc}
\usepackage{hyperref}
\usepackage{lipsum}
\usepackage[margin=3cm]{geometry}
\usepackage[numbers,sort&compress]{natbib}
\usepackage{amsmath, amsthm, amsfonts, amssymb, bm, bbm}
\usepackage{booktabs}
\usepackage{graphicx}
\usepackage{listings}
\usepackage[braket]{qcircuit}
\usepackage{hyperref, cleveref}
\hypersetup{
    colorlinks,
    linkcolor={blue},
    citecolor={blue},
    urlcolor={blue}
}
\usepackage{pgfplots}
\pgfplotsset{compat=1.15}
\usetikzlibrary{external}
\usetikzlibrary{backgrounds,shapes,positioning}
\usetikzlibrary{calc,decorations.markings}
\theoremstyle{plain}
\newtheorem{theorem}{Theorem}
\newtheorem{lemma}[theorem]{Lemma}

\theoremstyle{definition}
\newtheorem{definition}{Definition}

\newcommand{\bb}[1]{\ensuremath{\mathbb{#1}}}

\newcommand{\dr}{\ensuremath{\rangle\!\rangle}}
\newcommand{\dl}{\ensuremath{\langle\!\langle}}
\newcommand{\dket}[1]{\ensuremath{|#1\dr}}
\newcommand{\dbra}[1]{\ensuremath{\dl#1|}}

\makeatletter
\def\ketbra#1{\def\tempa{#1}\futurelet\next\ketbra@i}
\def\ketbra@i{\ifx\next\bgroup\expandafter\ketbra@ii\else\expandafter\ketbra@end\fi}
\def\ketbra@ii#1{\left| \tempa \middle\rangle\!\middle\langle #1 \right|}
\def\ketbra@end{\left| \tempa \middle\rangle\!\middle\langle \tempa \right|}
\makeatother

\makeatletter
\def\braket#1{\def\tempa{#1}\futurelet\next\braket@i}
\def\braket@i{\ifx\next\bgroup\expandafter\braket@ii\else\expandafter\braket@end\fi}
\def\braket@ii#1{\left\langle \tempa \middle| #1 \right\rangle}
\def\braket@end{\left\langle \tempa \middle| \tempa \right\rangle}
\makeatother

\makeatletter
\def\dketbra#1{\def\tempa{#1}\futurelet\next\dketbra@i}
\def\dketbra@i{\ifx\next\bgroup\expandafter\dketbra@ii\else\expandafter\dketbra@end\fi}
\def\dketbra@ii#1{\left| \tempa \middle\dr\!\middle\dl #1 \right|}
\def\dketbra@end{\left| \tempa \middle\dr\!\middle\dl \tempa \right|}
\makeatother

\makeatletter
\def\dbraket#1{\def\tempa{#1}\futurelet\next\dbraket@i}
\def\dbraket@i{\ifx\next\bgroup\expandafter\dbraket@ii\else\expandafter\dbraket@end\fi}
\def\dbraket@ii#1{\left\dl \tempa \middle| #1 \right\dr}
\def\dbraket@end{\left\dl \tempa \middle| \tempa \right\dr}
\makeatother

\AtBeginDocument{
\heavyrulewidth=.08em
\lightrulewidth=.05em
\cmidrulewidth=.03em
\belowrulesep=.65ex
\belowbottomsep=0pt
\aboverulesep=.4ex
\abovetopsep=0pt
\cmidrulesep=\doublerulesep
\cmidrulekern=.5em
\defaultaddspace=0.5em
}
\usepackage{hyperref,graphicx, amsmath,amssymb, amsthm, color, bm, bbm,dsfont, cleveref,floatrow,enumerate,media9}
\usepackage[most]{tcolorbox}
\usepackage{graphicx}
\usepackage{array}
\usepackage{calc}
\newcolumntype{L}[1]{>{\raggedright\let\newline\\\arraybackslash\hspace{0pt}}m{#1}}
\newcolumntype{C}[1]{>{\centering\let\newline\\\arraybackslash\hspace{0pt}}m{#1}}
\newcolumntype{R}[1]{>{\raggedleft\let\newline\\\arraybackslash\hspace{0pt}}m{#1}}
\usepackage{siunitx}
\usepackage{cancel}
\DeclareMathOperator*{\argmax}{argmax}

\usepackage{thmtools}
\usepackage{enumitem}
\usepackage{blindtext}
\declaretheoremstyle[
notefont=\bfseries, notebraces={}{},
bodyfont=\normalfont,
headformat=\NAME~\NUMBER:\NOTE
]{nopar}

\newtcbtheorem{boxthm}{Theorem}{
	enhanced,
	sharp corners,
	attach boxed title to top left={yshift=-\tcboxedtitleheight/2, xshift=3mm,yshifttext=-\tcboxedtitleheight/2},
	colback=white,
	colframe=red!20,
	coltitle=black,
	boxrule=0.7mm,
	fonttitle=\bfseries,
	boxed title style={
		sharp corners,
		size=small,
		colback=red!20,
		colframe=red!20,
	} 
}{thm}
\crefname{tcb@cnt@boxthm}{theorem}{theorems}

\newtcbtheorem{boxlem}{Lemma}{
	enhanced,
	sharp corners,
	attach boxed title to top left={yshift=-\tcboxedtitleheight/2, xshift=3mm,yshifttext=-\tcboxedtitleheight/2},
	colback=white,
	colframe=blue!20,
	coltitle=black,
	boxrule=0.7mm,
	fonttitle=\bfseries,
	boxed title style={
		sharp corners,
		size=small,
		colback=blue!20,
		colframe=blue!20,
	} 
}{lem}
\crefname{tcb@cnt@boxlem}{lemma}{lemmas}

\newtcbtheorem{repboxlem}{Lemma 2}{
	enhanced,
	sharp corners,
	attach boxed title to top left={yshift=-\tcboxedtitleheight/2, xshift=3mm,yshifttext=-\tcboxedtitleheight/2},
	colback=white,
	colframe=blue!20,
	coltitle=black,
	boxrule=0.7mm,
	fonttitle=\bfseries,
	boxed title style={
		sharp corners,
		size=small,
		colback=blue!20,
		colframe=blue!20,
	} 
}{lem}
\crefname{tcb@cnt@repboxlem}{lemma}{lemmas}

\newtcbtheorem{boxprop}{Proposition}{
	enhanced,
	sharp corners,
	attach boxed title to top left={yshift=-\tcboxedtitleheight/2, xshift=3mm,yshifttext=-\tcboxedtitleheight/2},
	colback=white,
	colframe=red!20,
	coltitle=black,
	boxrule=0.7mm,
	fonttitle=\bfseries,
	boxed title style={
		sharp corners,
		size=small,
		colback=red!20,
		colframe=red!20,
	} 
}{prop}
\crefname{tcb@cnt@boxprop}{proposition}{propositions}

\newtcbtheorem{boxdefn}{Definition}{
	enhanced,
	sharp corners,
	attach boxed title to top left={yshift=-\tcboxedtitleheight/2, xshift=3mm,yshifttext=-\tcboxedtitleheight/2},
	colback=white,
	colframe=green!55!black!15,
	coltitle=black,
	boxrule=0.7mm,
	fonttitle=\bfseries,
	boxed title style={
		sharp corners,
		size=small,
		colback=green!55!black!15,
		colframe=green!55!black!15,
	} 
}{defn}
\crefname{tcb@cnt@boxdefn}{definition}{definitions}

\newtcbtheorem{boxconj}{Conjecture}{
	enhanced,
	sharp corners,
	attach boxed title to top left={yshift=-\tcboxedtitleheight/2, xshift=3mm,yshifttext=-\tcboxedtitleheight/2},
	colback=white,
	colframe=yellow!20,
	coltitle=black,
	boxrule=0.7mm,
	fonttitle=\bfseries,
	boxed title style={
		sharp corners,
		size=small,
		colback=yellow!20,
		colframe=yellow!20,
	} 
}{conj}
\crefname{tcb@cnt@boxconj}{conjecture}{conjectures}

\newtcbtheorem{boxcoro}{Corollary}{
	enhanced,
	sharp corners,
	attach boxed title to top left={yshift=-\tcboxedtitleheight/2, xshift=3mm,yshifttext=-\tcboxedtitleheight/2},
	colback=white,
	colframe=green!55!black!15,
	coltitle=black,
	boxrule=0.7mm,
	fonttitle=\bfseries,
	boxed title style={
		sharp corners,
		size=small,
		colback=green!55!black!15,
		colframe=green!55!black!15,
	} 
}{coro}
\crefname{tcb@cnt@boxcoro}{corollary}{corollaries}

\newcounter{theo}

\crefname{theo}{theorem}{theorems}

\newcounter{lembx}

\crefname{lembx}{lemma}{lemmas}

\newcounter{defbx}

\crefname{defbx}{definition}{definitions}
\usepackage[framemethod=TikZ]{mdframed}
\usepackage[most]{tcolorbox}
\theoremstyle{plain}

\theoremstyle{definition}

\crefname{thm}{theorem}{theorems}
\newtheorem{protocol}{Protocol}

\crefname{subroutine}{subroutine}{subroutines}
\crefname{protocol}{protocol}{protocols}

\definecolor{nblue}{rgb}{0.2,0.2,0.7}
\definecolor{ngreen}{rgb}{0.1,0.5,0.1}
\definecolor{nred}{rgb}{0.8,0.2,0.2}
\definecolor{nblack}{rgb}{0,0,0}

\DeclareMathOperator{\tr}{\mathrm{Tr}}

\usepackage{scalerel}
\usepackage{stackengine,wasysym}

\newcommand{\mc}[1]{\mathcal{#1}}

\newcommand{\mbb}[1]{\mathbb{#1}}

\newcommand{\inner}[2]{\left\langle#1,#2\right\rangle}

\newcommand{\set}[1]{\mathbb{#1}} 

\newcommand{\rep}[1]{\phi(#1)} 
\newcommand{\adj}[0]{\nu_{\rm drs.}^{\rm adj.}}
\newcommand{\eff}[0]{\nu_{\rm drs.}^{\rm eff.}}

\newcommand{\nrep}[1]{\nu{(#1)}} 

\usepackage{xcolor}

\newcommand{\fidorbit}[2]{f(\orbit{#1}{#2}) }
\newcommand{\hatfidorbit}[2]{\hat{f}(\orbit{#1}{#2}) }
\newcommand{\orbit}[2]{#1^{\circlearrowright#2}}
\newcommand{\params}[0]{\bm \theta}
\newcommand{\cali}[0]{\set R^m}
\newcommand{\pcal}[1]{p^{\rm cal}(#1|\params)}
\newcommand{\pncal}[1]{p^{\rm \lnot cal}(#1)}

\newcommand{\hidden}[1]{}

\newcommand{\qeclit}[0]{Wagner2022, Wootton2022,Hashim2022,Piveteau2021}

\newcommand{\rcisgood}[0]{Hashim2021,ville2021, Gu2022}

\newcommand{\learnability}[0]{Huang2022,Chen2022}

\newcommand{\hamadd}[0]{Chuang1997,Branderhorst_2009,Shabani2011,Flammia2012,Aaronson2017,Huang2020,Evans2019,FlammiaPopRec,Kunjummen2021,Levy2021,Bertoni2022,Akhtar2022,Hu2022,Chen_2021,Anshu_2021,Wilde2022,Franca2022,Haah2021}

\usetikzlibrary{shapes.geometric,arrows}
\begin{document}

\title{The Error Reconstruction and Compiled Calibration of Quantum Computing Cycles}

\author[1,4,*]{Arnaud Carignan-Dugas} 
\author[2,*]{Dar Dahlen}
\author[1,4,*]{Ian Hincks}
\author[1,4,*]{Egor Ospadov}
\author[1,3,4,*]{Stefanie J. Beale}
\author[1,4,*]{Samuele Ferracin}
\author[1,4,*]{Joshua Skanes-Norman}
\author[1,3,4,*]{Joseph Emerson}
\author[1,3,4,*]{Joel J. Wallman}

\affil[1]{Quantum Benchmark Inc., 51 Breithaupt Street, Suite 100
Kitchener, ON N2H 5G5, Canada}
\affil[2]{Quantum Benchmark Inc., 2081 Center St
Berkeley, CA }
\affil[3]{Institute for Quantum Computing, University of Waterloo, Waterloo, Ontario N2L 3G1, Canada}
\affil[4]{Keysight Technologies Canada, Kanata, ON K2K 2W5, Canada }
\affil[*]{Affiliations at the time of conducted research. May not correspond to current authors' affiliations.}

\begin{abstract}
Quantum computers are inhibited by physical errors that occur during computation.
For this reason, the development of increasingly sophisticated error  characterization and error suppression techniques is central to the progress of quantum computing.
Error distributions are considerably influenced by the precise gate scheduling across the entire quantum processing unit.
To account for this holistic feature, we may ascribe each error profile to a (clock) cycle, which is a scheduled list of instructions over an arbitrarily large fraction of the chip.
A celebrated technique known as randomized compiling introduces some randomness within cycles' instructions, which yields effective cycles with simpler, stochastic error profiles.
In the present work, we leverage the structure of cycle benchmarking (CB) circuits as well as known Pauli channel estimation techniques to derive a method, which we refer to as cycle error reconstruction (CER), to estimate with multiplicative precision the marginal error distribution associated with any effective cycle of interest. The CER protocol is designed to scale for an arbitrarily large number of qubits. 
Furthermore, we develop a fast compilation-based calibration method, referred to as stochastic calibration (SC), to identify and suppress local coherent error sources occurring in any effective cycle of interest.
We performed both protocols on IBM-Q 5-qubit devices.
Via our calibration scheme, we obtained up to a 5-fold improvement of the circuit performance.
\end{abstract}

\maketitle

\section{Introduction}
Many groups around the world are designing increasingly large quantum computers to try to outperform conventional computers.
Manipulating large quantum systems for sufficiently long periods of time is a major technical challenge because
every step of a quantum computation is prone to physical errors that are unavoidable due to the inherent fragility of quantum systems to external interactions. 
For this reason, error diagnostic and error suppression tools are central to the progress of quantum computing.
Error diagnostic tools are useful because they enable hardware developers to 
identify and quantify the sources of noise in their devices and users of 
quantum computers to test or anticipate the performance of their computation.
Error suppression tools allow users of quantum hardware to further reduce the effect
of errors that remain after the devices are optimized.
In this paper, we introduce both diagnostic and optimization routines, which we refer to as Cycle Error Reconstruction (CER) and Stochastic Calibration (SC), respectively.

 A quantum processing unit (QPU) is a collection of local quantum $d$-level systems, hereafter referred to as qudits, which can be subjected to instructions such as quantum gates and measurements. Often, individual systems are 2-dimensional, and are referred to as qubits. Given their commonality, we will refer to local systems as qubits throughout the paper, although most of the discussion and results apply to more general qudits.  
A \emph{circuit} is a decomposition of a computing task into an explicit list of parallel and/or sequential instructions to apply to qudits. The most common type of instruction is a gate, but instructions can also designate state preparation and measurement operations.
To run a circuit through a QPU, it must be ultimately phrased into instructions that are native to the device’s inner-workings; often, these primitive instructions address only one or two qudits at a time. We shall refer to the set of qubits targeted by an instruction as the instruction's support. 
It is useful to express a circuit as a list of time steps where each step consists of parallel (although not necessarily simultaneously implemented) instructions applied to non-overlapping supports; each such step is referred to as a \emph{cycle}. A cycle is defined with a fixed timing schedule: two identical set of parallel instructions performed with regards to different schedules correspond to two different cycles.

In all QPU implementations, circuits are subject to error processes. 
Process error diagnostics form a wide family, since the meaning of a ``process'' can range from a physical pulse to an entire circuit.

At one end of the spectrum, dynamic-centric protocols 
are focused on the characterization of the infinitesimal generators of motions involved in 
a controlled dynamical evolution (see e.g. \cite{\hamadd}). For example, this type of diagnostics may aim at 
learning the time-dependent Hamiltonian induced in a pulse, or the leading Lindbladian 
jumps affecting the master equation of motion. While information about the dynamics
plays an important role in e.g. pulse-level control, it may not trivially translate into instruction-level error profiles. Even if 
the translation were perfectly reliable, it remains less efficient and precise to obtain instruction-level error profiles from Lindbladian-level information, 
since a fine time resolution of the Lindbladian is necessary 
to integrate the equations of motion that lead to instructions. 

On the other side of the spectrum, high-level benchmarking methods 
provide broad performance metrics which can be tied 
to sets of instructions as in Randomized Benchmarking schemes (RB) \cite{Emerson2005,DCEL2006,Magesan2011,Magesan2012a,Proctor2017,Wallman2017,Merkel2018,Dugas2018,
Magesan2012b,GambettaCorcoles2012,Barends2014dice,Epstein2014,Granade2014,Wallman2014,Dugas2015,
Wallman2015,Wallman2015b,Sheldon2016,Cross2016,Combes2017,Hashagen2018,Brown2018,Franca2018,
Helsen2018,Onorati2018,Proctor2018,Boone2019,Helsen2019,Dirkse2019,Harper2019,Erhard2019,Helsen2020,Proctor2022,Hines2022,Polloreno2023_DRB},  
or to a collection of circuits as in e.g. quantum volume (QV) 
benchmarking \cite{Quantum_volume2019}, cross-entropy benchmarking 
(XEB) \cite{XEB2016,linearxeb}, circuit mirroring \cite{Proctor2020_capabilities}. Such schemes are useful for
QPU cross-comparisons, as well as to broadly delimitate a QPU's 
range in terms of accessible circuit width and depth (see e.g. \cite{RBK2020_volumetric}). Other circuit-level
benchmarking schemes include output accreditation \cite{Ferracin_2019_accreditation, Ferracin2021accreditation}, 
where the idea is to bound the probability of failure given a specific circuit application.

In this work, we add to the family of error diagnostic protocols 
that fall between the two extremes described above. Our goal is 
to characterize -- in details -- elementary processes that generate universal 
quantum computation and that are closely connected to the 
effective discrete instructions present in quantum circuits. 
Characterizing individual instructions in isolation has
important disadvantages since in modern QPUs, error correlations when considering multiple parallel instructions are important. More generally, the global error induced by a computing cycle depends on an \emph{a priori} 
unknown function of the joint instructions at play, as well as on the relative 
timing of those instructions. As such, to ensure that diagnostics schemes 
remain closely tied to the integrated QPU performance, 
cycle-centric characterization approaches were developed. The 
idea is that a cycle captures more context (i.e. the 
combination of instructions and their specific 
schedule) and that such context plays an important role in determining the error profile across the QPU.
Cycle-centric schemes prior to CER include Cycle 
Benchmarking (CB) \cite{Erhard2019}, Average Circuit Eigenvalue Sampling 
(ACES) \cite{FlammiaACES}, and Gate-Set 
Tomography (GST) \cite{Merkel2013,Blume-Kohout2013,Blume_Kohout_2017,Greenbaum2015, Nielsen_2021,Brieger2021,Rudinger_2021_crosstalk}. Note that in GST, cycles are referred to as ($n$-qubit) gates. 

Accounting for all instructions and their relative timing over the QPU is certainly a 
crucial step toward accurately anticipating the performance of computation, but there are 
other possible considerations to make as well. In practice, 
the results of a specific application circuit can generally be improved by 
randomly sampling over equivalent lists of instructions. 
This common error suppression technique, which is widely known to as 
randomized compiling (RC) \cite{Wallman2016}, transforms deterministic instructions into random variables. The effect of RC is highly desirable in many different areas of quantum computing, as it improves the predictability and performance of quantum computations \cite{\rcisgood}.

As we show in more details in \cref{{subsec:effectivedressedcycles}}, the average circuit obtained from randomized compiling can be closely approximated by a sequence of appropriately averaged random cycles which we
refer to as ``effective dressed cycles'' (see \cref{def:edc}). Characterizing effective dressed cycles rather than deterministic instructions on their supported registers provides a relevant and holistic context to error profiles, and thereby closes the gap between errors determined under diagnostic routines and the effective error affecting the application performance. The direct characterization of effective dressed cycles is
proper to CB, CER and ACES, while standard GST involves the characterization of deterministic cycles. 

Both CB and CER differentiate themselves from ACES in that they are expressly built to retrieve error diagnostic estimates of a single effective dressed cycle with 
\emph{multiplicative} precision. That is, for estimating an error rate 
$\epsilon$ to a few decimals, CB and CER only require $O(1/\rm{polylog}
(\epsilon))$ number of runs as opposed to the $O(1/\epsilon^2)$ scaling 
proper to  additive-precision protocols. This is an important feature since 
current error rates can reach $10^{-4}$, which would require over $10^8$ 
runs to resolve with additive-precision protocols.
This limitation of additive-precision protocols arises due to the need to 
resolve fluctuations introduced by shot noise.
The multiplicative precision on the estimate of an error rate $\epsilon$ is 
essentially obtained from the $O(1/\epsilon)$ repetition of a 
fixed target cycle in a circuit \cite{Harper2019}. The strategy 
to boost the precision through channel repetition has the added 
benefit to decouple the cycle error profile from state 
preparation and measurement (SPAM) errors and can be traced back 
to RB. That said, this strategy is not unique to RB schemes, and 
is featured in e.g. long-sequence GST (LGST) \cite{Nielsen_2021}.
Since the trick necessary to obtain a multiplicative precision 
of a targeted cycle error profile involves the purposeful 
substantial degradation of the signal, it means that the number 
of error diagnostic circuits necessary for schemes such as CB 
and CER scales as the number of effective dressed cycles of 
interest. 

Since cycles are composed of local instructions over the QPU, 
combinatorics tell us that there are an exponential number of possible 
cycles in the number of qubits. In practice, however, only a few 
distinct entangling cycles -- which dominate the error profile of the 
effective dressed cycles --  might appear in a quantum computing 
program. Striking examples are the quantum supremacy circuits presented 
in \cite{supremacy2019}, which involved only 4 distinct entangling 
cycles over the 54-qubit Sycamore processor. It is precisely in those 
scenarios where circuits involve few distinct entangling cycles that 
cycle-specific protocols such as CB and CER become particularly relevant.

CB pioneered the characterization of effective dressed 
cycles, and introduced an overall circuit structure that we
re-use in both CER and SC. The structure is illustrated in \cref{fig:cb_structure}. 
As originally formulated, CB provides the total error probability of an effective dressed cycle 
of interest through a randomized SPAM strategy. Simple modifications to 
the SPAM strategy and classical post-processing, inspired by analysis tools
developed in \cite{FW2019,HWFW2019}, 
can allow CB-structured circuits to yield much more than the total error probability. We captured 
different modifications of CB-structured circuits in \cref{fig:pie}.

CER is a CB-structured circuit specifically designed 1) to yield the marginal error distributions
generated by any Clifford effective dressed cycle of interest in an efficient and scalable manner 
(i.e. the noise under RC), and 2) to get estimates up to multiplicative precision 
(i.e. requiring a $O(1/\rm{polylog}(\epsilon))$ number of runs for any error rate $\epsilon$ instead of a number of runs that grows as $1/\epsilon^2$). 
This combined pair of properties is particularly relevant to error mitigation and correction. Indeed, 
multiplicative precision will become crucial if  
we port such characterization schemes to extremely 
low logical error rates (e.g. a $10^{-8}$ logical error rate may require a prohibitive $10^{16}$ runs
to be estimated by an additive-precision protocol). In the context of error mitigation, CER pairs well with a technique known as 
Probabilistic Error Cancellation (PEC) since it provides precise 
error probability estimates directly associated with effective QPU instructions \cite{Temme_2017,Berg2022, Ferracin2022}.
The decay rates estimated in CER are proportional to marginal
(as opposed to total) error probabilities over limited qubit supports, meaning that they do not grow in
the cycle's width unless the process of scaling the architecture itself increases near-local error
probabilities. This makes CER very scalable.

\begin{figure}[h!]
\centering
\def\svgwidth{\columnwidth}
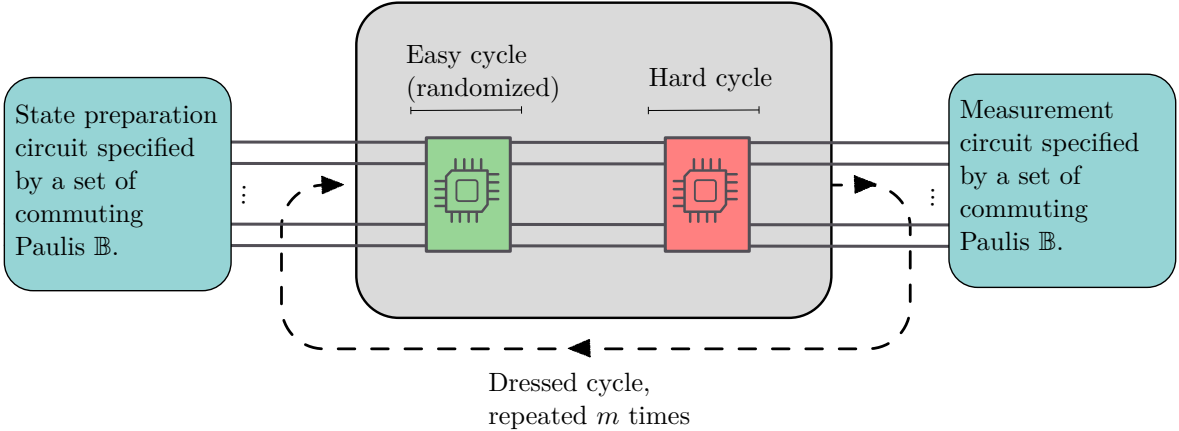
\caption{CB-structured circuits consist of repeated dressed cycles wrapped between state preparation and measurement circuits. The randomization of the easy cycles tailors the dressed cycle into an effective dressed cycle with a stochastic error profile. In this work, we introduce two CB-structured protocols, namely CER and SC. Their differences are highlighted in \cref{fig:pie}.
}
\label{fig:cb_structure}
\end{figure}

In \cite{ORBIT2014}, a protocol known as 
``\emph{Optimized Randomized Benchmarking for Immediate Tune-up}'' (ORBIT)  
is proposed as a quick means to optimize Clifford gate design. 
The idea is to use the results from parallel rounds of 
Clifford-based Interleaved RB as a quickly accessible optimization function. 
The main shortcoming of this method is that Interleaved RB 
is subject to important systematic errors that can contaminate the 
optimization function \cite{Dugas2019unitarity}. To avoid the inaccuracies stemming from Interleaved RB,
one could instead resort to CB as a medium
to quickly explore the optimization landscape. However, this would still be a sub-optimal 
calibration strategy since CB provides a total error probability estimate, a piece of information
that is agnostic to (usually local or near-local) calibration constraints.
In this work, we introduce a fast calibration routine 
-- referred to as {Stochastic Calibration} (SC)-- which leverages the 
knowledge of calibration constraints to reduce the number of necessary SPAM strategies.
This CB-structured scheme retains the multiplicative precision of CB, 
but generally requires fewer circuits, and 
naturally allows for the estimation of near-local optimization functions, 
a consideration that ensures the scalability of the routine.

Our work is divided as follows. In \cref{sec:RC}, we demonstrate how --  
in the presence of Markovian (possibly gate-dependent) errors --
randomized compiling yields effective circuits that can be expressed as a sequence of
effective dressed cycles (\cref{lem:eff}), each with Pauli stochastic error profiles. In \cref{sec:cer}, 
we present our Cycle Error Reconstruction (CER) protocol, and provide
reconstructed error profiles for some of the IBM-Q 5-qubit devices. In \cref{sec:sc}, we introduce our Stochastic Calibration (SC) protocol, and demonstrate a $5$-fold reduction of targeted errors 
in the \texttt{ibmq\_burlington} QPU.

The data displayed in this paper was obtained in 2020 and constitutes one of the first CER+SC data sets.
To understand the time gap between the data acquisition and the paper, first know that 
CER -- which was then referred to as K-body Noise Reconstruction (KNR) -- was conceptualized in 2019, simultaneously as the Pauli Infidelity Estimation (PIE) subroutine described in \cite{FW2019,HWFW2019}. In 2020, CER was patented \cite{knrpatent2018} and incorporated as one of the True-Q software \cite{trueq} error diagnostic tools, together with SC.
The inner-workings of CER were
publicly accessible through the patent description, but naturally lacked the depth and accessibility of an academic paper. As such, the True-Q users that gathered CER diagnostics (such as in \cite{Hashim_2021}, figure 2) naturally resorted to citing the technical papers \cite{FW2019,HWFW2019,Erhard2019}, which are closely related to CER. The purpose of the following paper is therefore to supplement the scientific literature with a rigorous yet pedagogical layout of the inner-workings and motivations behind CER and SC.

\begin{figure}
\centering
\tikzexternaldisable
\begin{tikzpicture}[node distance=2cm]
\tikzstyle{box} = [rectangle, rounded corners, text width=4.1cm,
                    minimum width = 4.5cm, draw=black, align=flush left] 
\tikzstyle{long_box} = [rectangle, 
                    minimum width = 14cm, draw=black, align=flush left]

\tikzstyle{arrow} = [thick, ->, >=stealth]

 \node (diagnostics) [ rectangle,anchor=north ,fill=black!5, text width=\textwidth, minimum height=1.cm] {};
\node (cycle) [box, fill=white!100,text width= 4.5 cm, below right = 0.1cm and -2.25cm of diagnostics.north ] {\vspace{0.cm}\\ \textbf{CB-structured circuits}   \vspace{0.0cm} };

\node (protocols) [ rectangle,fill=yellow!5, text width=\textwidth, minimum height=3.1cm,below=0cm of diagnostics.south west, anchor=north west, minimum height=4.1cm] {};
\node (protocols_txt) [ below right = 0.1cm and 0cm of protocols.north west] {\textbf{Protocols:}};

\node (cb) [box, fill=red!20, below right = 0.6cm and 0.8cm of protocols.north west] {\vspace{0.cm}\\ {\textbf{Cycle Benchmarking (CB)}  \cite{Erhard2019}: Takes an effective dressed cycle of interest, and outputs its process fidelity.} \vspace{0.0cm}};
\node (cer) [box, fill=blue!20, below right = 0.0cm and .1cm of cb.north east] { \vspace{0.cm}\\ \textbf{Cycle Error Reconstruction (CER)}: Takes an effective dressed cycle of interest and 
outputs a specified marginal error distribution (see \cref{sec:cer}).
\vspace{0.cm}};
\node (sc) [box, fill=green!20, below right = 0.0cm and .1cm of cer.north east] {\vspace{0.cm}\\ \textbf{Stochastic Calibration (SC)} : Takes an effective dressed cycle of interest and a set of calibration parameters and outputs calibrated parameters.
(see \cref{sec:sc}). \vspace{0.0cm} };

\node (queries) [ rectangle,fill=black!5, text width=\textwidth, minimum height=2.8cm,below=0cm of protocols.south west, anchor=north west] {};
\node (queries_text) [below right = 0.1cm and 0cm of queries.north west] {\textbf{Queries:}};

\node (cb_query) [box, fill=red!20, below right = 0.6cm and 0.8cm of queries.north west] {\vspace{0.cm}\\ Uniformly random Pauli queries. \vspace{0.1cm}};
\node (cer_query) [box, fill=blue!20, below right = 0.0cm and .1cm of cb_query.north east] { \vspace{0.cm}\\ 
Strategically designed Pauli queries based on the targeted marginal error distribution to learn.
\vspace{0.1cm}};
\node (sc_query) [box, fill=green!20, below right = 0.0cm and .1cm of cer_query.north east] {\vspace{0.cm}\\  Strategically designed Pauli queries based on the targeted optimization landscape to learn. \vspace{0.1cm} };

\node (pier) [ rectangle,fill=yellow!5, text width=\textwidth, minimum height=7cm,below=0cm of queries.south west, anchor=north west] {};
\node (pier_text) [align=left,below right = 0.1cm and 0cm of pier.north west] {\textbf{PIE:}};

\node (query_compression) [box, fill=white!100, below= 0.8cm of cer_query.south, text width= 13.2 cm] {\vspace{0.cm}\\ \textbf{SPAM strategy :} Translate from a list of Pauli queries to a set of state preparation and measurement (SPAM) bases.
Note that a large amount of queries can be obtained from a single SPAM basis. More precisely, a fixed  stabilizer state preparation and measurement strategy on $n$ qubits
encodes $2^n$ distinct Pauli queries \cite{Erhard2019,FW2019,HWFW2019,HFW2020,FlammiaPopRec, helsen2021estimating}.
\vspace{0.0cm}};

\node (pier_core) [box, fill=white!100, below= 0.3cm of query_compression.south, text width= 13.2 cm] {\vspace{0.cm}\\ \textbf{Run CB-structured quantum circuits:} For each SPAM basis, sequentially apply the dressed cycle of interest for different sequence lengths and various random 
dressings \cite{Erhard2019}.
\vspace{0.0cm}};

\node (answer) [box, fill=white!100, below= 0.3cm of pier_core.south, text width= 13.2 cm] {\vspace{0.cm}\\ \textbf{Answer the queries:} Using post-processing methods described in e.g. \cite{Helsen2018,Erhard2019,FW2019,helsen2021estimating}, retrieve, with relative precision, the Pauli infidelity (or the orbital average, see \cref{eq:f_orbital}) for each Pauli in the list of queries.
\vspace{0.0cm}};

\node (post_processing) [ rectangle,fill=black!5, text width=\textwidth, minimum height=4.2cm,below=0cm of pier.south west, anchor=north west] {};
\node (post_processing_text) [align=left,below right = 0.1cm and 0cm of post_processing.north west] {\textbf{Analysing} \\ \textbf{results:}};

\node (cb_results) [box, fill=red!20, below right = 1.1cm and 0.8cm of post_processing.north west] {\vspace{0.cm}\\  Take the sample average of the estimated fidelities as the estimator for the process infidelity of the effective dressed cycle of interest. \vspace{0.0cm}};
\node (cer_results) [box, fill=blue!20, below right = 0.0cm and .1cm of cb_results.north east] { \vspace{0.cm}\\ With the estimated fidelities, use \cref{lem:mufi} to reconstruct the desired marginal error distribution (such as in \cref{fig:Burlington}).
\vspace{0.cm}};
\node (sc_results) [box, fill=green!20, below right = 0.0cm and .1cm of cer_results.north east] {\vspace{0.cm}\\ 
With the estimated fidelities, reconstruct the desired optimization function, and find the calibrated parameters (see \cref{fig:sc_of}). \vspace{0.0cm} };

\draw [arrow] (cycle) to [out=-170,in=90] (cb);
\draw [arrow] (cycle) -- (cer);
\draw [arrow] (cycle) to [out=-10,in=90] (sc);

\draw [arrow] (cb) -- (cb_query);
\draw [arrow] (cer) -- (cer_query);
\draw [arrow] (sc) -- (sc_query); 

\draw [arrow] (cb_query)  to [out=-90,in=160] (query_compression);
\draw [arrow] (cer_query) -- (query_compression);
\draw [arrow] (sc_query) -- (query_compression);

\draw [arrow] (query_compression) -- (pier_core); 
\draw [arrow] (pier_core) -- (answer); 

\draw [arrow] (answer) to [out=-170,in=90] (cb_results);
\draw [arrow] (answer) -- (cer_results);
\draw [arrow] (answer) to [out=-10,in=90] (sc_results);
 
\end{tikzpicture}
\tikzexternalenable
\caption{Cycle Benchmarking (CB), Cycle Error Reconstruction (CER) and Stochastic Calibration (SC) can all be thought of as different query strategies to be sent to the Pauli Infidelity Estimation (PIE) oracle. Of course, each strategy ultimately yields different information.}
\label{fig:pie}
\end{figure}
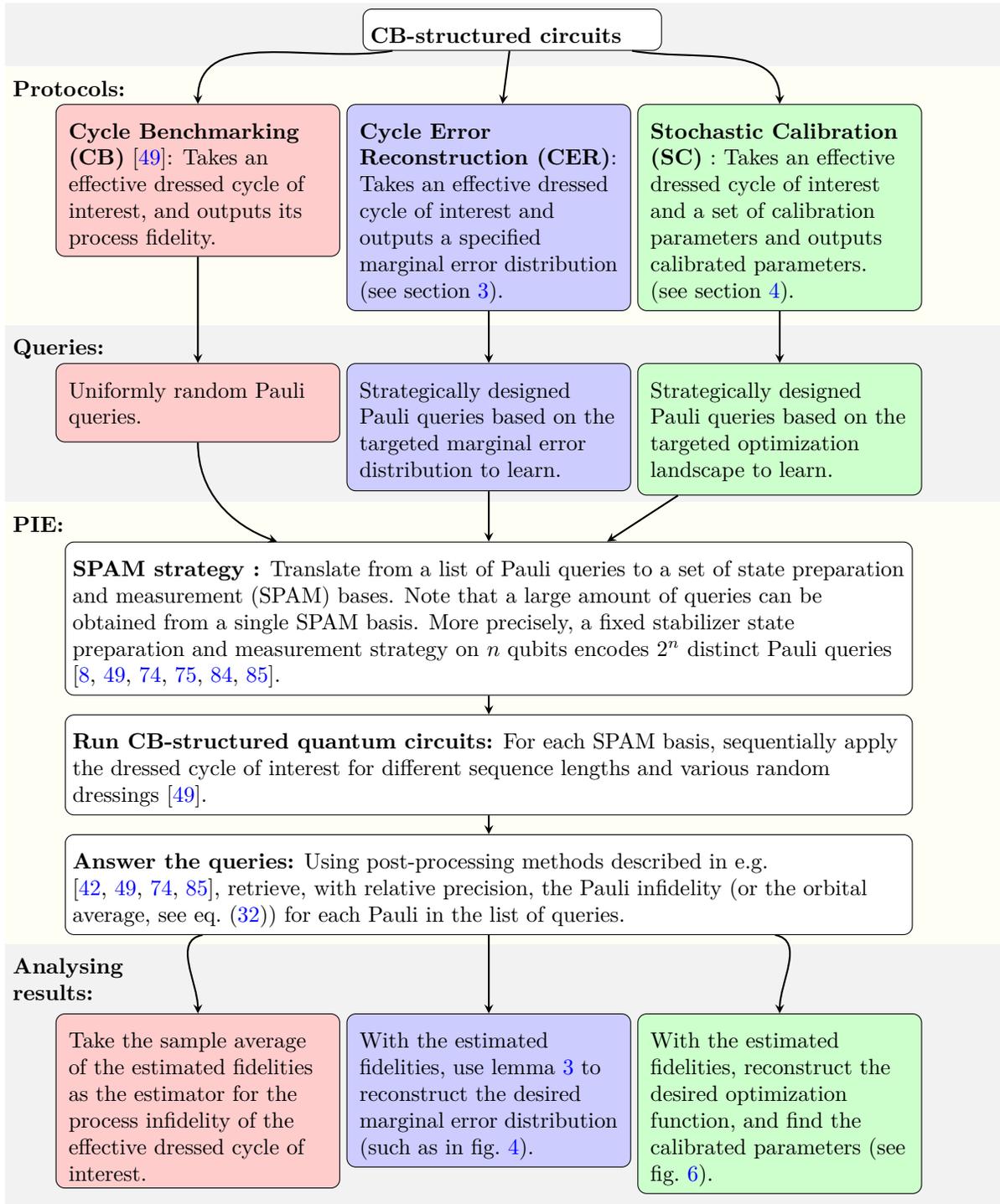

\section{Introductory Material}\label{sec:RC}

\subsection{Cycles and Dressed Cycles}\label{subsec:cycles}

\begin{definition}[Cycle]\label{def:cycle} 
A cycle is a scheduled sequence of (possibly simultaneous) native instructions applied over an arbitrarily large fraction of the QPU.
\end{definition}

The above is an intentionally broad definition of a cycle; the idea behind the concept of cycle is to consider a process that remains close to elementary yet integrated operations (as such, cycles are meant to be short time steps) and to capture important background information (such as the precise scheduling of instructions over the QPU) to properly contextualize error profiles.

A quantum circuit can be expressed as a sequence of cycles, each containing a set of gate operations on the system. Examples of cycles include simultaneous single-qubit rotations, or a set of entangling gates applied to disjoint qubits simultaneously. Some cycles are much more error-prone than others. 
Most of the time, there is an easily-identifiable dichotomy amongst cycles based on their error profiles; using this dichotomy, we categorize cycles as either ``hard" (cycles with gates that are more error-prone) or ``easy" (cycles with gates that are less error-prone). Hard cycles typically consist of entangling gates applied to disjoint qubits simultaneously, while easy cycles commonly consist of parallel single-qubit operations. Every circuit can be decomposed as an alternating sequence of easy and hard cycles so that a circuit is expressed as E-H-E… for easy cycles E and hard cycles H.
There are numerous globally equivalent circuits with identically-scheduled hard cycles, but different easy cycles. It is often advantageous to randomize over equivalent combinations of easy cycles to suppress the effect of some errors. 
This strategy -- known as {randomized compiling} (RC)\cite{Wallman2016} -- is advantageous to state-of-the-art circuit compilers.
In the paradigm where circuits are constructed with randomized compiling, it is only natural to consider every sequential ``easy-hard'' pair of cycles as a cycle itself, which we refer to as a 
\emph{dressed cycle} (see \cref{fig:rc}  {\bf{a}}):

\begin{definition}[Dressed Cycle]\label{def:dressedCycle}
A dressed cycle is a sequential pair of easy and hard cycles. The easy cycle 
may be random, but the hard cycle is fixed.
\end{definition}

Think of a hard cycle 
as the immutable core of a dressed cycle, and the easy cycle as its variable dressing.
Examples include the dressed CNOT on a 2-qubit processor, which may consist 
of a CNOT followed (or preceded) by a random cycle of parallel single-qubit gates.
Every quantum circuit can be expressed as a sequence of dressed cycles.
In fact, there are a number of ways of selecting the easy cycle (or dressing) on each dressed cycle such that the effective circuit is equivalent to the original circuit.
One very effective way to leverage this is to randomize over the equivalent versions of the dressed cycles in a circuit, as is done in RC \cite{Wallman2016} (see \cref{fig:rc} {\bf b}).

\subsection{Noisy and fiducial cycles}

The ideal implementation of a cycle is described by a unitary operator $U \in \bb{U}(2^n)$, and an implementation of $U$ is described by a completely positive and trace-preserving map $\nrep{U}$, that is, $\nrep{U}$ is a linear map (sometimes referred to as a superoperator) with a Kraus operator decomposition
\begin{align}
    \nrep{U}(\rho) = \sum_j K_j \rho K_j^\dag
\end{align}
for some Kraus operators $K_j$. The letter $\nu$ simply stands for ``noisy''.
The ideal implementation of $U$ is the channel $\phi(U)$ defined by
\begin{align}
    \rep{U}(\rho) = U \rho U^\dag.
\end{align}
The letter $\phi$ stands for ``fiducial'', and is also a standard choice in representation theory.

A standard quantum circuit consists of preparing each qubit in an initial state (e.g., the ground state), applying a sequence of cycles $\nrep{U_1}, \ldots, \nrep{U_m}$, and measuring each qubit in a fixed basis (e.g., the energy eigenbasis).
In this paper, we are primarily interested in characterizing the errors in the implementations of cycles rather than errors in state preparation and measurement (SPAM) procedures and so we refer to the ``circuit'' as the sequence of cycles and ignore additional circuit features such as mid-circuit measurements.

\subsection{Randomized Compiling and Effective Dressed Cycles}\label{subsec:effectivedressedcycles}

As mentioned earlier, randomized compiling is based on the recognition that some cycles (particularly, independent single-qubit gates) are relatively ``easy'' to implement with small errors and so we can randomize the easy gates to suppress errors in the ``hard'' gates.
To apply randomized compiling to a circuit, we express the ideal circuit as
\begin{align}\label{eq:riffled}
\rep{E_m} \rep{H_{m-1}} \rep{E_{m-1}} \ldots \rep{H_1} \rep{E_1} \rep{H_0} \rep{ E_0}
\end{align}
where the $E_i$ and $H_i$ are easy and hard cycles respectively and we write the circuit in operator order. Notice that the circuit starts and ends with easy cycles.
The implementation of the bare circuit is denoted as (see \cref{fig:rc} {\bf{a}})
\begin{align}
    \nrep{E_m} \nrep{H_{m-1}} \nrep{E_{m-1}} \ldots \nrep{H_0} \nrep{E_0}.
\end{align}
For a circuit in the form of \cref{eq:riffled}, we can ``twirl'' the 
$i$th hard cycle $H_i$ by compiling twirling gates from some set \bb{T} 
into the adjacent easy gates. 
Specifically, we compile a randomly chosen twirling operation 
$T_i \in \bb{T}$ into $E_{i+1}$ and a correction gate 
$T_i^c := H_i^\dagger T_i^\dag H_i$ into $E_{i}$. Note that we are implicitly assuming that $T_i^c$ is an easy gate for all $i$, that is, that $H_i^\dagger \bb{T} H_i$ is a set of easy gates. In theory, we may 
allow the twirling set $\set T$ to change for each cycle. In practice, a unique twirling set $\set T$ often suffices. This common simplification stems from the common classification of 
parallel local operations as easy cycles.
We define the $i$th randomized dressed cycle to be (see \cref{fig:rc} {\bf{b}})
\begin{align}
    \nrep{H_i} \nrep{ T_i^c E_i T_{i-1}}~,
\end{align}
where $T_{-1}=T_m=I$.
A randomized dressed cycle inherits a stochastic nature, and generally greatly differs from the deterministic dressed cycle $\nrep{H_i} \nrep{E_i}$. For this reason, instead of contrasting $\nrep{H_i} \nrep{E_i}$ with $\nrep{H_i} \nrep{T_i^c E_i T_{i-1}}$, it is more sensible (from a perturbation analysis standpoint) to compare $\nrep{H_i} \nrep{E_i}$ with {a cycle which, in the noiseless case, applies the same mapping. For this reason, we introduce}
\emph{adjusted dressed cycles}:
\begin{definition}[Adjusted dressed cycles]\label{def:adc}
We define an adjusted dressed cycle as:
\begin{align}
    \adj(H_i,E_i|T_i, T_{i-1}):= \phi(T_i)\nrep{H_i} \nrep{T_i^c E_i T_{i-1}}\phi^\dagger(T_{i-1})~.
\end{align}
The adjusted dressed cycle
retains the stochastic nature of the dressed cycle since it depends on the random twirling operations $T_i$, $T_{i-1}$, and in the 
limit where cycles are ideally implemented, we retrieve $\adj(H_i,E_i| T_i, T_{i-1})=\phi(H_iE_i)$. Let's extend the notation to account for the beginning and end of the circuit, and define $\adj(E_m| T_{m-1}):= \nrep{E_m T_{m-1}}\phi^\dagger(T_{m-1})$ and $\adj(H_0,E_0| T_{0}):= \phi(T_0)\nrep{H_0} \nrep{T_0^c E_0}$.
\end{definition}
With this notation, an instance of a randomly compiled (RC) circuit applied to a QPU (\cref{fig:rc} {\bf{b}}),
\begin{align}
   \mc C_{\rm RC}(\vec{T}):=\nrep{E_m T_{m-1}} \nrep{H_{m-1}} \nrep{T_{m-1}^c E_{m-1} T_{m-2}} \cdots \nrep{H_{1}} \nrep{T_{1}^c E_{1} T_0} \nrep{H_{0}} \nrep{T_{0}^c E_{0}}~,
\end{align}
can be re-expressed as a sequence of adjusted dressed cycles (\cref{fig:rc} {\bf{c}})
\begin{align}
   \mc C_{\rm RC}(\vec{T})=\adj(E_m| T_{m-1}) \adj(H_{m-1},E_{m-1}| T_{m-1}, T_{m-2}) \cdots\adj(H_1,E_1| T_1,T_0) \adj(H_0,E_0| T_0)~,
\end{align}
where $\vec{T} := (T_{m-1}, \cdots, T_0) \in \set{T}^m$.
The point of randomized compiling is to take the average over many such circuits in order to approximate
the sample average 
\begin{align}\label{eq:avg_rc}
    \left\langle \mc C_{\rm RC}(\vec{T}) \right\rangle_{\vec{T}}:= \frac{1}{|\set T|^m} \sum_{ \vec{T} \in \set T^m} \mc C_{\rm RC}(\vec{T})~.
\end{align}
\Cref{eq:avg_rc} can be generalized as an integral over compact groups using the natural Haar measure. We are now ready to define 
\emph{effective dressed cycle}:
\begin{definition}[Effective dressed cycles]\label{def:edc}
We define an effective dressed cycle as the average 
implementation of the adjusted dressed cycle:
\begin{align}\label{eq:avg-dressed}
    \eff(H_i,E_i)& :=
    \left\langle \adj(H_i,E_i| T_i, T_{i-1}) \right\rangle_{T_i, T_{i-1}}
    \notag \\
    & = \frac{1}{|\set T|^2} \sum_{T_{i} \in \set T}  \sum_{T_{i-1} \in \set T} \adj(H_i,E_i|T_i, T_{i-1})~.
\end{align}
\end{definition}
Let also $\eff(E_m) :=
\left\langle \adj(E_m| T_{m-1}) \right\rangle_{T_{m-1}}$ and $\eff(H_0, E_0) :=
\left\langle \adj(H_0, E_0| T_{0}) \right\rangle_{T_{0}}$. For $i \notin \{0,m\}$, as $\adj(H_i,E_i| T_i, T_{i-1})$ 
and $\adj(H_{i+1},E_{i+1}| T_{i+1}, T_i)$ both depend on $T_i$, the average overall operation applied to a QPU will not \emph{exactly} factorize, that is, in general we have
\begin{align}
    \left\langle \mc C_{\rm RC}(\vec{T}) \right\rangle_{\vec{T}} \neq \eff( E_m) \eff(H_{m-1},E_{m-1}) \ldots 
   \eff(H_1,E_1)
   \eff(H_0,E_0)~.
\end{align}

Nevertheless, as we show in \cref{app:rc}, the expected circuit factorizes as a sequence of effective dressed cycles to a great
approximation, as expressed in the following lemma  (see \cref{fig:rc} {\bf{d}}):
\begin{lemma}\label{lem:eff}
The average RC circuit approximately factorizes as a sequence of effective dressed cycles,
        \begin{align}\label{eq:avg_circuit}
        \left\langle \mc C_{\rm RC}(\vec{T}) \right\rangle_{\vec{T}}  & \approx 
        \eff( E_m) \eff(H_{m-1},E_{m-1}) \ldots 
       \eff(H_1,E_1) \eff(H_0,E_0) 
    \end{align}
where the corrections are second order in the deviations $\adj(H_i ,E_i|T_i, T_{i-1}) -\eff(H_i,E_i) $ and linear in the number of cycles.
\end{lemma}

\begin{figure}[h!]
\centering
\def\svgwidth{\columnwidth}
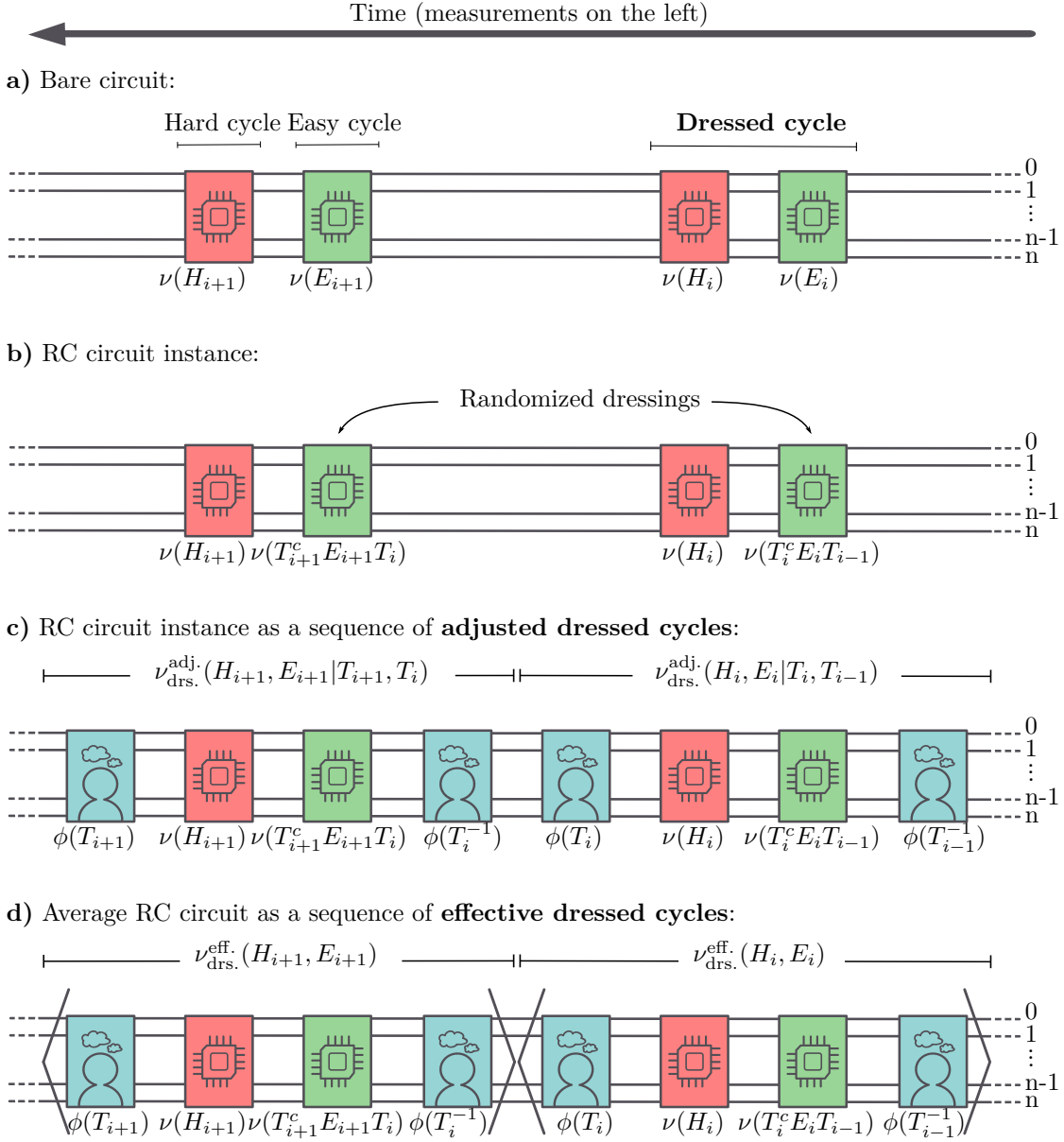
\caption{Circuit representations expressed in terms of cycles. Notice that the time flows from the right to the left to accommodate the matrix notation below each cycle. {\bf a)} An application circuit can always be expressed as an alternating sequence of easy and hard cycles or, equivalently, as
a sequence of dressed cycles. {\bf b)} Randomized compiling (RC) consists of sampling the 
easy cycle dressings in a way that doesn't alter the overall circuit \cite{Wallman2016}. 
{\bf c)} Dressed cycles can be re-expressed with virtual adjustments. {\bf{d)} } As shown in \cref{lem:eff}, 
the average RC circuit can be closely approximated by a sequence of effective dressed cycles, even if the easy cycles are not perfect, and even if errors are gate-dependent.
}
\label{fig:rc}
\end{figure}

\subsection{The Error Profile of Effective Dressed Cycles}

The above analysis holds for arbitrary twirling sets, including $n$-fold Weyl-Heisenberg groups, meaning that RC, CB and CER can be performed on \emph{qudit-based} systems \cite{Goss2022}. For the current paper, however, we are primarily interested in Pauli twirling, wherein $\bb{T} = \bb{P}_n = \{I, X, Y, Z\}^{\otimes n}$.
When the hard cycles are multi-qubit Clifford gates, $H_i^\dag \bb{P}_n H_i = \bb{P}_n$ by definition and so the correction gates can be compiled into the easy gates, which consist of tensor products of arbitrary single-qubit gates, as required.

In \cref{app:twirled_error}, we demonstrate the following lemma:

\begin{lemma}\label{lem:pauli_channel}
Let the effective dressed cycle be a composition of the form
\begin{align}\label{eq:noise_stoch}
    \eff(H_i,E_i) = \phi(H_i E_i) \mc S_i~.
\end{align}
Then, $\mc S_i$ is a Pauli stochastic channel of the form
\begin{align}\label{eq:channel_pauli_sum}
    \mc{S}_i(\rho) = \sum_{P \in \bb{P}_n} p_i(P) P \rho P^\dag
\end{align}
for some probability distribution $p_i$ if either of these two conditions hold:
\begin{enumerate}
    \item If the cycles $H_i$ and $E_i$ are Cliffords;
    \item If the easy cycles have a fixed error, i.e. $\nrep{E_i}= \mc E \circ \rep{E_i}$.
\end{enumerate}
\end{lemma}

The second condition implies that \cref{eq:channel_pauli_sum} always holds to order $\nrep{E_i}\phi(E_i^{-1}) - \langle\nrep{E_i}\phi(E_i^{-1}) \rangle_{E_i}$, which is small from the assumption that easy gates have a high fidelity.
With a bit more consideration in the compiling step, we can make \cref{eq:channel_pauli_sum} hold exactly in all cases (see \cref{app:twirled_error}). The idea is to treat parallel rounds of $\pi/8$ phase rotations as hard cycles
(even though they aren't typically error \-prone). The twirling group for such cycle is slightly modified 
to be the group generated by local Paulis and 
$H_i^2$. That is, in the instance where $H_i=Z_{\pi/8}^{\otimes n}$ and $E_i=I$, then the implementation will look like
\begin{align}
    \nrep{T_{i+1}^cE_{i+1} T_{i}}\nrep{Z_{\pi/8}^{\otimes n}}\nrep{T_i^c T_{i-1}}~,
\end{align}
where $T_i \in \set D_4^{\otimes n}$ ($\set D_4$ is the dihedral group 
generated by Paulis and $Z_{\pi/4}$) and $T_{i}^c$ is still defined as 
$H_i^\dag T_i H_i$.
From \cref{eq:channel_pauli_sum}, the implemented circuit averaged over all 
randomizations is
\begin{align} \label{eq:RC_avg}
    \left\langle \mc C_{\rm RC}(\vec{T}) \right\rangle_{\vec{T}} \approx  \phi(E_m)\mc S_{m} \phi(H_{m-1}E_{m-1}) \mc S_{m-1} \cdots \phi(H_1E_1) \mc S_1 \phi(H_0 E_0) \mc S_0~.
\end{align}

\section{Learning Cycle Errors}\label{sec:cer}
In the previous section, we have seen that by structuring circuits into serial dressed cycles, and by 
allowing easy cycles to be appropriately randomized, we effectively obtain ideal cycles $\phi(H_i E_i)$ interleaved with 
stochastic errors $\mc{S}_i$. The general form of stochastic errors depends on the randomizing sets, but 
when twirling over $n$-fold Paulis $\set P_n$, $\mc{S}_i$ becomes a probabilistic sum over Pauli errors, as depicted in \cref{eq:channel_pauli_sum}. In this section, our goal is to learn the most likely errors appearing in the probability distribution $p_i$. In the next sections, we omit the $i$ indexing for the sake of clarity. 

In full generality, the error probability distribution $p_i$ has a domain of $4^n$ distinct Paulis (including the identity). 
In practice, however, it is often the case that a large proportion 
of error scenarios are so improbable that they can be neglected.
In such instances, we say that the errors can be well described via 
a \emph{reduced model}, which in our case corresponds to a set of constraints over the error probability distribution.

A simple example of a reduced model would be to
restrict the distribution to single-qubit errors and errors that are shared between neighboring qubit pairs. 
In a planar architecture, this nearest-neighbors model would reduce the number of error scenarios to $O(n)$. 
As we will see, reduced models motivate the 
recourse to marginal probabilities, which are defined in the next subsection. 
Characterizing quantum error channels through marginals, as proposed in e.g.
\cite{Emerson2007}, is an efficient way to avoid the exponential scaling of 
fine-grained error profiles.

\subsection{Marginal Pauli Distributions}

Before defining marginal probabilities over sets of Paulis, 
we introduce some elementary notation.
Let $\set A$ be a set of qubit indices (e.g. $\{0,2\}$),
and let $\set A^{c}$ be the complement of $\set A$ with respect to $\set Z_n$ (i.e. $\set Z_n - \set A$). Let
$\set P_{\set A} \subseteq \set P_n$ refer to the subgroup of $4^{|\set A|}$ Paulis which act as the identity on the 
qubits indexed by $\set A^c$. The instance $P_{\{i\}}$ is simplified to the more
common notation $P_i$. 
For a Pauli $P_{\set A} \in \set P_{\set A}$, we
say that $\set A$ is the support of $P_{\set A}$.
The weight of a Pauli error $w(P)$ is the number of non-identity Paulis appearing in $P$. The weight is upper-bounded by the size of the support: $w(P_{\set A}) \leq |\set A|$.
Every Pauli $P \in \set P_n$ can be expressed as a product $P=P_{\set A} P_{\set A^c}$ where
$P_{\set A} \in \set P_{\set A}$, $P_{\set A^c} \in \set P_{\set A^c}$.

\begin{definition}[Marginal Pauli Probability]
We define the marginal probability 
$\mu(P_{\set A})$ as
\begin{align}\label{eq:marginals}
    \mu(P_{\set A}):= \sum_{Q_{\set A^c} \in \set P_{\set A^c}} p(P_{\set A} Q_{\set A^c})~,
\end{align}
which defines a distribution over the set $\set P_{\set A}$. Accordingly, we refer to $\set A$ as the qubit support of the marginal distribution.
\end{definition}
For instance, with $\set A=\{0,2\}$ and $IX_{\{0,2\}} \in \set P_{\{0,2\}}$ then 
\begin{align}
    \mu(IX_{\{0,2\}})= \sum_{\substack{Q_i \in \{I,X,Y,Z\}\\ i \notin \{0,2\}}} p(I_0 \otimes Q_1 \otimes X_2 \otimes Q_3 \cdots \otimes Q_{n-1})~.
\end{align}
We use $\mu(I)$ as a shorthand
for $\mu(I \cdots I_{\set Z_n})$; $\mu(I)$ is commonly known as the average process fidelity, which is generally defined as follows.

\begin{definition}[Average Process Fidelity]
Given an error channel $\mc E$ with Kraus decomposition
\begin{align}
    \mc E(\rho) = \sum_i K_i \rho K_i^\dag~,
\end{align}
then the average process fidelity of $\mc E$ is defined 
as \cite{Dugas2019unitarity}
\begin{align}
    f_P(\mc E) := \sum_i \frac{\left|\tr K_i \right|^2}{d^2}~.
\end{align}
In the instance where $\mc E$ is a Pauli stochastic channel, 
$f_P(\mc E)= \mu(I)$.
\end{definition}

The conceptual advantage of reconstructing marginals rather than the full distribution $p$ is that they naturally offer coarse-grained pictures of $p$.
Notice that $p(P_{\set A})$ can be interpreted as a special case of a marginal $\mu(P_{\set A})$ when $\set A=\set Z_n$.

If we are starting from a reduced error model, it is often straightforward to use a few
marginal probability distributions to obtain a global error probability distribution 
$p$ over $\set P_n$. Indeed, typical reduced error models -- such as the ones assuming low-weight errors (e.g. $w(P) \leq 2$)
or nearest-neighbour interactions -- involve strong constraints on \cref{eq:marginals}. That is, many probabilities do not contribute to the sum in \cref{eq:marginals}.
Obtaining probabilities from marginals is then a matter of inverting a simplified system of linear equations. 
For instance, if we assumed an error model with Pauli weights of at most $2$, we would get, for $i\neq j $, $P,Q \in \{X,Y,Z\}$,
\begin{align}
        p(PQ_{\{i,j\}})&= \mu(PQ_{\{i,j\}})~; \label{eq:pmu1} \\
        p(P_{i})&= \mu(P_{i}) - \sum_{\substack{k \neq i\\ Q \in \{X,Y,Z\}}} \mu(PQ_{\{i,k\}})~.\label{eq:pmu2}
\end{align}
Perhaps more importantly, reduced error models can be deduced or validated from 
marginal distributions. For instance, if we observed that the marginals of weight-$3$ Paulis were 
insignificant for all sets $\set A$ of $3$ qubits,
then we could safely deduce that significant errors are at most weight-2. More formally, let 
\begin{align}
    \mu(P_{\set A}) < \epsilon, ~~~~~\forall P_{\set A} , ~\set A: w(P_{\set A})=|\set A|=3~,
\end{align}
so that all Pauli errors with weight 3 have marginals less than $\epsilon$ for some $\epsilon$ negligible compared to $1-\mu(I)$. Then, we get {bounds on the probabilities of errors with weight 2 and an expression for the probabilities of weight 1 errors}, analogously to 
\cref{eq:pmu1,eq:pmu2},
\begin{align}
        \mu(PQ_{\{i,j\}})-3 \epsilon < p(PQ_{\{i,j\}})&\leq \mu(PQ_{\{i,j\}})~; \\
        p(P_{i})&= \mu(P_{i}) - \sum_{\substack{k \neq i\\ Q \in \{X,Y,Z\}}} \mu(PQ_{\{i,k\}})~.
\end{align}
{That is, if $\epsilon$ is negligible, then we can say that $p(A)\approx \mu(A)$ for $\{A\in \mathbb{P} : w(A)\leq 2\}$}.

\subsection{H-orbits}
We've defined marginal distributions $\mu$ over the sets $\set P_{\set A}$. Now, we define $H$-orbits:
\begin{definition}[$H$-orbits]
Let $H$ be a hard cycle such that 
$H \set P_{\set A} H^{-1}= \set P_{\set A}$. We define $H$-orbits and sets of orbits as:
\begin{align}\label{def:horbit}
    \orbit{P_{\set A}}{H}&:= \{H^{j} P_{\set A} H^{-j}|~ j \in \set N\}~, \\
    \orbit{\set P_{\set A}}{H}&:= \{\orbit{P_{\set A}}{H}| ~P_{\set A}\in \set P_{\set A}\}~, 
\end{align}
\end{definition}

 For example, if $H$ acts as a phase gate $\sqrt{Z_2}$ on qubit $\{2\}$ then, 
$\orbit{X_2}{H}=\orbit{Y_2}{H}=\{X_2,Y_2\}$, $\orbit{Z_2}{H}=\{Z_2\}$, and $\orbit{\set P_{2}}{H}=\{\{I_2\},\{X_2,Y_2\},\{Z_2\} \}$. 
From these definitions, we can naturally define a marginal distribution over $\orbit{\set P_{\set A}}{H}$ through
\begin{align}
    \mu(\orbit{P_{\set A}}{H}):= \sum_{Q_{\set A} \in \orbit{P_{\set A}}{H}} \mu(Q_{\set A})~.
\end{align}
In our previous example, the marginal distribution 
over $\orbit{\set P_{2}}{\sqrt{Z_2}}$ would contain 3 distinct marginal probabilities, namely $\mu(I_2)$, $\mu(Z_2)$
 and $\mu(\orbit{X_2}{\sqrt{Z_2}})=\mu(X_2,Y_2):= \mu(X_2)+\mu(Y_2)$. 

\subsection{Estimating Pauli Fidelities}\label{sec:fidelity}

As we emphasize in the present work, an added benefit of recasting a circuit in terms of effective dressed cycles is that since the error channel is tailored into a stochastic channel, these cycles 
can be individually and efficiently characterized via a subroutine which we refer to as Pauli Infidelity Estimation (PIE) 
\cite{FW2019,HWFW2019}. 
Since PIE is well detailed in \cite{FW2019}, we treat it here as an 
oracle (see \cref{proto:pie} for more details): given a set of queries regarding an effective dressed cycle of interest (the queries take
the form of a set of Pauli observables $\set S = \{P_1, P_2, \cdots\}$ ) PIE yields an answer in the form of a set of corresponding Pauli fidelities $f(P)$, which are defined through:
\begin{definition}[Pauli fidelities] Given a channel $\mc S$, the fidelity of $P$ is defined as
\begin{align}
    f(P):= \inner{P}{\mc{S}(P)}~,
\end{align}
where we use the normalized Hilbert-Schmidt inner product
\begin{align}
\langle A, B \rangle = 2^{-n} \tr A^\dag B.
\end{align}
\end{definition}
In the most general scenario (i.e. when no assumptions are made on the SPAM errors), some of the Pauli fidelities cannot be individually estimated through the RB-like protocols described in \cite{FW2019}; to be more precise, given a hard gate $H$, only the following geometric orbital averages can be estimated:
\begin{align}\label{eq:f_orbital}
    \fidorbit{P}{H}_{\rm geo}:=\left[\Pi_{R \in \orbit{Q}{H}} f(R)\right]^{\frac{1}{|\orbit{Q}{H}|}}~,
\end{align}
where $\orbit{P}{H}$ is a set of Paulis called the $H$-orbit of $P$, and is defined in \cref{def:horbit}.
For simplicity, we resort to the more intuitive arithmetic orbital averages,
\begin{align}\label{eq:f_orbital}
    \fidorbit{P}{H}:=\frac{1}{|\orbit{P}{H}|}\sum_{Q \in \orbit{P}{H}} f(Q)~,
\end{align}
the reasoning being that the geometric an arithmetic means are equivalent up to order ${\rm Var}_{R \in \orbit{Q}{H}} \left(f(R)\right)$, and that this discrepancy is often much smaller than the precision on the estimates in practice. 
The number of terms in the RHS of \Cref{eq:f_orbital} is usually small. 
Indeed, in most typical cases hard cycles have a small cyclicity and thus generate small orbits. For example, 
if {the hard gate} $H={\rm CZ}$ on the qubit pair $\{1,2\}$, then $\orbit{XX_{\{1,2\}}}{H}=\{XX_{\{1,2\}},YY_{\{1,2\}}\}$ and
\begin{align}
    f(\orbit{XX_{\{1,2\}}}{H})=f(XX_{\{1,2\}},YY_{\{1,2\}})=\frac{1}{2}\Big[f(XX_{\{1,2\}})+f(YY_{\{1,2\}})\Big]~.
\end{align}
Gates such as $\rm CX$, $\rm CZ$, $\rm SWAP$, $\sqrt{Z}$ , as well as cycles containing combinations of these gates applied in parallel, all generate orbits of at most $2$ elements. 
We are now ready to define a Pauli fidelity estimation oracle:

\begin{definition}[PIE oracle]\label{def:pie}
A Pauli fidelity estimation (PIE) oracle accepts a hard cycle $H$ and a subset of Paulis $\set S \subseteq \bb{P}_n$ and returns the corresponding subset of estimated Pauli orbital fidelities, i.e.
\begin{align}
    {\rm PIE}(\set S, H)&=\{\hatfidorbit{P}{H} : P \in \set S\}~.
\end{align}
\end{definition}

We didn't write $\hatfidorbit{P}{H}_{\rm geo}$ because we consider the slight difference between 
the geometric and the arithmetic means as a negligible bias in the accuracy of the estimates. 
We briefly describe the inner workings of ${\rm PIE}(\set S, H)$ in the appendix and refer the interested reader to Ref. \cite{FW2019} for more details. We emphasize that ${\rm PIE}(\set S,H)$
belongs to the family of RB-like protocols, and hence inherits a multiplicative precision on the estimates $\hat{\mu}(P)$ and a robustness to SPAM 
errors. Through classical post-processing, many Pauli fidelities can generally be estimated with recourse to relatively few quantum circuits.

The error probability distribution $p$ can be obtained from Pauli fidelities through
\begin{align}
   p(P) &= \frac{1}{d^2}\sum_{Q \in \bb{P}_n} \chi_Q(P) f(Q)~, \label{eq:fid2prob} \\
    \chi_P(Q) &= \begin{cases} 
    1 & P, Q \mbox{ commute} \\
    -1 & \mbox{otherwise.}
    \end{cases} \label{def:chi}
\end{align}
$\chi_Q(P)$ are irreducible characters. By combining \cref{eq:marginals,eq:fid2prob}, we can see that marginal probabilities can be obtained through Pauli fidelities through a reduced sum:
\begin{align}
    \mu(P_{\set A}) 
    &= \frac{1}{d^2}\sum_{\substack{R_{\set A^c} \in \set P_{\set A^c}\\ Q \in \set P_n}} \chi_Q(P_{\set A} R_{\set A^c}) f(Q) \notag \\
    &= \frac{1}{d^2}\sum_{\substack{R_{\set A^c} \in \set P_{\set A^c}\\ Q \in \set P_n}} \chi_Q(P_{\set A}) \chi_Q(R_{\set A^c}) f(Q) \tag{$\chi_P(QR)=\chi_P(Q) \chi_P(R)$} \\
    & = \frac{1}{4^{|\set A|}} \sum_{Q_{\set A} \in \set P_{\set A}} \chi_{Q_{\set A}}( P_{\set A}) f(Q_{\set A})~, \label{eq:fid2marginal}
\end{align}
where the last line is obtained from character orthogonality relations \cite{Steinberg2012}. As a simple example, 
\begin{align}
    \mu(X_2)= \frac{1}{4} \Big[ 1 +f(X_2)-f(Y_2)-f(Z_2) \Big]~.
\end{align}
If we have recourse to ${\rm PIE}(\set P_{\set A}, H)$, we may not be able to immediately use \cref{eq:fid2marginal} 
since instead of having all the individual Pauli fidelity estimates $\{\hat{f}(Q_{\set A}): Q_{\set A} \in \set P_{\set A}\}$, we have access to the orbital averages $\{\hatfidorbit{Q_{\set A}}{H}: Q_{\set A} \in \set P_{\set A}\}$. Fortunately, there is a simple way to adapt \cref{eq:fid2marginal} accordingly:
\begin{lemma}\label{lem:mufi}
    Let $\set P_{\set A}$ be invariant under the action of $H$ (that is $H \set P_{\set A} H^{-1} = \set P_{\set A}$) and let $\orbit{\set P_{\set A}}{H}$ be the set of distinct $H$-orbits
    defined in \cref{def:horbit}. The marginal distribution over $\orbit{\set P_{\set A}}{H}$ 
    can be obtained from $\{f(\orbit{Q_{\set A}}{H}):Q_{\set A} \in \set P_{\set A}\}$ through:
    \begin{align}\label{eq:mufi}
\mu(\orbit{P_{\set A}}{H}) 
        & =\frac{|\orbit{P_{\set A}}{H}|}{4^{|\set A|}}\sum_{Q_{\set A} \in \set P_{\set A}} \chi_{Q_{\set A}}(P_{\set A}) \fidorbit{Q_{\set A}}{H}~.    
    \end{align}
\end{lemma}
\begin{proof}
   Let $c$ be the smallest positive integer such that $H^{c} \in \set P_n$. Then,
    \begin{align}
        \frac{\mu(\orbit{P_{\set A}}{H})}{|\orbit{P_{\set A}}{H}|} & = \sum_{j=0}^{c-1} \frac{\mu(H^jP_{\set A}H^{-j})}{c} \notag \\
        &=\frac{1}{4^{|\set A|}c} \sum_{j=0}^{c-1} \sum_{Q_{\set A} \in \set P_{\set A}} \chi_{Q_{\set A}}(H^{j} P_{\set A} H^{-j}) f(Q_{\set A})  \tag{\cref{eq:fid2marginal}} \\
        & = \frac{1}{4^{|\set A|}c} \sum_{j=0}^{c-1} \sum_{Q_{\set A} \in \set P_{\set A}} \chi_{H^{j} Q_{\set A} H^{-j}}(H^{j} P_{\set A} H^{-j}) f(H^{j} Q_{\set A} H^{-j})
 \tag{$H^{j} \set P_{\set A} H^{-j} = \set P_{\set A}$}    \\
 & =  \frac{1}{4^{|\set A|}}\sum_{Q_{\set A} \in \set P_{\set A}} \chi_{Q_{\set A}}(P_{\set A})   \sum_{j=0}^{c-1} \frac{f(H^{j} Q_{\set A} H^{-j})}{c} \tag{$\chi_{H^{j} Q H^{-j}}(H^{j} P_{\set A} H^{-j}) = \chi_Q(P_{\set A})$} \\
  &= \frac{1}{4^{|\set A|}} \sum_{Q_{\set A} \in \set P_{\set A}} \chi_{Q_{\set A}}(P_{\set A}) \fidorbit{Q_{\set A}}{H}, \tag{\cref{eq:f_orbital}}
 \end{align}
{where we can drop terms with $j\geq c$ in the first line because conjugating by a Pauli returns the original Pauli and going past $j=c-1$ causes terms to repeat.
Equivalently, for a less efficient version of this calculation, we could replace $c$ with $|\orbit{P_{\set A}}{H}|$.}
\end{proof}
{In this section we have shown how to convert the average fidelities $\{\fidorbit{Q_{\set A}}{H}: Q_{\set A} \in \set P_{\set A}\}$ returned by PIE into average marginal probabilities $\{\mu(\orbit{Q_{\set A}}{H}): Q_{\set A} \in \set P_{\set A}\}$. Our estimation procedure inherits the robustness to SPAM errors that characterises PIE and RB-like protocols in general. This is a crucial advantage, as it guarantees the effectiveness of our method in the presence of arbitrary (and potentially unknown) classes of SPAM errors. In practice, some information about SPAM errors is often available \cite{Lin_2021, \learnability}. 
For this reason, it is natural to ask whether this information may be used to enhance our estimation procedure. In the appendix we give a positive answer to this question. Specifically, we show that by incorporating the information about SPAM errors into PIE, it is possible to estimate the individual fidelities $\{f({Q_{\set A}}): Q_{\set A} \in \set P_{\set A}\}$ and hence the individual marginal probabilities $\{\mu({Q_{\set A}}): Q_{\set A} \in \set P_{\set A}\}$ through eq. \ref{eq:fid2marginal}.} We emphasize however that our method for resolving the orbits has additive precision rather than multiplicative, and is therefore more vulnerable to shot noise.

\subsection{Cycle Error Reconstruction}

Typical cycles of interest are composed of a parallel round of gates. Let $H= G_0  \cdots   G_{s-1}$ be a hard cycle composed of $s$ parallel gates $G_i$ with respective support $\set A_i$.  
For instance, given a 5-qubit system, $H$ may be comprised of two parallel CZ gates on the pairs $\set A_0=\{0,1\}$, $\set A_1=\{3,4\}$, and a Hadamard gate on the remaining qubit $\set A_2=\{2\}$. The main interest in defining the parallel gate supports $\set A_i$ of a cycle $H$ is that for any error $Q_{\set A_i} \in \set P_{\set A_i}$, then $H Q_{\set A_i}H^{-1}$ remains in  $\set P_{\set A_i}$.

Many important noise mechanisms such as dephasing, amplitude damping, and unitary over-rotation are fully depicted within the support of a given gate. Hence, the marginal distributions associated with individual gate supports are enough to capture dominant error sources in the absence of crosstalk effects.
In turn, crosstalk effects act on a small number of qubits, which implies that marginal distributions associated with single unions of gate supports (i.e. $\set A_i \cup \set A_j$) are often sufficient to reconstruct the most significant part of the error distribution $p$.

In general, one may reconstruct the marginal distributions for all the unions of $k$ distinct parallel gate supports:
\begin{protocol}[Cycle Error Reconstruction] \label{pro:cer}
\begin{itemize}
    \item[]
    \item[]  {\bf Input}: A hard cycle $H$; $k$, the number of parallel gate supports to consider in the marginals (default $k=2$).
    \item[] {\bf Output}: Error marginal distributions for all the unions of $k$ distinct parallel gate supports.
    \item[]  {\bf Steps}:
Let $H$ have parallel gate supports $\set A_0, \cdots, \set A_{s-1}$. Let $\set U_k$ be the set of all $s \choose k$ possible unions of $k$ distinct gate supports. 
\begin{itemize}
    \item[] For $\set A \in \set U_k$:
        \begin{enumerate}
            \item Call ${\rm PIE}(\set P_{\set A}, H)$ and store the estimated Pauli fidelity sums $\{\hatfidorbit{Q_{\set A}}{H} : Q_{\set A} \in \set P_{\set A}\}$.\\
            \item Use \cref{lem:mufi} to recover the marginal probabilities $\mu(\orbit{Q_{\set A}}{H})$ over the set $\orbit{\set P_{\set A}}{H}$.
        \end{enumerate}
\end{itemize}
\end{itemize}
\end{protocol}

\subsection{Reading Cycle Error Reconstruction Heatmaps}\label{sec:read}

\begin{figure}[h]
\centering
\includegraphics[width=15cm]{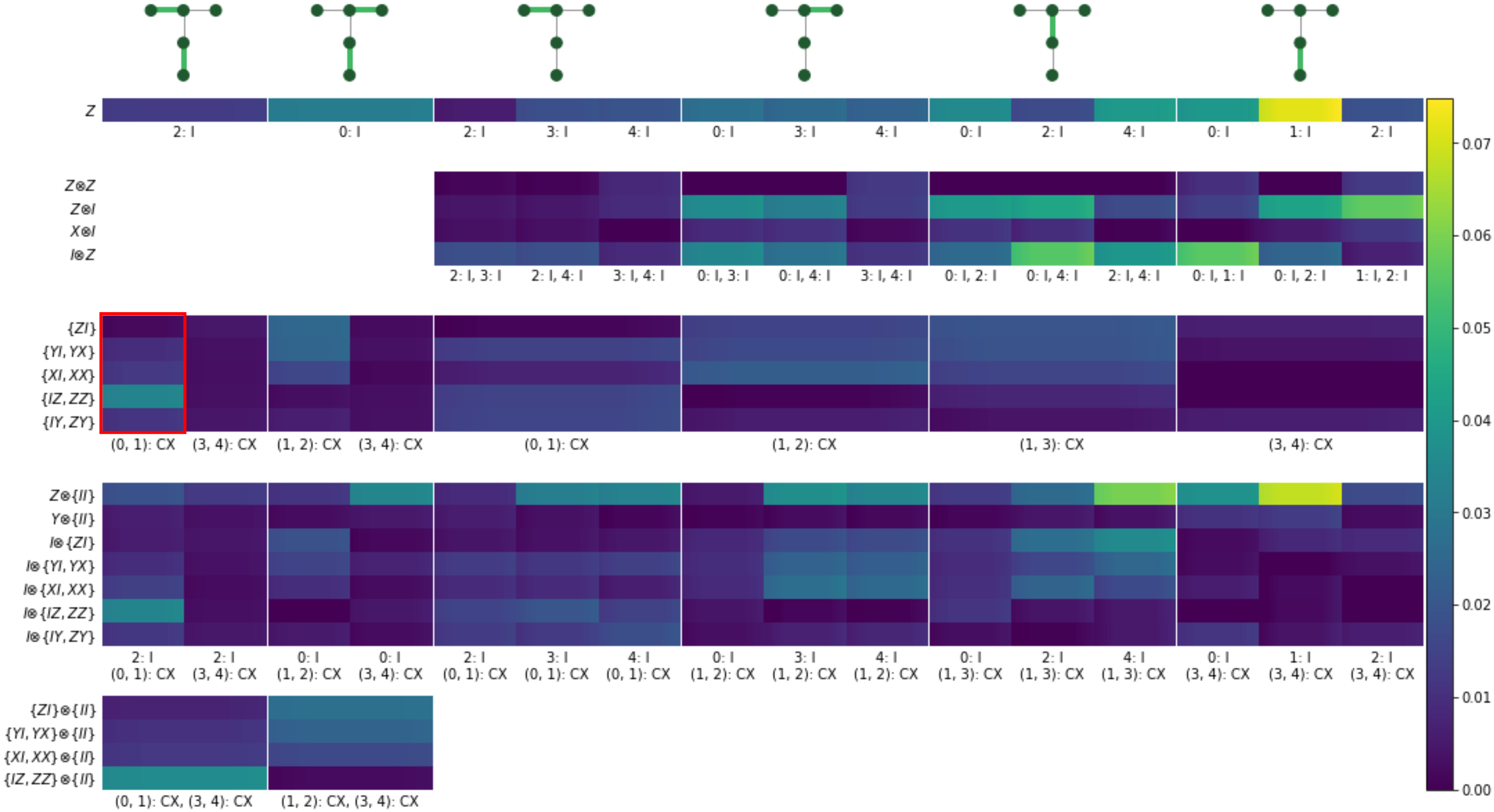}
\caption{Cycle Error Reconstruction ($k \in \{1,2\}$) heatmap from the \texttt{ibmq\_burlington} chip in 2020. Each $T$-shaped graph indicates a different hard cycle; the bold green lines indicate the support of the entangling gates. 
\Cref{sec:read} provides the explanations for interpreting this heatmap. The sub-column enclosed in red serves as an example in the text.} \label{fig:Burlington}
\end{figure}

As shown in \cref{fig:Burlington,fig:Essex,fig:Ourense,fig:Vigo}, we implemented cycle error reconstruction (\cref{pro:cer}) over various
$T$-shaped 5-qubit systems for $k \in \{1,2\}$.
For each system, we reconstructed some marginal error probability distributions for $6$ distinct hard cycles $H_0, \cdots, H_5$.
Given a $T$-shaped system, the qubit labelling starts from $0$ at the top left and increases from left to right and then top to bottom. The hard cycles are composed of identity and CX gates performed in parallel; the support of the entanglers is depicted as a bold green line in the $T$-shaped graphs. We use ``$(i,j): {\rm CX}$'' to denote a CX gate with control $i$ and target $j$.

To properly read the cycle error reconstruction heatmaps (\cref{fig:Burlington,fig:Essex,fig:Ourense,fig:Vigo}), first notice that 
they are partitioned in blocks $B_{i,j}$ (some blocks are left empty) 
together forming $5$ rows and $6$ columns  ($i \in \{0, \cdots ,4\}$, $j \in \{0, \cdots, 5\}$). Each of the $6$ columns $B_{\star, j}$ of our heatmaps contains marginal probabilities associated with a fixed hard cycle $H_j$. $H_j$ can be inferred from the sub-column labels. For example, the hard cycle $H_0$ corresponds to $\{(0,1):{\rm CX}; ~2: {\rm I} ; ~(3,4):{\rm CX}\}$.

Each (non-empty) block $B_{i,j}$ contains between $1$ and $3$ sub-columns, 
each representing a specific marginal 
distribution $\mu_j$ over a set of Paulis $\orbit{\set P_{\set A}}{H_j}$. The qubit support $\set A$ can be identified
via the label located below the sub-column. For instance, if we look below the first sub-column of the block
$B_{2,0}$ (identified in red in \cref{fig:Burlington}), we see ``$(0,1):~{\rm CX}$'', which implies that $\set A=\{0,1\}$
and that the hard cycle $H_0$ acts as $\rm CX$ on those qubits.

The row labels corresponds to different Pauli errors but unspecified qubit support $\set A$.
For instance, consider once again the first sub-column of the block
$B_{2,0}$ (identified in red in \cref{fig:Burlington}). Within that sub-column, the cell for which 
the associated row label is $\{IZ, ZZ\}$ (second from the bottom) indicates the value of 
the marginal probability $\mu_0(IZ_{\{0,1\}})+\mu_0(ZZ_{\{0,1\}})=:\mu_0(\orbit{ZZ_{\{0,1\}}}{H_0})$. 
To improve the readability of the heatmap, when 
some marginal probabilities $\mu_j(\orbit{P_{\set 
A}}{H_j})$ are below some threshold for all the $\set A$s and $H_j$s under consideration, 
their associated row doesn't appear. For example, the probability $\mu_0(IZ_{\{0,1\}})+\mu_0(ZZ_{\{0,1\}})$ is significant, hence the entire row ``$\{IZ, ZZ\}$'' is displayed, and includes the relatively negligible value $\mu_5(IZ_{\{3,4\}})+\mu_5(ZZ_{\{3,4\}})$.

A look at \cref{fig:Burlington,fig:Essex,fig:Ourense,fig:Vigo} reveals a noteworthy similarity: some entangling gates perform particularly well on their respective qubit support, but important crosstalk errors -- which manifest as weight-1 Z errors -- occur on neighbouring idle qubits. In the next section, we introduce a calibration protocol (\cref{proto:sc})
which successfully suppress these dominant errors, as shown in \cref{fig:cer_sc}.

\section{Calibrating Hard Cycles}\label{sec:sc}
In the previous sections, we portrayed hard cycles as monolithic operations, while
in reality they are functions of many tunable parameters. In mathematical terms, 
we may formulate this statement by adding a variable to the noisy representation 
of hard cycles: $\nrep{H, \params}$. Here the vector $\params = (\theta_0, \cdots, \theta_{m-1})$ contains $m$ real parameters which may include information regarding pulse durations and shapes,
but also regarding parameterized gates (such as virtual Zs) meant to act as compensations for coherent error mechanisms.

\subsection{Stochastic Calibration}
Given these additional details, a natural endeavour is to find a choice of parameters that minimizes the probability of error $\sum_{P \neq I} p(P|\params)$. Notice that we shouldn't expect 
the probability of error to reach zero for some parameter value: certain error mechanisms such as
depolarization and dephasing simply cannot be suppressed by calibration procedures. 
More generally, error probabilities can be split in a sum of calibration-sensitive and calibration-insensitive terms,
\begin{align}
p(P|\params)=\pcal{P}+\pncal{P}~,
\end{align}
where
\begin{align}
\pncal{P}:=\min_{\params\in \cali} p(P|\params)~
\end{align}
is the error component that is not sensitive to variations in $\params$. By such definition, calibrating the cycle $H$ reduces to minimizing the sum of calibration-sensitive terms
\begin{align}\label{eq:sum_cal}
    \sum_{P \neq I} \pcal{P}~.
\end{align}
The nuance here is that the sum \ref{eq:sum_cal} may contain fewer non-trivial terms than $\sum_{P \neq I} p(P|\params)$.
Indeed, calibration routines can only address coherent noise sources (that is, coherent \emph{before} their 
stochastic averaging through randomized compiling), and those are usually resulting from few 
processes such as crosstalk and gate over-rotation. For instance, as shown through the Cycle Error 
Reconstruction data from last section (see \cref{fig:Burlington}), crosstalk -- which appears to be the main source of errors --  manifests itself as weight-1
$Z$ errors. As such, we may expect calibration-sensitive errors to be dominated by the following reduced sum
\begin{align}\label{eq:ex_sum}
    \sum_{i=0}^4 \pcal{Z_i}~.
\end{align}
Still carrying with our example, if we neglect calibration-sensitive errors other than weight-1 Z 
errors, it follows from \cref{eq:fid2prob} that the Pauli fidelity 
$f(XXXXX_{\{0,1,2,3,4\}}| \params)$ is maximized if and only if the sum \ref{eq:ex_sum} is
minimized. That is, only one Pauli fidelity estimate may be used as an objective function for 
calibration purposes. More generally, consider the following calibration strategy:
\begin{protocol}[Stochastic Calibration] \label{proto:sc}
    \begin{itemize}
        \item[]
        \item[] {\bf Input}: A hard cycle $H$; a finite set $\Theta =\{\params^{(0)}, \cdots,\params^{(|\Theta|-1)}\}$ of parameter combinations $\params^{(i)} \in \set R^m$; a set $\set S$ of Paulis that defines the following objective function 
        \begin{align}\label{eq:of}
            \mc{O}_{\set S}(\params):=\sum_{Q \in \set S} {f}(\orbit{Q}{H}| \params)~.
        \end{align}
        \item[] {\bf Output}: Calibrated parameters $\params^*=(\theta_0^*, \cdots, \theta_{m-1}^*)$.
        \item[] {\bf Steps}: 
            \begin{itemize}
                \item[] For $\params^{(i)} \in \Theta$:
                    \begin{itemize}
                        \item[1.] Run ${\rm PIE}(\set S, H)$ to get the estimates 
                        $\hat{f}(\orbit{Q}{H}| \params^{(i)}),~ \forall Q \in \set S$.
                        \item[2.] Using \cref{eq:of}, estimate the objective function $\mc{O}_{\set S}$ at $\params^{(i)}$. 
                    \end{itemize}
                \item[] {The optimal parameters $\params^*$ generally won't lie in the discrete probing set $\Theta$. To explore values outside of $\Theta$, one could use gradient-free optimization methods such as Bayesian optimization over some error model family. In this work, as a proof-of-concept, we estimated the objective function over a continuous domain $\mc D$ by performing a curve fitting routine based on the partially reconstructed objective function. We returned $\params^*$ as the estimate of $\argmax_{\params \in \mc D} \mc{O}_{\set S}(\params)$.}
            \end{itemize}
    \end{itemize}
\end{protocol}

\Cref{proto:sc} revolves around constructing and estimating the objective function appearing in \cref{eq:of}. To get a better understanding of the objective function, consider the dual of \cref{eq:fid2prob},
\begin{align}
    f(Q) = \sum_{R \in \set P_n} \chi_{Q}(R) p(R)~,
\end{align}
which may be massaged into
\begin{align}\label{eq:prob2fid}
    f(Q)=1-2 \sum_{\substack{R \in \set P_n \\ \chi_Q(R)=-1}} p(R)~.
\end{align}
By substituting \cref{eq:prob2fid} into \cref{eq:of}, it follows that the calibration is sensitive to all Pauli errors that anti-commute with at least one element of $\orbit{\set S}{H}$. For instance (let's drop the qubit support indices for simplicity), if $\orbit{\set S}{H}=\{YI,IY\}$, 
then the objective function would be sensitive to 
errors in
\begin{align}
    \underbrace{\{XI,ZI,XX,ZX,XY,ZY,XZ,ZZ\}}_{\text{Anti-commutes with }YI} \cup \underbrace{\{IX,IZ,XX,XZ,YX,YZ,ZX,ZZ\}}_{\text{Anti-commutes with }IY}~.
\end{align}

In the formulation of \cref{proto:sc}, we purposely left 
$\set S$ as an unspecified input parameter for the sake of generality.
To make a suitable choice of $\set S$ we may first establish a set of
errors $\set S^{\rm cal}$ that we want our calibration protocol to be sensitive to.
As mentioned earlier, only coherent error processes are subject to calibration-induced suppression. Hence, a knowledge of 
those processes -- based on physical models and/or on Cycle Error Reconstruction evidences -- 
may strongly filter the list of candidates for calibration-sensitive errors.

Once $\set S^{\rm cal}$ is determined, we want to find a set $\set S$ such that 
$\forall Q \in \set S^{\rm cal}$, $\exists R \in \orbit{\set S}{H}$ such that $Q$ and $R$ anti-commute.
This condition usually still leaves many alternatives for $\set S$.
To narrow down those options, we may further take into account the fact that different errors 
$Q \in \set S^{\rm cal}$ are affected by different tunable parameters.
In such case, it is sometimes possible to make a choice of $\set S$ that simplifies the 
objective function to a sum of terms with distinct parameter dependencies:
\begin{align}
    \mc O_{\set S}(\params) = \sum_{i=1}^s \mc O_{\set S_i} ({\bm \alpha}_i)~,
\end{align}
where $\set S_1, \cdots, \set S_s$ partition $\set S$, and where the vector of parameters $\params=\bm{\alpha}_1 \times \cdots \times \bm{\alpha}_s$ is re-expressed as a Cartesian product of $\bm \alpha_i$ vectors.
For example, let $\set S^{\rm cal}=\{X_{0},X_1\}$ and $\params = (\theta_0,\theta_1)$ where the 
tunable parameter $\theta_i$ only affects qubit $i$. The choice $\set S=\{Z_0,Z_1\}$ (let $\set S=\orbit{\set S}{H}$)
yields
\begin{align}
    \mc{O}_{\set S}(\params) = \mc{O}_{Z_0}(\theta_0) + \mc{O}_{Z_1}(\theta_1)~,
\end{align}
whereas the choice $\set S=\{ZZ_{\{0,1\}}\}$ also yields an objective function that is
sensitive to errors in $\set S^{\rm cal}=\{X_{0},X_1\}$, but that can't be re-expressed as a sum
of individual $\theta_i$-dependent terms. The subdivision of the objective function into
many independent parts is often preferable since it reduces the number $|\Theta|$ of parameter combinations
necessary for estimating $\params^*$. 


\subsection{Experimental Results: Stochastic Calibration over Virtual Phase Compensations}

\begin{figure}[h]
\includegraphics[width=8cm]{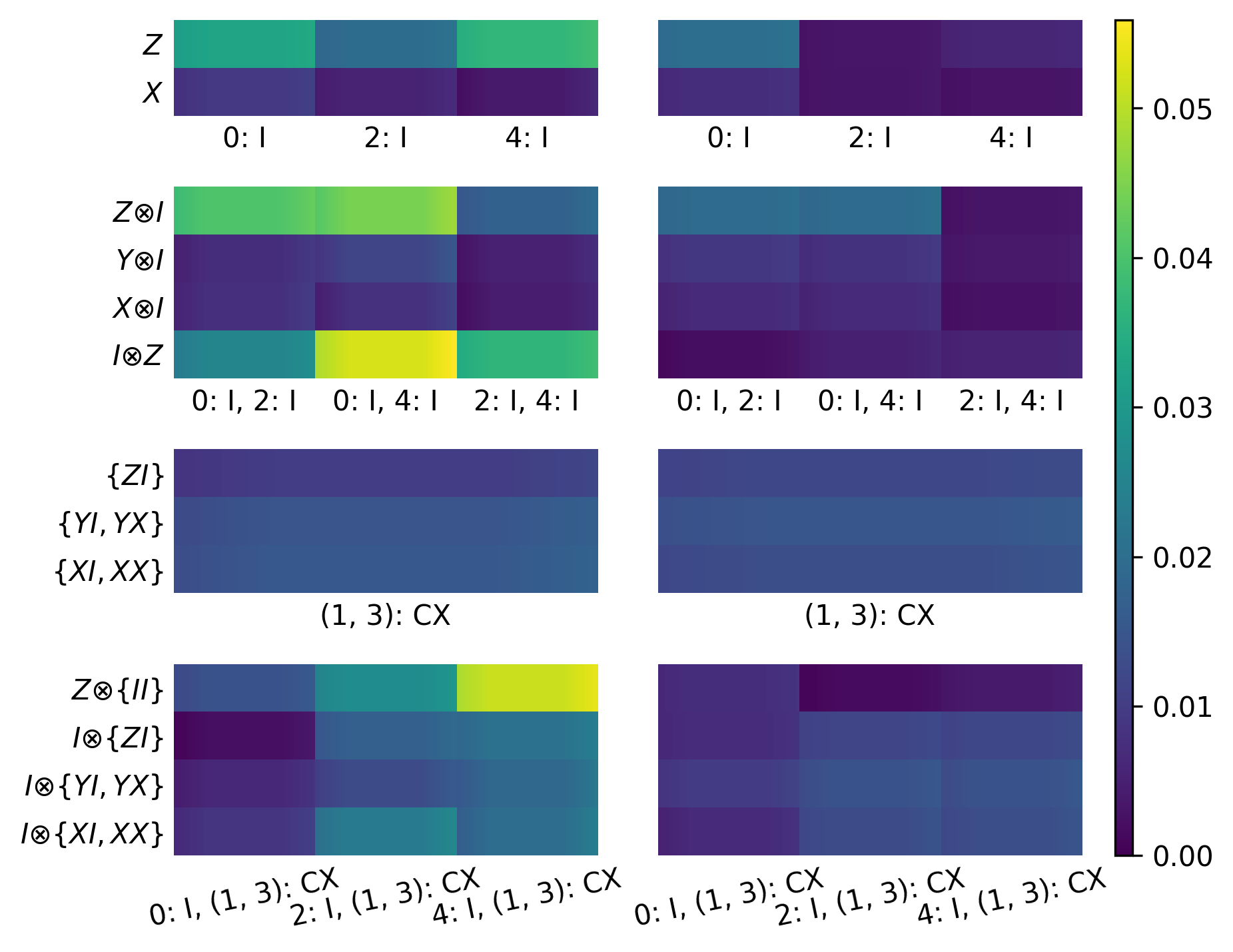}
\caption{Cycle Error Reconstruction heatmaps before (left) and after (right) Stochastic Calibration on the \texttt{ibmq\_burlington} chip (2020). The hard cycle $H$ corresponded to $\{(0,2,4):{\rm I };~ (1,3):{\rm CX}\}$, and the set $\set S$ was chosen to be $\{X_0,X_2,X_4\}$. The tunable parameters corresponded to virtual phase compensations on spectator qubits $\{0,2,4\}$. The objective function was expressed as a sum of three independent terms (see \cref{eq:of_ibm}), which where individually fitted to quadratic curves. The calibrated parameters $\params^*=(-4.8(2)^\circ,-15.1(4)^\circ, -20.8(1)^\circ)$ were obtained by finding the maxima on the three quadratic curves. The heatmap on the right corresponds to the
reconstructed error profile of $\nrep{H, \params^*}$.
}\label{fig:cer_sc}
\end{figure}

We demonstrated the practical effectiveness of \cref{proto:sc} by running it on the \texttt{ibmq\_burlington} chip with a hard cycle $H$ corresponding to $\{(0,2,4):{\rm I };~ (1,3):{\rm CX}\}$. It was apparent from Cycle
Error Reconstruction data -- see the ``before'' part of \cref{fig:cer_sc} -- that the hard cycle was well calibrated on
the pair $\{1,3\}$ but that important crosstalk errors -- which manifested as weight-1 $Z$ errors on the idle qubits $0$, $2$ and $4$ -- still remained. From this observation, we limited tunable parameters to virtual phase compensations on the idle qubits: $\params=(\theta_{\{0\}},\theta_{\{2\}},\theta_{\{4\}})$. Since virtual operations are nearly perfect, 
we assumed that the variation over $\params$ wouldn't introduce other significant types of errors, and we hence established that the set of calibration-sensitive errors was well approximated by $\set S^{\rm cal}=\{Z_{0},Z_2,Z_4\}$. Invoking the locality of virtual phase compensations, we further assumed that error probabilities obeyed
\begin{align}
    p(Z_i|\params)= p(Z_i|\theta_{\{i\}})~.
\end{align}
This assumption motivated the choice $\set S=\{X_0,X_2,X_4\}$, which simplifies the objective function to
a sum of terms with unshared parameter dependencies:
\begin{align}
    \mc O_{\set S}(\params) &= \mc O_{X_0}(\theta_{\{0\}}) + \mc O_{X_2}(\theta_{\{2\}}) +\mc O_{X_4}(\theta_{\{4\}}) \notag \\
    & = f(X_0|\theta_{\{0\}}) + f(X_2|\theta_{\{2\}})+f(X_4|\theta_{\{4\}})~. \label{eq:of_ibm}
\end{align}
Due to \cref{eq:of_ibm}, we could limit $\Theta$ to a set of only $9$ parameter combinations
\begin{align}\label{eq:param_combinations}
    \params^{(j)} = (1-j) \cdot (5^\circ,5^\circ,5^\circ)~,
\end{align}
with $j \in \{0, \cdots, 8\}$. As depicted in \cref{fig:sc_of}, we used the $\rm PIE$ oracle 
to get estimates $\hat{f}(X_{i}| \theta_{\{i\}})$ for qubits $i \in \{0,2,4\}$ and angles 
$\theta_{\{i\}} \in \{-35^\circ, -30^\circ, \cdots , 0^\circ, 5^\circ\} $. We then fitted each data sets
$\{\hat{f}(X_{i}| -35^\circ), \cdots, \hat{f}(X_{i}| 5^\circ) \}$ to a degree-$2$ polynomial via the weighted least squares method. From these fits, we evaluated the maximum of the locally-approximated objective function, 
and found the calibrated parameters $\theta_{\{0\}}^*=-4.8(2)^\circ$, $\theta_{\{2\}}^*=-15.1(4)^\circ$, $\theta_{\{4\}}^*=-20.8(1)^\circ$. Finally, we performed Cycle Error Reconstruction 
on the calibrated cycle $\nrep{H, \params^*}$. The resulting cycle, as shown on the 
``after'' part of \cref{fig:cer_sc}, has much lower error rates on the spectator qubits.

\begin{figure}[h]
\includegraphics[width=15cm]{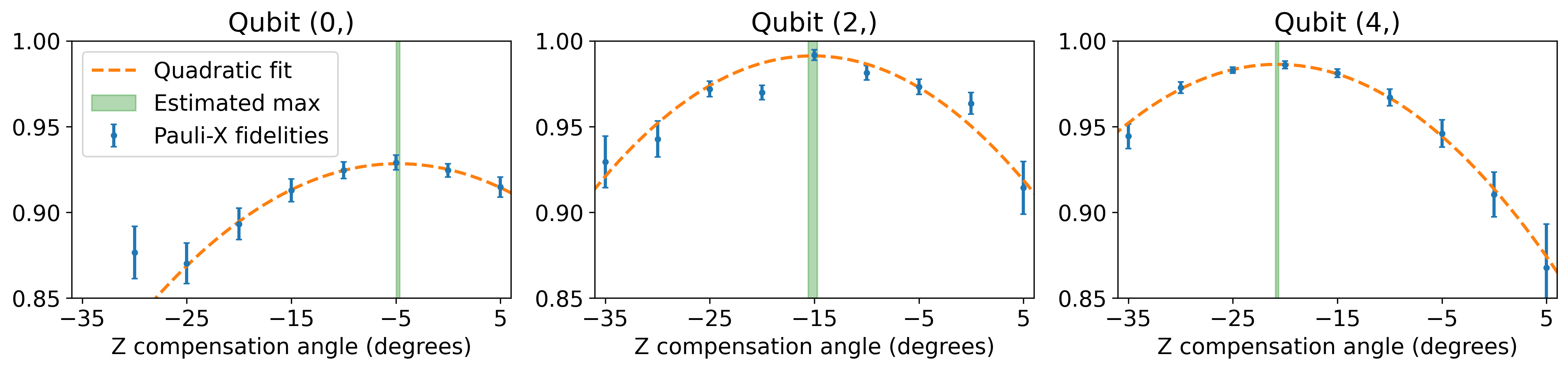}
\caption{Objective function (see \cref{eq:of_ibm}) from Stochastic Calibration on the \texttt{ibmq\_burlington} chip (2020). The hard cycle $H$ corresponded to $\{(0,2,4):{\rm I };~ (1,3):{\rm CX}\}$, and the set $\set S$ was chosen to be $\{X_0,X_2,X_4\}$. The tunable parameters corresponded to virtual phase compensations on spectator qubits $\{0,2,4\}$. The objective function was expressed as a sum of three independent terms (see \cref{eq:of_ibm}), which where individually fitted to quadratic curves based on $9$ estimates (see \cref{eq:param_combinations}). The calibrated parameters $\params^*=(-4.8(2)^\circ,-15.1(4)^\circ, -20.8(1)^\circ)$ were obtained by finding the maxima on the three quadratic curves.}\label{fig:sc_of}
\end{figure}

\section{Conclusion}

In this paper, we presented two protocols, CER and SC, which build upon decades of research into benchmarking and mitigation of errors in quantum systems.
We further demonstrated that randomized compilation (RC), in addition to its traditional use for noise tailoring, has the important benefit of facilitating the learning and calibration procedures for cycle errors.
The protocols presented in this paper, CER and SC, are included in the True-Q software along with RC and other noise characterization, mitigation, and suppression protocols \cite{trueq}.
We began in \cref{sec:RC} by showing that the average RC circuit can be recast as a sequence of \emph{effective dressed cycles} (see \cref{lem:eff}) -- which can be thought of as the circuit primitives -- and that the error profile associated with effective dressed cycles is almost purely stochastic even in the advent of gate-dependent errors (see \cref{lem:pauli_channel}).

Then, in \cref{sec:cer}, we leveraged the structure of CB circuits to introduce \emph{Cycle Error Reconstruction} (CER), a highly scalable protocol that provides, with multiplicative precision, marginal error probability estimates of any effective dressed cycle of interest.
We performed CER for a variety of cycles on multiple IBM-Q devices (see \cref{fig:Burlington,fig:Essex,fig:Ourense,fig:Vigo}).

Finally, in \cref{sec:sc}, we introduced a fast, scalable and efficient CB-structured calibration technique 
-- referred to as \emph{Stochastic Calibration} (SC) -- that provides major improvements 
on the ``\emph{Optimized Randomized Benchmarking for Immediate Tune-up}'' (ORBIT) calibration scheme. We performed
SC on the \texttt{ibmq\_burlington} chip, and added compilation-based phase compensations for 
crosstalk-induced Z rotations, which lead to a $5$-fold error probability reduction of targeted errors (see \cref{fig:cer_sc}).
The protocols introduced in this paper form a solid base from which we can continue to develop strategies to empower quantum computing as we move through the NISQ era.
We have demonstrated that these protocols can enable significant improvements in performance for quantum computing devices.
Further, our noise characterization protocol is scalable and provides detailed information about the error model in a system, providing a state-of-the-art characterization method.

Directions for future work include improvements in scalability and developing methods for extracting more detailed noise profiles at similar cost, including resolving degeneracies and characterizing coherent errors.
Some researchers have begun exploring these areas, however there remains plenty of future work \cite{Chen2022, Huang2022,carignandugas2023estimating}.
The information retrieved from CER can be used to inform mitigation protocols (and has been in e.g. refs. \cite{Ferracin2022, Berg2022}).
Future applications of CER could include the selection of error correcting codes and recovery operations, and characterization and mitigation protocols are already being explored for the quantum error correction regime \cite{\qeclit}.

\emph{Acknowledgements}-- We acknowledge the use of IBM Quantum services for this work. The views expressed are those of the authors, and do not reflect the official policy or position of IBM or the IBM Quantum team. This research was supported by the Government of Canada, the U.S. Army Research Office through grant W911NF-21-1-0007, Transformative Quantum Technologies, and Quantum Benchmark Inc.

\bibliographystyle{unsrtnat}
\bibliography{qcvv}

\begin{thebibliography}{96}
\providecommand{\natexlab}[1]{#1}
\providecommand{\url}[1]{\texttt{#1}}
\expandafter\ifx\csname urlstyle\endcsname\relax
  \providecommand{\doi}[1]{doi: #1}\else
  \providecommand{\doi}{doi: \begingroup \urlstyle{rm}\Url}\fi

\bibitem[Chuang and Nielsen(1997)]{Chuang1997}
Isaac~L. Chuang and Michael~A. Nielsen.
\newblock {Prescription for experimental determination of the dynamics of a
  quantum black box}.
\newblock \emph{Journal of Modern Optics}, 44\penalty0 (11-12):\penalty0 2455,
  November 1997.
\newblock ISSN 0950-0340.
\newblock \doi{10.1080/09500349708231894}.
\newblock URL
  \url{http://www.tandfonline.com/doi/abs/10.1080/09500349708231894}.

\bibitem[Branderhorst et~al.(2009)Branderhorst, Nunn, Walmsley, and
  Kosut]{Branderhorst_2009}
M~P~A Branderhorst, J~Nunn, I~A Walmsley, and R~L Kosut.
\newblock Simplified quantum process tomography.
\newblock \emph{New Journal of Physics}, 11\penalty0 (11):\penalty0 115010, nov
  2009.
\newblock \doi{10.1088/1367-2630/11/11/115010}.
\newblock URL \url{https://doi.org/10.1088%2F1367-2630%2F11%2F11%2F115010}.

\bibitem[Shabani et~al.(2011)Shabani, Kosut, Mohseni, Rabitz, Broome, Almeida,
  Fedrizzi, and White]{Shabani2011}
A.~Shabani, R.~L. Kosut, M.~Mohseni, H.~Rabitz, M.~A. Broome, M.~P. Almeida,
  A.~Fedrizzi, and A.~G. White.
\newblock Efficient measurement of quantum dynamics via compressive sensing.
\newblock \emph{Phys. Rev. Lett.}, 106:\penalty0 100401, Mar 2011.
\newblock \doi{10.1103/PhysRevLett.106.100401}.
\newblock URL \url{https://link.aps.org/doi/10.1103/PhysRevLett.106.100401}.

\bibitem[Flammia et~al.(2012)Flammia, Gross, Liu, and Eisert]{Flammia2012}
Steven~T. Flammia, David Gross, Yi-Kai Liu, and Jens Eisert.
\newblock {Quantum tomography via compressed sensing: error bounds, sample
  complexity and efficient estimators}.
\newblock \emph{New Journal of Physics}, 14\penalty0 (9):\penalty0 095022,
  September 2012.
\newblock ISSN 1367-2630.
\newblock \doi{10.1088/1367-2630/14/9/095022}.
\newblock URL
  \url{http://stacks.iop.org/1367-2630/14/i=9/a=095022?key=crossref.8f6ba9f6f004c7996a1f28c85ef1b8f3}.

\bibitem[Aaronson(2017)]{Aaronson2017}
Scott Aaronson.
\newblock Shadow tomography of quantum states.
\newblock 2017.
\newblock \doi{10.48550/ARXIV.1711.01053}.
\newblock URL \url{https://arxiv.org/abs/1711.01053}.

\bibitem[Huang et~al.(2020)Huang, Kueng, and Preskill]{Huang2020}
Hsin-Yuan Huang, Richard Kueng, and John Preskill.
\newblock Predicting many properties of a quantum system from very few
  measurements.
\newblock \emph{Nature Physics}, 16\penalty0 (10):\penalty0 1050--1057, jun
  2020.
\newblock \doi{10.1038/s41567-020-0932-7}.
\newblock URL \url{https://doi.org/10.1038\%2Fs41567-020-0932-7}.

\bibitem[Evans et~al.(2019)Evans, Harper, and Flammia]{Evans2019}
Tim~J. Evans, Robin Harper, and Steven~T. Flammia.
\newblock Scalable bayesian hamiltonian learning.
\newblock 2019.
\newblock \doi{10.48550/ARXIV.1912.07636}.
\newblock URL \url{https://arxiv.org/abs/1912.07636}.

\bibitem[{Flammia} and {O'Donnell}(2021)]{FlammiaPopRec}
Steven~T. {Flammia} and Ryan {O'Donnell}.
\newblock {Pauli error estimation via Population Recovery}.
\newblock \emph{Quantum}, 5:\penalty0 549, September 2021.
\newblock \doi{10.22331/q-2021-09-23-549}.

\bibitem[Kunjummen et~al.(2021)Kunjummen, Tran, Carney, and
  Taylor]{Kunjummen2021}
Jonathan Kunjummen, Minh~C. Tran, Daniel Carney, and Jacob~M. Taylor.
\newblock Shadow process tomography of quantum channels.
\newblock 2021.
\newblock \doi{10.48550/ARXIV.2110.03629}.
\newblock URL \url{https://arxiv.org/abs/2110.03629}.

\bibitem[Levy et~al.(2021)Levy, Luo, and Clark]{Levy2021}
Ryan Levy, Di~Luo, and Bryan~K. Clark.
\newblock Classical shadows for quantum process tomography on near-term quantum
  computers.
\newblock 2021.
\newblock \doi{10.48550/ARXIV.2110.02965}.
\newblock URL \url{https://arxiv.org/abs/2110.02965}.

\bibitem[Bertoni et~al.(2022)Bertoni, Haferkamp, Hinsche, Ioannou, Eisert, and
  Pashayan]{Bertoni2022}
Christian Bertoni, Jonas Haferkamp, Marcel Hinsche, Marios Ioannou, Jens
  Eisert, and Hakop Pashayan.
\newblock Shallow shadows: Expectation estimation using low-depth random
  clifford circuits, 2022.
\newblock URL \url{https://arxiv.org/abs/2209.12924}.

\bibitem[Akhtar et~al.(2022)Akhtar, Hu, and You]{Akhtar2022}
Ahmed~A. Akhtar, Hong-Ye Hu, and Yi-Zhuang You.
\newblock Scalable and flexible classical shadow tomography with tensor
  networks, 2022.
\newblock URL \url{https://arxiv.org/abs/2209.02093}.

\bibitem[Hu and You(2022)]{Hu2022}
Hong-Ye Hu and Yi-Zhuang You.
\newblock Hamiltonian-driven shadow tomography of quantum states.
\newblock \emph{Physical Review Research}, 4\penalty0 (1), jan 2022.
\newblock \doi{10.1103/physrevresearch.4.013054}.
\newblock URL \url{https://doi.org/10.1103%2Fphysrevresearch.4.013054}.

\bibitem[Chen et~al.(2021)Chen, Yu, Zeng, and Flammia]{Chen_2021}
Senrui Chen, Wenjun Yu, Pei Zeng, and Steven~T. Flammia.
\newblock Robust shadow estimation.
\newblock \emph{{PRX} Quantum}, 2\penalty0 (3), sep 2021.
\newblock \doi{10.1103/prxquantum.2.030348}.
\newblock URL \url{https://doi.org/10.1103%2Fprxquantum.2.030348}.

\bibitem[Anshu et~al.(2021)Anshu, Arunachalam, Kuwahara, and
  Soleimanifar]{Anshu_2021}
Anurag Anshu, Srinivasan Arunachalam, Tomotaka Kuwahara, and Mehdi
  Soleimanifar.
\newblock Sample-efficient learning of interacting quantum systems.
\newblock \emph{Nature Physics}, 17\penalty0 (8):\penalty0 931--935, may 2021.
\newblock \doi{10.1038/s41567-021-01232-0}.
\newblock URL \url{https://doi.org/10.1038%2Fs41567-021-01232-0}.

\bibitem[Wilde et~al.(2022)Wilde, Kshetrimayum, Roth, Hangleiter, Sweke, and
  Eisert]{Wilde2022}
Frederik Wilde, Augustine Kshetrimayum, Ingo Roth, Dominik Hangleiter, Ryan
  Sweke, and Jens Eisert.
\newblock Scalably learning quantum many-body hamiltonians from dynamical data.
\newblock 2022.
\newblock \doi{10.48550/ARXIV.2209.14328}.
\newblock URL \url{https://arxiv.org/abs/2209.14328}.

\bibitem[França et~al.(2022)França, Markovich, Dobrovitski, Werner, and
  Borregaard]{Franca2022}
Daniel~Stilck França, Liubov~A. Markovich, V.~V. Dobrovitski, Albert~H.
  Werner, and Johannes Borregaard.
\newblock Efficient and robust estimation of many-qubit hamiltonians.
\newblock 2022.
\newblock \doi{10.48550/ARXIV.2205.09567}.
\newblock URL \url{https://arxiv.org/abs/2205.09567}.

\bibitem[Haah et~al.(2021)Haah, Kothari, and Tang]{Haah2021}
Jeongwan Haah, Robin Kothari, and Ewin Tang.
\newblock Optimal learning of quantum hamiltonians from high-temperature gibbs
  states.
\newblock 2021.
\newblock \doi{10.48550/ARXIV.2108.04842}.
\newblock URL \url{https://arxiv.org/abs/2108.04842}.

\bibitem[Emerson et~al.(2005)Emerson, Alicki, and \.{Z}yczkowski]{Emerson2005}
Joseph Emerson, Robert Alicki, and Karol \.{Z}yczkowski.
\newblock {Scalable noise estimation with random unitary operators}.
\newblock \emph{Journal of Optics B: Quantum and Semiclassical Optics},
  7\penalty0 (10):\penalty0 S347, October 2005.
\newblock ISSN 1464-4266.
\newblock \doi{10.1088/1464-4266/7/10/021}.
\newblock URL
  \url{http://stacks.iop.org/1464-4266/7/i=10/a=021?key=crossref.0612a2d913e264bbddc29e955e77423d}.

\bibitem[Dankert et~al.(2009); arxiv.org/abs/quant-ph/0606161 (2006)Dankert,
  Cleve, Emerson, and Livine]{DCEL2006}
C.~Dankert, R.~Cleve, J.~Emerson, and E.~Livine.
\newblock Exact and approximate unitary 2-designs and their applications to
  fidelity estimation.
\newblock \emph{Phys. Rev. A}, 80\penalty0 (012304), 2009);
  arxiv.org/abs/quant-ph/0606161 (2006.

\bibitem[Magesan et~al.(2011)Magesan, Gambetta, and Emerson]{Magesan2011}
Easwar Magesan, Jay~M. Gambetta, and Joseph Emerson.
\newblock {Scalable and Robust Randomized Benchmarking of Quantum Processes}.
\newblock \emph{Physical Review Letters}, 106\penalty0 (18):\penalty0 180504,
  May 2011.
\newblock ISSN 0031-9007.
\newblock \doi{10.1103/PhysRevLett.106.180504}.
\newblock URL \url{http://link.aps.org/doi/10.1103/PhysRevLett.106.180504}.

\bibitem[Magesan et~al.(2012{\natexlab{a}})Magesan, Gambetta, and
  Emerson]{Magesan2012a}
Easwar Magesan, Jay~M. Gambetta, and Joseph Emerson.
\newblock {Characterizing quantum gates via randomized benchmarking}.
\newblock \emph{Physical Review A}, 85\penalty0 (4):\penalty0 042311, April
  2012{\natexlab{a}}.
\newblock ISSN 1050-2947.
\newblock \doi{10.1103/PhysRevA.85.042311}.
\newblock URL \url{http://link.aps.org/doi/10.1103/PhysRevA.85.042311}.

\bibitem[{Proctor} et~al.(2017){Proctor}, {Rudinger}, {Young}, {Sarovar}, and
  {Blume-Kohout}]{Proctor2017}
Timothy {Proctor}, Kenneth {Rudinger}, Kevin {Young}, Mohan {Sarovar}, and
  Robin {Blume-Kohout}.
\newblock {What Randomized Benchmarking Actually Measures}.
\newblock \emph{Physical Review Letter}, 119:\penalty0 130502, September 2017.
\newblock \doi{10.1103/PhysRevLett.119.130502}.

\bibitem[Wallman(2018)]{Wallman2017}
Joel~J. Wallman.
\newblock Randomized benchmarking with gate-dependent noise.
\newblock \emph{{Quantum}}, 2:\penalty0 47, January 2018.
\newblock ISSN 2521-327X.
\newblock \doi{10.22331/q-2018-01-29-47}.
\newblock URL \url{https://doi.org/10.22331/q-2018-01-29-47}.

\bibitem[Merkel et~al.(2021)Merkel, Pritchett, and Fong]{Merkel2018}
Seth~T. Merkel, Emily~J. Pritchett, and Bryan~H. Fong.
\newblock Randomized benchmarking as convolution: Fourier analysis of gate
  dependent errors.
\newblock \emph{Quantum}, 5:\penalty0 581, nov 2021.
\newblock \doi{10.22331/q-2021-11-16-581}.
\newblock URL \url{https://doi.org/10.22331%2Fq-2021-11-16-581}.

\bibitem[{Carignan-Dugas} et~al.(2018){Carignan-Dugas}, {Boone}, {Wallman}, and
  {Emerson}]{Dugas2018}
Arnaud {Carignan-Dugas}, Kristine {Boone}, Joel~J. {Wallman}, and Joseph
  {Emerson}.
\newblock {From randomized benchmarking experiments to gate-set circuit
  fidelity: how to interpret randomized benchmarking decay parameters}.
\newblock \emph{New Journal of Physics}, 20:\penalty0 092001, September 2018.
\newblock \doi{10.1088/1367-2630/aadcc7}.

\bibitem[Magesan et~al.(2012{\natexlab{b}})Magesan, Gambetta, Johnson, Ryan,
  Chow, Merkel, da~Silva, Keefe, Rothwell, Ohki, Ketchen, and
  Steffen]{Magesan2012b}
Easwar Magesan, Jay~M. Gambetta, B.~R. Johnson, Colm~A. Ryan, Jerry~M. Chow,
  Seth~T. Merkel, Marcus~P. da~Silva, George~A. Keefe, Mary~B. Rothwell,
  Thomas~A. Ohki, Mark~B. Ketchen, and M.~Steffen.
\newblock {Efficient Measurement of Quantum Gate Error by Interleaved
  Randomized Benchmarking}.
\newblock \emph{Physical Review Letters}, 109\penalty0 (8):\penalty0 080505,
  August 2012{\natexlab{b}}.
\newblock ISSN 0031-9007.
\newblock \doi{10.1103/PhysRevLett.109.080505}.
\newblock URL \url{http://arxiv.org/abs/1203.4550}.

\bibitem[{Gambetta} et~al.(2012){Gambetta}, {C{\'o}rcoles}, {Merkel},
  {Johnson}, {Smolin}, {Chow}, {Ryan}, {Rigetti}, {Poletto}, {Ohki}, {Ketchen},
  and {Steffen}]{GambettaCorcoles2012}
Jay~M. {Gambetta}, A.~D. {C{\'o}rcoles}, S.~T. {Merkel}, B.~R. {Johnson},
  John~A. {Smolin}, Jerry~M. {Chow}, Colm~A. {Ryan}, Chad {Rigetti},
  S.~{Poletto}, Thomas~A. {Ohki}, Mark~B. {Ketchen}, and M.~{Steffen}.
\newblock {Characterization of Addressability by Simultaneous Randomized
  Benchmarking}.
\newblock \emph{Phys. Rev. Lett.}, 109:\penalty0 240504, Dec 2012.
\newblock \doi{10.1103/PhysRevLett.109.240504}.

\bibitem[{Barends} et~al.(2014){Barends}, {Kelly}, {Veitia}, {Megrant},
  {Fowler}, {Campbell}, {Chen}, {Chen}, {Chiaro}, {Dunsworth}, {Hoi},
  {Jeffrey}, {Neill}, {O'Malley}, {Mutus}, {Quintana}, {Roushan}, {Sank},
  {Wenner}, {White}, {Korotkov}, {Cleland}, and {Martinis}]{Barends2014dice}
R.~{Barends}, J.~{Kelly}, A.~{Veitia}, A.~{Megrant}, A.~G. {Fowler},
  B.~{Campbell}, Y.~{Chen}, Z.~{Chen}, B.~{Chiaro}, A.~{Dunsworth}, I.~C.
  {Hoi}, E.~{Jeffrey}, C.~{Neill}, P.~J.~J. {O'Malley}, J.~{Mutus},
  C.~{Quintana}, P.~{Roushan}, D.~{Sank}, J.~{Wenner}, T.~C. {White}, A.~N.
  {Korotkov}, A.~N. {Cleland}, and John~M. {Martinis}.
\newblock {Rolling quantum dice with a superconducting qubit}.
\newblock \emph{Physical Review A}, 90:\penalty0 030303, 2014.
\newblock \doi{10.1103/PhysRevA.90.030303}.

\bibitem[Epstein et~al.(2014)Epstein, Cross, Magesan, and
  Gambetta]{Epstein2014}
Jeffrey~M. Epstein, Andrew~W. Cross, Easwar Magesan, and Jay~M. Gambetta.
\newblock {Investigating the limits of randomized benchmarking protocols}.
\newblock \emph{Physical Review A}, 89\penalty0 (6):\penalty0 062321, jun 2014.
\newblock ISSN 1050-2947.
\newblock \doi{10.1103/PhysRevA.89.062321}.
\newblock URL \url{http://link.aps.org/doi/10.1103/PhysRevA.89.062321}.

\bibitem[Granade et~al.(2014)Granade, Ferrie, and Cory]{Granade2014}
Christopher Granade, Christopher Ferrie, and David~G. Cory.
\newblock {Accelerated Randomized Benchmarking}.
\newblock \emph{New Journal of Physics}, 17\penalty0 (1):\penalty0 013042,
  2014.
\newblock ISSN 13672630.
\newblock \doi{10.1088/1367-2630/17/1/013042}.
\newblock URL \url{http://dx.doi.org/10.1088/1367-2630/17/1/013042}.

\bibitem[Wallman and Flammia(2014)]{Wallman2014}
Joel~J. Wallman and Steven~T. Flammia.
\newblock {Randomized benchmarking with confidence}.
\newblock \emph{New Journal of Physics}, 16\penalty0 (10):\penalty0 103032,
  October 2014.
\newblock ISSN 1367-2630.
\newblock \doi{10.1088/1367-2630/16/10/103032}.
\newblock URL
  \url{http://stacks.iop.org/1367-2630/16/i=10/a=103032?key=crossref.f0d3ce1050f8d1f51d7d311740e43a2f}.

\bibitem[{Carignan-Dugas} et~al.(2015){Carignan-Dugas}, {Wallman}, and
  {Emerson}]{Dugas2015}
A.~{Carignan-Dugas}, J.~J. {Wallman}, and J.~{Emerson}.
\newblock {Characterizing universal gate sets via dihedral benchmarking}.
\newblock \emph{Physical Review A}, 92\penalty0 (6):\penalty0 060302, December
  2015.
\newblock \doi{10.1103/PhysRevA.92.060302}.

\bibitem[Wallman et~al.(2015{\natexlab{a}})Wallman, Granade, Harper, and
  Flammia]{Wallman2015}
Joel Wallman, Chris Granade, Robin Harper, and Steven~T Flammia.
\newblock Estimating the coherence of noise.
\newblock \emph{New Journal of Physics}, 17\penalty0 (11):\penalty0 113020,
  2015{\natexlab{a}}.
\newblock URL \url{http://stacks.iop.org/1367-2630/17/i=11/a=113020}.

\bibitem[Wallman et~al.(2015{\natexlab{b}})Wallman, Barnhill, and
  Emerson]{Wallman2015b}
Joel~J. Wallman, Marie Barnhill, and Joseph Emerson.
\newblock {Robust Characterization of Loss Rates}.
\newblock \emph{Physical Review Letters}, 115\penalty0 (6):\penalty0 060501,
  2015{\natexlab{b}}.
\newblock ISSN 0031-9007.
\newblock \doi{10.1103/PhysRevLett.115.060501}.
\newblock URL \url{http://link.aps.org/doi/10.1103/PhysRevLett.115.060501}.

\bibitem[{Sheldon} et~al.(2016){Sheldon}, {Magesan}, {Chow}, and
  {Gambetta}]{Sheldon2016}
Sarah {Sheldon}, Easwar {Magesan}, Jerry~M. {Chow}, and Jay~M. {Gambetta}.
\newblock {Procedure for systematically tuning up cross-talk in the
  cross-resonance gate}.
\newblock \emph{Physical Review A}, 93:\penalty0 060302, June 2016.
\newblock \doi{10.1103/PhysRevA.93.060302}.

\bibitem[{Cross} et~al.(2016){Cross}, {Magesan}, {Bishop}, {Smolin}, and
  {Gambetta}]{Cross2016}
A.~W. {Cross}, E.~{Magesan}, L.~S. {Bishop}, J.~A. {Smolin}, and J.~M.
  {Gambetta}.
\newblock {Scalable randomised benchmarking of non-Clifford gates}.
\newblock \emph{npj Quantum Information}, 2:\penalty0 16012, April 2016.
\newblock \doi{10.1038/npjqi.2016.12}.

\bibitem[Combes et~al.(2017)Combes, Granade, Ferrie, and Flammia]{Combes2017}
Joshua Combes, Christopher Granade, Christopher Ferrie, and Steven~T. Flammia.
\newblock {Logical Randomized Benchmarking}.
\newblock feb 2017.
\newblock URL \url{http://arxiv.org/abs/1702.03688}.

\bibitem[Hashagen et~al.(2018)Hashagen, Flammia, Gross, and
  Wallman]{Hashagen2018}
A.~K. Hashagen, S.~T. Flammia, D.~Gross, and J.~J. Wallman.
\newblock Real randomized benchmarking.
\newblock \emph{Quantum}, 2:\penalty0 85, aug 2018.
\newblock \doi{10.22331/q-2018-08-22-85}.
\newblock URL \url{https://doi.org/10.22331%2Fq-2018-08-22-85}.

\bibitem[{Brown} and {Eastin}(2018)]{Brown2018}
Winton~G. {Brown} and Bryan {Eastin}.
\newblock {Randomized benchmarking with restricted gate sets}.
\newblock \emph{Physical Review A}, 97:\penalty0 062323, June 2018.
\newblock \doi{10.1103/PhysRevA.97.062323}.

\bibitem[{Fran{\c{c}}a} and {Hashagen}(2018)]{Franca2018}
D.~S. {Fran{\c{c}}a} and A.~K. {Hashagen}.
\newblock {Approximate randomized benchmarking for finite groups}.
\newblock \emph{Journal of Physics A Mathematical General}, 51\penalty0
  (39):\penalty0 395302, Sep 2018.
\newblock \doi{10.1088/1751-8121/aad6fa}.

\bibitem[{Helsen} et~al.(2019){Helsen}, {Xue}, {Vandersypen}, and
  {Wehner}]{Helsen2018}
Jonas {Helsen}, Xiao {Xue}, Lieven M.~K. {Vandersypen}, and Stephanie {Wehner}.
\newblock {A new class of efficient randomized benchmarking protocols}.
\newblock \emph{npj Quantum Information}, 5:\penalty0 71, August 2019.
\newblock \doi{10.1038/s41534-019-0182-7}.

\bibitem[Onorati et~al.(2019)Onorati, Werner, and Eisert]{Onorati2018}
E.~Onorati, A.{\hspace{0.167em} }H. Werner, and J.~Eisert.
\newblock Randomized benchmarking for individual quantum gates.
\newblock \emph{Physical Review Letters}, 123\penalty0 (6), aug 2019.
\newblock \doi{10.1103/physrevlett.123.060501}.
\newblock URL \url{https://doi.org/10.1103%2Fphysrevlett.123.060501}.

\bibitem[Proctor et~al.(2019)Proctor, Carignan-Dugas, Rudinger, Nielsen,
  Blume-Kohout, and Young]{Proctor2018}
Timothy~J. Proctor, Arnaud Carignan-Dugas, Kenneth Rudinger, Erik Nielsen,
  Robin Blume-Kohout, and Kevin Young.
\newblock Direct randomized benchmarking for multiqubit devices.
\newblock \emph{Physical Review Letters}, 123\penalty0 (3), jul 2019.
\newblock \doi{10.1103/physrevlett.123.030503}.
\newblock URL \url{https://doi.org/10.1103%2Fphysrevlett.123.030503}.

\bibitem[Boone et~al.(2019)Boone, {Carignan-Dugas}, Wallman, and
  Emerson]{Boone2019}
Kristine Boone, Arnaud {Carignan-Dugas}, Joel~J. Wallman, and Joseph Emerson.
\newblock Randomized benchmarking under different gate sets.
\newblock \emph{Phys. Rev. A}, 99:\penalty0 032329, Mar 2019.
\newblock \doi{10.1103/PhysRevA.99.032329}.
\newblock URL \url{https://link.aps.org/doi/10.1103/PhysRevA.99.032329}.

\bibitem[Helsen et~al.(2019{\natexlab{a}})Helsen, Wallman, Flammia, and
  Wehner]{Helsen2019}
Jonas Helsen, Joel~J. Wallman, Steven~T. Flammia, and Stephanie Wehner.
\newblock Multiqubit randomized benchmarking using few samples.
\newblock \emph{Physical Review A}, 100\penalty0 (3), Sep 2019{\natexlab{a}}.
\newblock ISSN 2469-9934.
\newblock \doi{10.1103/physreva.100.032304}.
\newblock URL \url{http://dx.doi.org/10.1103/PhysRevA.100.032304}.

\bibitem[{Dirkse} et~al.(2019){Dirkse}, {Helsen}, and {Wehner}]{Dirkse2019}
Bas {Dirkse}, Jonas {Helsen}, and Stephanie {Wehner}.
\newblock {Efficient unitarity randomized benchmarking of few-qubit Clifford
  gates}.
\newblock \emph{Physical Review A}, 99:\penalty0 012315, January 2019.
\newblock \doi{10.1103/PhysRevA.99.012315}.

\bibitem[Harper et~al.(2019)Harper, Hincks, Ferrie, Flammia, and
  Wallman]{Harper2019}
Robin Harper, Ian Hincks, Chris Ferrie, Steven~T. Flammia, and Joel~J. Wallman.
\newblock Statistical analysis of randomized benchmarking.
\newblock \emph{Phys. Rev. A}, 99:\penalty0 052350, May 2019.
\newblock \doi{10.1103/PhysRevA.99.052350}.
\newblock URL \url{https://link.aps.org/doi/10.1103/PhysRevA.99.052350}.

\bibitem[Erhard et~al.(2019)Erhard, Wallman, Postler, Meth, Stricker, Martinez,
  Schindler, Monz, Emerson, and Blatt]{Erhard2019}
Alexander Erhard, Joel~J. Wallman, Lukas Postler, Michael Meth, Roman Stricker,
  Esteban~A. Martinez, Philipp Schindler, Thomas Monz, Joseph Emerson, and
  Rainer Blatt.
\newblock Characterizing large-scale quantum computers via cycle benchmarking.
\newblock \emph{Nature Communications}, 10\penalty0 (1), nov 2019.
\newblock \doi{10.1038/s41467-019-13068-7}.
\newblock URL \url{https://doi.org/10.1038%2Fs41467-019-13068-7}.

\bibitem[Helsen et~al.(2022)Helsen, Roth, Onorati, Werner, and
  Eisert]{Helsen2020}
J.~Helsen, I.~Roth, E.~Onorati, A.H. Werner, and J.~Eisert.
\newblock General framework for randomized benchmarking.
\newblock \emph{{PRX} Quantum}, 3\penalty0 (2), jun 2022.
\newblock \doi{10.1103/prxquantum.3.020357}.
\newblock URL \url{https://doi.org/10.1103%2Fprxquantum.3.020357}.

\bibitem[Proctor et~al.(2022)Proctor, Seritan, Rudinger, Nielsen, Blume-Kohout,
  and Young]{Proctor2022}
Timothy Proctor, Stefan Seritan, Kenneth Rudinger, Erik Nielsen, Robin
  Blume-Kohout, and Kevin Young.
\newblock Scalable randomized benchmarking of quantum computers using mirror
  circuits.
\newblock \emph{Physical Review Letters}, 129\penalty0 (15), oct 2022.
\newblock \doi{10.1103/physrevlett.129.150502}.
\newblock URL \url{https://doi.org/10.1103%2Fphysrevlett.129.150502}.

\bibitem[Hines et~al.(2022)Hines, Lu, Naik, Hashim, Ville, Mitchell,
  Kriekebaum, Santiago, Seritan, Nielsen, Blume-Kohout, Young, Siddiqi, Whaley,
  and Proctor]{Hines2022}
Jordan Hines, Marie Lu, Ravi~K. Naik, Akel Hashim, Jean-Loup Ville, Brad
  Mitchell, John~Mark Kriekebaum, David~I. Santiago, Stefan Seritan, Erik
  Nielsen, Robin Blume-Kohout, Kevin Young, Irfan Siddiqi, Birgitta Whaley, and
  Timothy Proctor.
\newblock Demonstrating scalable randomized benchmarking of universal gate
  sets.
\newblock 2022.
\newblock \doi{10.48550/ARXIV.2207.07272}.
\newblock URL \url{https://arxiv.org/abs/2207.07272}.

\bibitem[{Polloreno} et~al.(2023){Polloreno}, {Carignan-Dugas}, {Hines},
  {Blume-Kohout}, {Young}, and {Proctor}]{Polloreno2023_DRB}
Anthony~M. {Polloreno}, Arnaud {Carignan-Dugas}, Jordan {Hines}, Robin
  {Blume-Kohout}, Kevin {Young}, and Timothy {Proctor}.
\newblock {A Theory of Direct Randomized Benchmarking}.
\newblock \emph{arXiv e-prints}, art. arXiv:2302.13853, February 2023.
\newblock \doi{10.48550/arXiv.2302.13853}.

\bibitem[{Cross} et~al.(2019){Cross}, {Bishop}, {Sheldon}, {Nation}, and
  {Gambetta}]{Quantum_volume2019}
Andrew~W. {Cross}, Lev~S. {Bishop}, Sarah {Sheldon}, Paul~D. {Nation}, and
  Jay~M. {Gambetta}.
\newblock {Validating quantum computers using randomized model circuits}.
\newblock \emph{Phys. Rev. A}, 100\penalty0 (3):\penalty0 032328, September
  2019.
\newblock \doi{10.1103/PhysRevA.100.032328}.

\bibitem[Boixo et~al.(2018)Boixo, Isakov, Smelyanskiy, Babbush, Ding, Jiang,
  Bremner, Martinis, and Neven]{XEB2016}
Sergio Boixo, Sergei~V. Isakov, Vadim~N. Smelyanskiy, Ryan Babbush, Nan Ding,
  Zhang Jiang, Michael~J. Bremner, John~M. Martinis, and Hartmut Neven.
\newblock Characterizing quantum supremacy in near-term devices.
\newblock \emph{Nature Physics}, 14\penalty0 (6):\penalty0 595--600, apr 2018.
\newblock \doi{10.1038/s41567-018-0124-x}.
\newblock URL \url{https://doi.org/10.1038%2Fs41567-018-0124-x}.

\bibitem[Chen et~al.(2022)Chen, Ding, Huang, and Kong]{linearxeb}
Jianxin Chen, Dawei Ding, Cupjin Huang, and Linghang Kong.
\newblock Linear cross entropy benchmarking with clifford circuits, 2022.
\newblock URL \url{https://arxiv.org/abs/2206.08293}.

\bibitem[Proctor et~al.(2021)Proctor, Rudinger, Young, Nielsen, and
  Blume-Kohout]{Proctor2020_capabilities}
Timothy Proctor, Kenneth Rudinger, Kevin Young, Erik Nielsen, and Robin
  Blume-Kohout.
\newblock Measuring the capabilities of quantum computers.
\newblock \emph{Nature Physics}, 18\penalty0 (1):\penalty0 75--79, dec 2021.
\newblock \doi{10.1038/s41567-021-01409-7}.
\newblock URL \url{https://doi.org/10.1038%2Fs41567-021-01409-7}.

\bibitem[{Blume-Kohout} and {Young}(2020)]{RBK2020_volumetric}
Robin {Blume-Kohout} and Kevin~C. {Young}.
\newblock {A volumetric framework for quantum computer benchmarks}.
\newblock \emph{Quantum}, 4:\penalty0 362, November 2020.
\newblock \doi{10.22331/q-2020-11-15-362}.

\bibitem[Ferracin et~al.(2019)Ferracin, Kapourniotis, and
  Datta]{Ferracin_2019_accreditation}
Samuele Ferracin, Theodoros Kapourniotis, and Animesh Datta.
\newblock Accrediting outputs of noisy intermediate-scale quantum computing
  devices.
\newblock \emph{New Journal of Physics}, 21\penalty0 (11):\penalty0 113038, nov
  2019.
\newblock \doi{10.1088/1367-2630/ab4fd6}.
\newblock URL \url{https://dx.doi.org/10.1088/1367-2630/ab4fd6}.

\bibitem[Ferracin et~al.(2021)Ferracin, Merkel, McKay, and
  Datta]{Ferracin2021accreditation}
Samuele Ferracin, Seth~T. Merkel, David McKay, and Animesh Datta.
\newblock Experimental accreditation of outputs of noisy quantum computers.
\newblock \emph{Phys. Rev. A}, 104:\penalty0 042603, Oct 2021.
\newblock \doi{10.1103/PhysRevA.104.042603}.
\newblock URL \url{https://link.aps.org/doi/10.1103/PhysRevA.104.042603}.

\bibitem[{Flammia}(2021)]{FlammiaACES}
Steven~T. {Flammia}.
\newblock {Averaged circuit eigenvalue sampling}.
\newblock \emph{arXiv e-prints}, art. arXiv:2108.05803, August 2021.

\bibitem[Merkel et~al.(2013)Merkel, Gambetta, Smolin, Poletto, C\'{o}rcoles,
  Johnson, Ryan, and Steffen]{Merkel2013}
Seth~T. Merkel, Jay~M. Gambetta, John~a. Smolin, Stefano Poletto, Antonio~D.
  C\'{o}rcoles, Blake~R. Johnson, Colm~a. Ryan, and Matthias Steffen.
\newblock {Self-consistent quantum process tomography}.
\newblock \emph{Physical Review A}, 87\penalty0 (6):\penalty0 062119, June
  2013.
\newblock ISSN 1050-2947.
\newblock \doi{10.1103/PhysRevA.87.062119}.
\newblock URL \url{http://link.aps.org/doi/10.1103/PhysRevA.87.062119}.

\bibitem[{Blume-Kohout} et~al.(2013){Blume-Kohout}, {King Gamble}, {Nielsen},
  {Mizrahi}, {Sterk}, and {Maunz}]{Blume-Kohout2013}
Robin {Blume-Kohout}, John {King Gamble}, Erik {Nielsen}, Jonathan {Mizrahi},
  Jonathan~D. {Sterk}, and Peter {Maunz}.
\newblock {Robust, self-consistent, closed-form tomography of quantum logic
  gates on a trapped ion qubit}.
\newblock \emph{ArXiv e-prints}, art. arXiv:1310.4492, October 2013.

\bibitem[Blume-Kohout et~al.(2017)Blume-Kohout, Gamble, Nielsen, Rudinger,
  Mizrahi, Fortier, and Maunz]{Blume_Kohout_2017}
Robin Blume-Kohout, John~King Gamble, Erik Nielsen, Kenneth Rudinger, Jonathan
  Mizrahi, Kevin Fortier, and Peter Maunz.
\newblock Demonstration of qubit operations below a rigorous fault tolerance
  threshold with gate set tomography.
\newblock \emph{Nature Communications}, 8\penalty0 (1), feb 2017.
\newblock \doi{10.1038/ncomms14485}.
\newblock URL \url{https://doi.org/10.1038%2Fncomms14485}.

\bibitem[Greenbaum(2015)]{Greenbaum2015}
Daniel Greenbaum.
\newblock Introduction to quantum gate set tomography.
\newblock 2015.
\newblock \doi{10.48550/ARXIV.1509.02921}.
\newblock URL \url{https://arxiv.org/abs/1509.02921}.

\bibitem[Nielsen et~al.(2021)Nielsen, Gamble, Rudinger, Scholten, Young, and
  Blume-Kohout]{Nielsen_2021}
Erik Nielsen, John~King Gamble, Kenneth Rudinger, Travis Scholten, Kevin Young,
  and Robin Blume-Kohout.
\newblock Gate set tomography.
\newblock \emph{Quantum}, 5:\penalty0 557, oct 2021.
\newblock \doi{10.22331/q-2021-10-05-557}.
\newblock URL \url{https://doi.org/10.22331%2Fq-2021-10-05-557}.

\bibitem[Brieger et~al.(2021)Brieger, Roth, and Kliesch]{Brieger2021}
Raphael Brieger, Ingo Roth, and Martin Kliesch.
\newblock Compressive gate set tomography.
\newblock 2021.
\newblock \doi{10.48550/ARXIV.2112.05176}.
\newblock URL \url{https://arxiv.org/abs/2112.05176}.

\bibitem[Rudinger et~al.(2021)Rudinger, Hogle, Naik, Hashim, Lobser, Santiago,
  Grace, Nielsen, Proctor, Seritan, Clark, Blume-Kohout, Siddiqi, and
  Young]{Rudinger_2021_crosstalk}
Kenneth Rudinger, Craig~W. Hogle, Ravi~K. Naik, Akel Hashim, Daniel Lobser,
  David~I. Santiago, Matthew~D. Grace, Erik Nielsen, Timothy Proctor, Stefan
  Seritan, Susan~M. Clark, Robin Blume-Kohout, Irfan Siddiqi, and Kevin~C.
  Young.
\newblock Experimental characterization of crosstalk errors with simultaneous
  gate set tomography.
\newblock \emph{{PRX} Quantum}, 2\penalty0 (4), nov 2021.
\newblock \doi{10.1103/prxquantum.2.040338}.
\newblock URL \url{https://doi.org/10.1103%2Fprxquantum.2.040338}.

\bibitem[Wallman and Emerson(2016)]{Wallman2016}
Joel~J. Wallman and Joseph Emerson.
\newblock {Noise tailoring for scalable quantum computation via randomized
  compiling}.
\newblock \emph{Physical Review A}, 94\penalty0 (5):\penalty0 052325, nov 2016.
\newblock ISSN 2469-9926.
\newblock \doi{10.1103/PhysRevA.94.052325}.
\newblock URL \url{http://link.aps.org/doi/10.1103/PhysRevA.94.052325}.

\bibitem[Hashim et~al.(2021{\natexlab{a}})Hashim, Naik, Morvan, Ville,
  Mitchell, Kreikebaum, Davis, Smith, Iancu, O'Brien, Hincks, Wallman, Emerson,
  and Siddiqi]{Hashim2021}
Akel Hashim, Ravi~K. Naik, Alexis Morvan, Jean-Loup Ville, Bradley Mitchell,
  John~Mark Kreikebaum, Marc Davis, Ethan Smith, Costin Iancu, Kevin~P.
  O'Brien, Ian Hincks, Joel~J. Wallman, Joseph Emerson, and Irfan Siddiqi.
\newblock Randomized compiling for scalable quantum computing on a noisy
  superconducting quantum processor.
\newblock \emph{Phys. Rev. X}, 11:\penalty0 041039, Nov 2021{\natexlab{a}}.
\newblock \doi{10.1103/PhysRevX.11.041039}.
\newblock URL \url{https://link.aps.org/doi/10.1103/PhysRevX.11.041039}.

\bibitem[Ville et~al.(2022)Ville, Morvan, Hashim, Naik, Lu, Mitchell,
  Kreikebaum, O'Brien, Wallman, Hincks, Emerson, Smith, Younis, Iancu,
  Santiago, and Siddiqi]{ville2021}
Jean-Loup Ville, Alexis Morvan, Akel Hashim, Ravi~K. Naik, Marie Lu, Bradley
  Mitchell, John-Mark Kreikebaum, Kevin~P. O'Brien, Joel~J. Wallman, Ian
  Hincks, Joseph Emerson, Ethan Smith, Ed~Younis, Costin Iancu, David~I.
  Santiago, and Irfan Siddiqi.
\newblock Leveraging randomized compiling for the quantum
  imaginary-time-evolution algorithm.
\newblock \emph{Physical Review Research}, 4\penalty0 (3), aug 2022.
\newblock \doi{10.1103/physrevresearch.4.033140}.
\newblock URL \url{https://doi.org/10.1103\%2Fphysrevresearch.4.033140}.

\bibitem[Gu et~al.(2022)Gu, Ma, Forcellini, and Liu]{Gu2022}
Yanwu Gu, Yunheng Ma, Nicolo Forcellini, and Dong~E. Liu.
\newblock Noise-resilient phase estimation with randomized compiling.
\newblock 2022.
\newblock \doi{10.48550/ARXIV.2208.04100}.
\newblock URL \url{https://arxiv.org/abs/2208.04100}.

\bibitem[Arute et~al.(2019)Arute, Arya, Babbush, Bacon, Bardin, Barends,
  Biswas, Boixo, Brandao, Buell, Burkett, Chen, Chen, Chiaro, Collins,
  Courtney, Dunsworth, Farhi, Foxen, Fowler, Gidney, Giustina, Graff, Guerin,
  Habegger, Harrigan, Hartmann, Ho, Hoffmann, Huang, Humble, Isakov, Jeffrey,
  Jiang, Kafri, Kechedzhi, Kelly, Klimov, Knysh, Korotkov, Kostritsa, Landhuis,
  Lindmark, Lucero, Lyakh, Mandrà, McClean, McEwen, Megrant, Mi, Michielsen,
  Mohseni, Mutus, Naaman, Neeley, Neill, Niu, Ostby, Petukhov, Platt, Quintana,
  Rieffel, Roushan, Rubin, Sank, Satzinger, Smelyanskiy, Sung, Trevithick,
  Vainsencher, Villalonga, White, Yao, Yeh, Zalcman, Neven, and
  Martinis]{supremacy2019}
Frank Arute, Kunal Arya, Ryan Babbush, Dave Bacon, Joseph Bardin, Rami Barends,
  Rupak Biswas, Sergio Boixo, Fernando Brandao, David Buell, Brian Burkett,
  Yu~Chen, Jimmy Chen, Ben Chiaro, Roberto Collins, William Courtney, Andrew
  Dunsworth, Edward Farhi, Brooks Foxen, Austin Fowler, Craig~Michael Gidney,
  Marissa Giustina, Rob Graff, Keith Guerin, Steve Habegger, Matthew Harrigan,
  Michael Hartmann, Alan Ho, Markus~Rudolf Hoffmann, Trent Huang, Travis
  Humble, Sergei Isakov, Evan Jeffrey, Zhang Jiang, Dvir Kafri, Kostyantyn
  Kechedzhi, Julian Kelly, Paul Klimov, Sergey Knysh, Alexander Korotkov, Fedor
  Kostritsa, Dave Landhuis, Mike Lindmark, Erik Lucero, Dmitry Lyakh, Salvatore
  Mandrà, Jarrod~Ryan McClean, Matthew McEwen, Anthony Megrant, Xiao Mi,
  Kristel Michielsen, Masoud Mohseni, Josh Mutus, Ofer Naaman, Matthew Neeley,
  Charles Neill, Murphy~Yuezhen Niu, Eric Ostby, Andre Petukhov, John Platt,
  Chris Quintana, Eleanor~G. Rieffel, Pedram Roushan, Nicholas Rubin, Daniel
  Sank, Kevin~J. Satzinger, Vadim Smelyanskiy, Kevin~Jeffery Sung, Matt
  Trevithick, Amit Vainsencher, Benjamin Villalonga, Ted White, Z.~Jamie Yao,
  Ping Yeh, Adam Zalcman, Hartmut Neven, and John Martinis.
\newblock Quantum supremacy using a programmable superconducting processor.
\newblock \emph{Nature}, 574:\penalty0 505–510, 2019.
\newblock URL \url{https://www.nature.com/articles/s41586-019-1666-5}.

\bibitem[Flammia and Wallman(2020)]{FW2019}
Steven~T. Flammia and Joel~J. Wallman.
\newblock Efficient estimation of pauli channels.
\newblock \emph{{ACM} Transactions on Quantum Computing}, 1\penalty0
  (1):\penalty0 1--32, dec 2020.
\newblock \doi{10.1145/3408039}.
\newblock URL \url{https://doi.org/10.1145%2F3408039}.

\bibitem[Helsen et~al.(2019{\natexlab{b}})Helsen, Wallman, Flammia, and
  Wehner]{HWFW2019}
Jonas Helsen, Joel~J. Wallman, Steven~T. Flammia, and Stephanie Wehner.
\newblock Multiqubit randomized benchmarking using few samples.
\newblock \emph{Phys. Rev. A}, 100:\penalty0 032304, Sep 2019{\natexlab{b}}.
\newblock \doi{10.1103/PhysRevA.100.032304}.
\newblock URL \url{https://link.aps.org/doi/10.1103/PhysRevA.100.032304}.

\bibitem[Temme et~al.(2017)Temme, Bravyi, and Gambetta]{Temme_2017}
Kristan Temme, Sergey Bravyi, and Jay~M. Gambetta.
\newblock Error mitigation for short-depth quantum circuits.
\newblock \emph{Physical Review Letters}, 119\penalty0 (18), nov 2017.
\newblock \doi{10.1103/physrevlett.119.180509}.
\newblock URL \url{https://doi.org/10.1103%2Fphysrevlett.119.180509}.

\bibitem[Berg et~al.(2022)Berg, Minev, Kandala, and Temme]{Berg2022}
Ewout van~den Berg, Zlatko~K. Minev, Abhinav Kandala, and Kristan Temme.
\newblock Probabilistic error cancellation with sparse pauli-lindblad models on
  noisy quantum processors.
\newblock 2022.
\newblock \doi{10.48550/ARXIV.2201.09866}.
\newblock URL \url{https://arxiv.org/abs/2201.09866}.

\bibitem[Ferracin et~al.(2022)Ferracin, Hashim, Ville, Naik, Carignan-Dugas,
  Qassim, Morvan, Santiago, Siddiqi, and Wallman]{Ferracin2022}
Samuele Ferracin, Akel Hashim, Jean-Loup Ville, Ravi Naik, Arnaud
  Carignan-Dugas, Hammam Qassim, Alexis Morvan, David~I. Santiago, Irfan
  Siddiqi, and Joel~J. Wallman.
\newblock Efficiently improving the performance of noisy quantum computers,
  2022.
\newblock URL \url{https://arxiv.org/abs/2201.10672}.

\bibitem[{Kelly} et~al.(2014){Kelly}, {Barends}, {Campbell}, {Chen}, {Chen},
  {Chiaro}, {Dunsworth}, {Fowler}, {Hoi}, {Jeffrey}, {Megrant}, {Mutus},
  {Neill}, {O'Malley}, {Quintana}, {Roushan}, {Sank}, {Vainsencher}, {Wenner},
  {White}, {Cleland}, and {Martinis}]{ORBIT2014}
J.~{Kelly}, R.~{Barends}, B.~{Campbell}, Y.~{Chen}, Z.~{Chen}, B.~{Chiaro},
  A.~{Dunsworth}, A.~G. {Fowler}, I.~C. {Hoi}, E.~{Jeffrey}, A.~{Megrant},
  J.~{Mutus}, C.~{Neill}, P.~J.~J. {O'Malley}, C.~{Quintana}, P.~{Roushan},
  D.~{Sank}, A.~{Vainsencher}, J.~{Wenner}, T.~C. {White}, A.~N. {Cleland}, and
  John~M. {Martinis}.
\newblock {Optimal Quantum Control Using Randomized Benchmarking}.
\newblock \emph{Phys. Rev. Lett.}, 112\penalty0 (24):\penalty0 240504, June
  2014.
\newblock \doi{10.1103/PhysRevLett.112.240504}.

\bibitem[{Carignan-Dugas} et~al.(2019){Carignan-Dugas}, {Wallman}, and
  {Emerson}]{Dugas2019unitarity}
Arnaud {Carignan-Dugas}, Joel~J. {Wallman}, and Joseph {Emerson}.
\newblock {Bounding the average gate fidelity of composite channels using the
  unitarity}.
\newblock \emph{New Journal of Physics}, 21\penalty0 (5):\penalty0 053016, May
  2019.
\newblock \doi{10.1088/1367-2630/ab1800}.

\bibitem[Wallman(U.S. patent 10838792, Nov. 2020)]{knrpatent2018}
Joel~J. Wallman.
\newblock Systems and methods for reconstructing noise from pauli fidelities,
  November~17 U.S. patent 10838792, Nov. 2020.
\newblock URL \url{https://patents.justia.com/patent/10838792}.

\bibitem[Beale et~al.(2020)Beale, Boone, Carignan-Dugas, Chytros, Dahlen,
  Dawkins, Emerson, Ferracin, Frey, Hincks, Hufnagel, Iyer, Jain, Kolbush,
  Ospadov, Pino, Qassim, Saunders, Skanes-Norman, Stasiuk, Wallman, Winick, and
  Wright]{trueq}
Stefanie~J. Beale, Kristine Boone, Arnaud Carignan-Dugas, Anthony Chytros, Dar
  Dahlen, Hillary Dawkins, Joseph Emerson, Samuele Ferracin, Virginia Frey, Ian
  Hincks, David Hufnagel, Pavithran Iyer, Aditya Jain, Jason Kolbush, Egor
  Ospadov, José~Luis Pino, Hammam Qassim, Jordan Saunders, Joshua
  Skanes-Norman, Andrew Stasiuk, Joel~J. Wallman, Adam Winick, and Emily
  Wright.
\newblock True-q, June 2020.
\newblock URL \url{https://doi.org/10.5281/zenodo.3945250}.

\bibitem[Hashim et~al.(2021{\natexlab{b}})Hashim, Naik, Morvan, Ville,
  Mitchell, Kreikebaum, Davis, Smith, Iancu, O'Brien, Hincks, Wallman, Emerson,
  and Siddiqi]{Hashim_2021}
Akel Hashim, Ravi~K. Naik, Alexis Morvan, Jean-Loup Ville, Bradley Mitchell,
  John~Mark Kreikebaum, Marc Davis, Ethan Smith, Costin Iancu, Kevin~P.
  O'Brien, Ian Hincks, Joel~J. Wallman, Joseph Emerson, and Irfan Siddiqi.
\newblock Randomized compiling for scalable quantum computing on a noisy
  superconducting quantum processor.
\newblock \emph{Physical Review X}, 11\penalty0 (4), nov 2021{\natexlab{b}}.
\newblock \doi{10.1103/physrevx.11.041039}.
\newblock URL \url{https://doi.org/10.1103%2Fphysrevx.11.041039}.

\bibitem[{Harper} et~al.(2020){Harper}, {Flammia}, and {Wallman}]{HFW2020}
Robin {Harper}, Steven~T. {Flammia}, and Joel~J. {Wallman}.
\newblock {Efficient learning of quantum noise}.
\newblock \emph{Nature Physics}, 16\penalty0 (12):\penalty0 1184--1188, January
  2020.
\newblock \doi{10.1038/s41567-020-0992-8}.

\bibitem[{Helsen} et~al.(2021){Helsen}, {Ioannou}, {Kitzinger}, {Onorati},
  {Werner}, {Eisert}, and {Roth}]{helsen2021estimating}
J.~{Helsen}, M.~{Ioannou}, J.~{Kitzinger}, E.~{Onorati}, A.~H. {Werner},
  J.~{Eisert}, and I.~{Roth}.
\newblock {Estimating gate-set properties from random sequences}.
\newblock \emph{arXiv e-prints}, art. arXiv:2110.13178, October 2021.
\newblock \doi{10.48550/arXiv.2110.13178}.

\bibitem[Goss et~al.(2022)Goss, Morvan, Marinelli, Mitchell, Nguyen, Naik,
  Chen, Jünger, Kreikebaum, Santiago, Wallman, and Siddiqi]{Goss2022}
Noah Goss, Alexis Morvan, Brian Marinelli, Bradley~K. Mitchell, Long~B. Nguyen,
  Ravi~K. Naik, Larry Chen, Christian Jünger, John~Mark Kreikebaum, David~I.
  Santiago, Joel~J. Wallman, and Irfan Siddiqi.
\newblock High-fidelity qutrit entangling gates for superconducting circuits.
\newblock \emph{Nature Communications}, 13\penalty0 (1), dec 2022.
\newblock \doi{10.1038/s41467-022-34851-z}.
\newblock URL \url{https://doi.org/10.1038%2Fs41467-022-34851-z}.

\bibitem[Emerson et~al.(2007)Emerson, Silva, Moussa, Ryan, Laforest, Baugh,
  Cory, and Laflamme]{Emerson2007}
Joseph Emerson, Marcus Silva, Osama Moussa, Colm~A. Ryan, Martin Laforest,
  Jonathan Baugh, David~G Cory, and Raymond Laflamme.
\newblock {Symmetrized characterization of noisy quantum processes.}
\newblock \emph{Science}, 317\penalty0 (5846):\penalty0 1893--6, September
  2007.
\newblock ISSN 1095-9203.
\newblock \doi{10.1126/science.1145699}.
\newblock URL \url{http://www.ncbi.nlm.nih.gov/pubmed/17901327}.

\bibitem[Steinberg(2012)]{Steinberg2012}
Benjamin Steinberg.
\newblock \emph{Character Theory and the Orthogonality Relations}.
\newblock Springer New York, New York, NY, 2012.
\newblock ISBN 978-1-4614-0776-8.
\newblock \doi{10.1007/978-1-4614-0776-8_4}.
\newblock URL \url{https://doi.org/10.1007/978-1-4614-0776-8_4}.

\bibitem[Lin et~al.(2021)Lin, Wallman, Hincks, and Laflamme]{Lin_2021}
Junan Lin, Joel~J. Wallman, Ian Hincks, and Raymond Laflamme.
\newblock Independent state and measurement characterization for quantum
  computers.
\newblock \emph{Physical Review Research}, 3\penalty0 (3), sep 2021.
\newblock \doi{10.1103/physrevresearch.3.033285}.
\newblock URL \url{https://doi.org/10.1103%2Fphysrevresearch.3.033285}.

\bibitem[Huang et~al.(2022)Huang, Flammia, and Preskill]{Huang2022}
Hsin-Yuan Huang, Steven~T. Flammia, and John Preskill.
\newblock Foundations for learning from noisy quantum experiments.
\newblock 2022.
\newblock \doi{10.48550/ARXIV.2204.13691}.
\newblock URL \url{https://arxiv.org/abs/2204.13691}.

\bibitem[Chen et~al.(2023)Chen, Liu, Otten, Seif, Fefferman, and
  Jiang]{Chen2022}
Senrui Chen, Yunchao Liu, Matthew Otten, Alireza Seif, Bill Fefferman, and
  Liang Jiang.
\newblock The learnability of pauli noise.
\newblock \emph{Nature Communications}, 14\penalty0 (1), jan 2023.
\newblock \doi{10.1038/s41467-022-35759-4}.
\newblock URL \url{https://doi.org/10.1038%2Fs41467-022-35759-4}.

\bibitem[Carignan-Dugas et~al.(2023)Carignan-Dugas, Ranu, and
  Dreher]{carignandugas2023estimating}
Arnaud Carignan-Dugas, Shashank~Kumar Ranu, and Patrick Dreher.
\newblock Estimating coherent contributions to the error profile using cycle
  error reconstruction.
\newblock 2023.

\bibitem[Wagner et~al.(2022)Wagner, Kampermann, Bru{\ss}, and
  Kliesch]{Wagner2022}
Thomas Wagner, Hermann Kampermann, Dagmar Bru{\ss}, and Martin Kliesch.
\newblock Pauli channels can be estimated from syndrome measurements in quantum
  error correction.
\newblock \emph{Quantum}, 6:\penalty0 809, sep 2022.
\newblock \doi{10.22331/q-2022-09-19-809}.
\newblock URL \url{https://doi.org/10.22331%2Fq-2022-09-19-809}.

\bibitem[Wootton(2022)]{Wootton2022}
James~R. Wootton.
\newblock Syndrome-derived error rates as a benchmark of quantum hardware.
\newblock 2022.
\newblock \doi{10.48550/ARXIV.2207.00553}.
\newblock URL \url{https://arxiv.org/abs/2207.00553}.

\bibitem[Hashim et~al.(2022)Hashim, Seritan, Proctor, Rudinger, Goss, Naik,
  Kreikebaum, Santiago, and Siddiqi]{Hashim2022}
Akel Hashim, Stefan Seritan, Timothy Proctor, Kenneth Rudinger, Noah Goss,
  Ravi~K. Naik, John~Mark Kreikebaum, David~I. Santiago, and Irfan Siddiqi.
\newblock Benchmarking quantum logic operations for achieving fault tolerance,
  2022.
\newblock URL \url{https://arxiv.org/abs/2207.08786}.

\bibitem[Piveteau et~al.(2021)Piveteau, Sutter, Bravyi, Gambetta, and
  Temme]{Piveteau2021}
Christophe Piveteau, David Sutter, Sergey Bravyi, Jay~M. Gambetta, and Kristan
  Temme.
\newblock Error mitigation for universal gates on encoded qubits.
\newblock \emph{Physical Review Letters}, 127\penalty0 (20), nov 2021.
\newblock \doi{10.1103/physrevlett.127.200505}.
\newblock URL \url{https://doi.org/10.1103%2Fphysrevlett.127.200505}.

\end{thebibliography}

\appendix

\section{CER heatmaps}

\begin{figure}[!h]
\includegraphics[width=15cm]{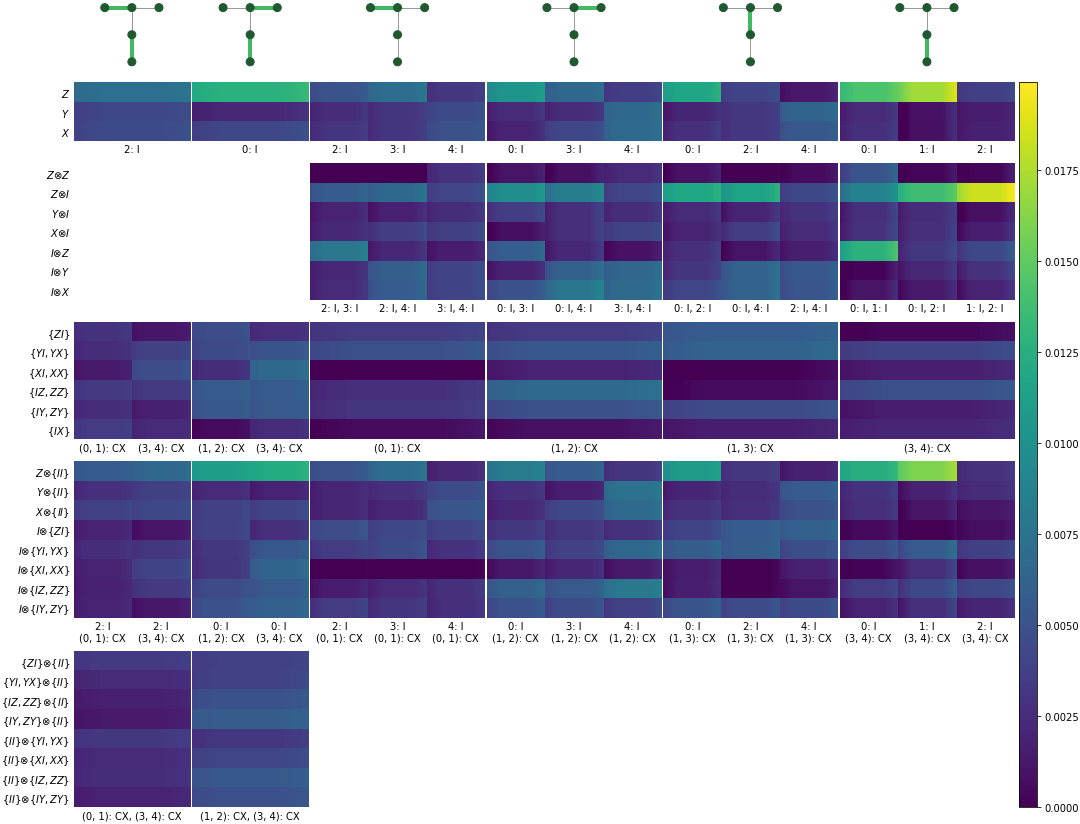}
\caption{Cycle Error Reconstruction ($k \in \{1,2\}$) heatmap from the \texttt{ibmq\_essex} chip. Each $T$-shaped graph indicates a different hard cycle; the bold green lines indicate the support of the entangling gates. \Cref{sec:read} provides the explanations for interpreting this heatmap.} \label{fig:Essex}
\end{figure}
\clearpage

\begin{figure}[!h]
\includegraphics[width=15cm]{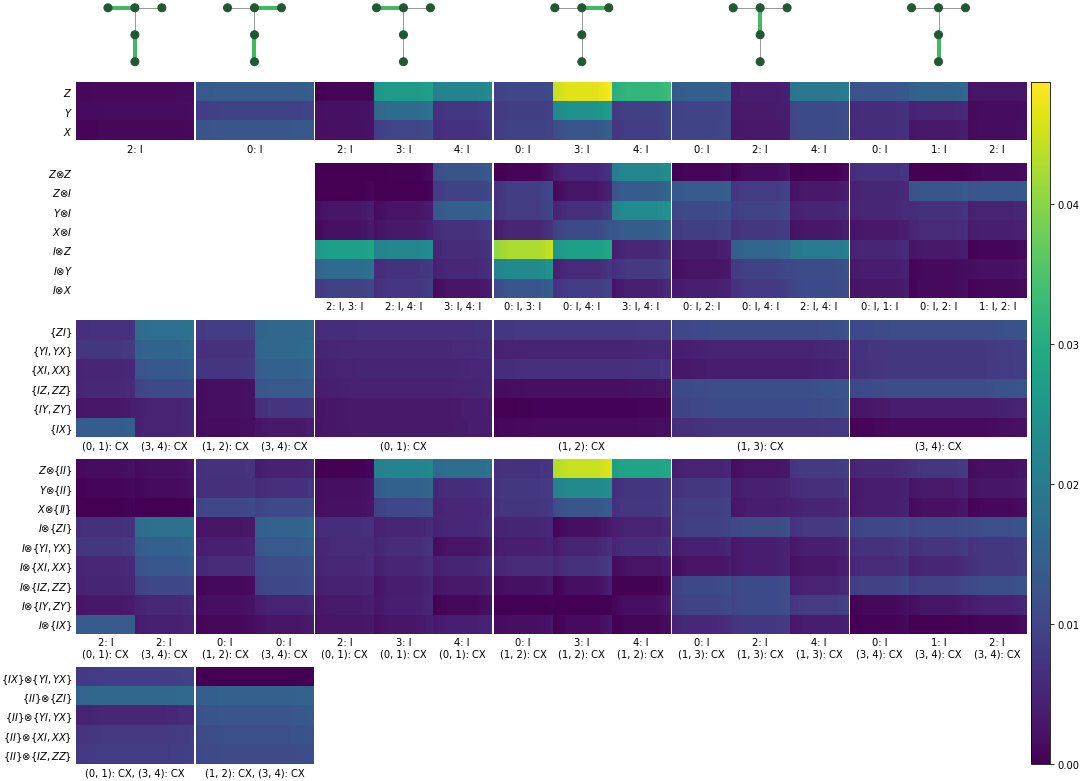}
\caption{Cycle Error Reconstruction ($k \in \{1,2\}$) heatmap from the \texttt{ibmq\_ourense} chip. Each $T$-shaped graph indicates a different hard cycle; the bold green lines indicate the support of the entangling gates. \Cref{sec:read} provides the explanations for interpreting this heatmap.} \label{fig:Ourense}
\end{figure}
\clearpage

\begin{figure}[!h]
\includegraphics[width=15cm]{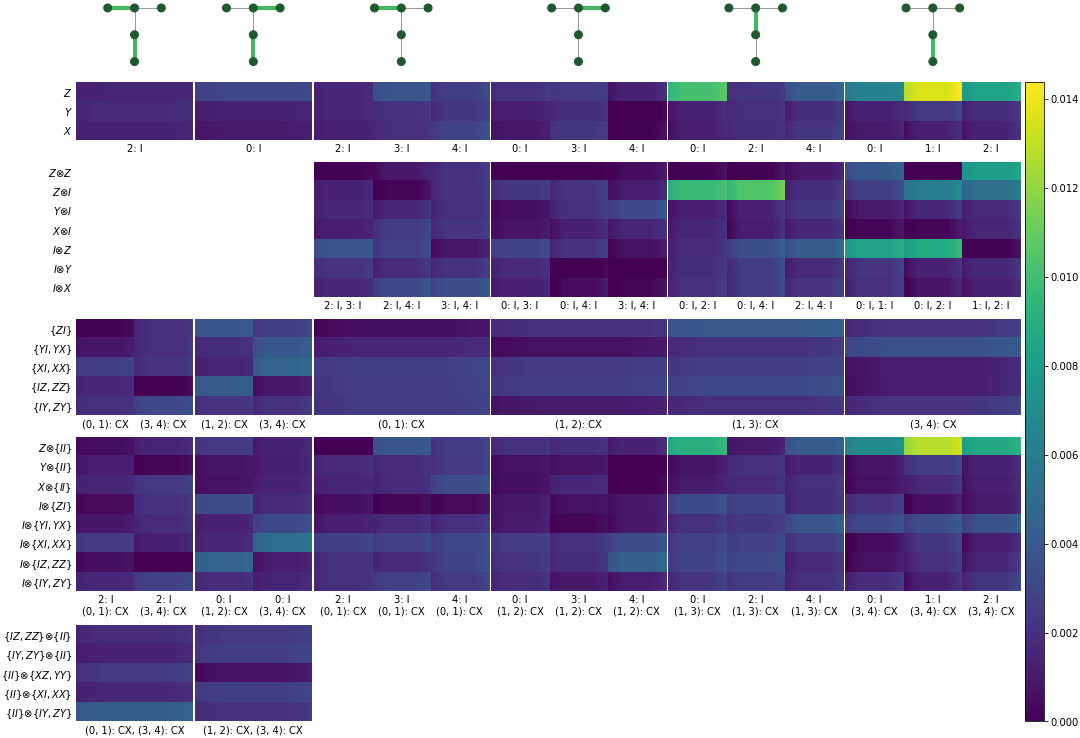}
\caption{Cycle Error Reconstruction ($k \in \{1,2\}$) heatmap from the \texttt{ibmq\_vigo} chip. Each $T$-shaped graph indicates a different hard cycle; the bold green lines indicate the support of the entangling gates. \Cref{sec:read} provides the explanations for interpreting this heatmap.}\label{fig:Vigo}
\end{figure}

\section{Randomized compiling and effective dressed cycles: proof of \cref{lem:eff}}\label{app:rc}

In this section we shall demonstrate \cref{lem:eff}, which states that the average RC circuit can be closely approximated
by a sequence of effective dressed cycles.
For conciseness, for any set of indexed variables $x_a, x_{a-1}, \ldots, x_b$ that can be multiplied (e.g., square matrices), we use the shorthand
\begin{align}
    x_{a:b} = 
    \begin{cases} 
    x_a x_{a-1} \ldots x_b & a \geq b \\
    1 & \mbox{otherwise}
    \end{cases}
\end{align}
to denote the ordered product of the variables. For conciseness, we will 
replace $\adj(H_i ,E_i|T_i, T_{i-1} )$ with $\nu^{\rm adj.}_i$ and $\eff(H_i, E_i) $ with $\nu^{\rm eff.}_i$.
Defining $\delta_i = \nu^{\rm adj.}_i - \nu^{\rm eff.}_i$ and using the linearity of of expectation values, we have
\begin{align}
      \left\langle \mc C_{\rm RC}(\vec{T}) \right\rangle_{\vec{T}} 
    &= \langle \nu^{\rm adj.}_{m:0}\rangle \notag \\
    &= \nu^{\rm eff.}_m \langle \nu^{\rm adj.}_{m-1:0}\rangle+ \langle \delta_m\nu^{\rm adj.}_{m-1:0}\rangle \notag \\
    &=  \nu^{\rm eff.}_m \langle \nu^{\rm adj.}_{m-1:0}\rangle+ \langle \delta_m  \nu^{\rm eff.}_{m-1} \nu^{\rm adj.}_{m-2:0}\rangle+ \langle \delta_{m} \delta_{m-1} \nu^{\rm adj.}_{m-2:0}\rangle
\end{align}
Now note that the only term in $\delta_m  \nu^{\rm eff.}_{m-1} \nu^{\rm adj.}_{m-2:0}$ that depends on $T_{m-1}$ is $\delta_m$ and so
\begin{align}
    \langle \delta_m  \nu^{\rm eff.}_{m-1} \nu^{\rm adj.}_{m-2:0}\rangle = 0,
\end{align}
that is,
\begin{align}
      \left\langle \mc C_{\rm RC}(\vec{T}) \right\rangle_{\vec{T}} 
    &=  \nu^{\rm eff.}_m \langle \nu^{\rm adj.}_{m-1:0}\rangle+ \langle \delta_{m} \delta_{m-1} \nu^{\rm adj.}_{m-2:0}\rangle \label{eq:avg_cir_a}
\end{align}
The same argument establishes that for any $i \in (0, m]$ we have
\begin{align}
    \langle \nu^{\rm adj.}_{i:0}\rangle &= \nu^{\rm eff.}_i \langle \nu^{\rm adj.}_{i-1:0}\rangle +
    \langle \delta_i \nu^{\rm adj.}_{i-1:0}\rangle 
   \notag\\
    &= 
    \nu^{\rm eff.}_{i}\langle \nu^{\rm adj.}_{i-1:0}\rangle +
    \langle \delta_i \nu^{\rm eff.}_{i-1} \nu^{\rm adj.}_{i-2:0}\rangle +
    \langle \delta_i \delta_{i-1} \nu^{\rm adj.}_{i-2:0}\rangle
 \notag\\
    &=   \nu^{\rm eff.}_{i}\langle \nu^{\rm adj.}_{i-1:0}\rangle +
    \langle \delta_j \delta_{i-1} \nu^{\rm adj.}_{i-2:0}\rangle \label{eq:mu_sub}
\end{align}
where we have used the fact that the only factor in the first order term that depends on $T_{i-1}$ is $\delta_i$ so that the first order term disappears on average.

By repetitively substituting \cref{eq:mu_sub} in \cref{eq:avg_cir_a}, the process applied averaged over all randomizations can be expressed as
\begin{align} 
   \left\langle \mc C_{\rm RC}(\vec{T}) \right\rangle_{\vec{T}} &=  \langle \nu^{\rm adj.}_{m:0}\rangle \notag \\
    &= \nu^{\rm eff.}_{m:0}+ \sum_{i = 1}^m \nu^{\rm adj.}_{m:i+1} \langle \delta_i \delta_{i-1} \nu^{\rm adj.}_{i-2:0}\rangle.
\end{align}
Consequently, the effective dressed cycle describes the effective error process in a randomly compiled circuit with some quadratic (i.e., second order) corrections.
These quadratic corrections are based on the dressed cycle.
We now show how the corrections can be rephrased in terms of perturbations from ``gate-independent'' noise on the easy cycles, that is, for
\begin{align}
\nrep{E_i} = \Lambda  \phi(E_i) + \epsilon(E_i)
\end{align}
where $\Lambda$ is a fixed error channel, and where $\epsilon$ is small for all easy cycles $E_i$.
When $\epsilon(E_i) = 0$, $\delta_i$ depends only on $T_{i}$ and all other terms in $\langle \delta_i \delta_{i-1}\nu^{\rm adj.}_{i-2:0} \rangle$ are independent of $T_i$ and so $\langle \delta_i \delta_{i-1}\nu^{\rm adj.}_{i-2:0} \rangle = 0$ for all $i$, that is, all the perturbation terms exactly vanish.
This matches the result from \cite{Wallman2016} for gate-independent noise.

\section{Effective error channels under RC: proof of \cref{lem:pauli_channel}}\label{app:twirled_error}

In the previous section, we showed that $\langle \mc C_{\rm RC}(\vec{T}) \rangle_{\vec{T}} \approx  \nu^{\rm eff.}_{m:0}$ where $\nu^{\rm eff.}_i:=\langle \adj(H_i,E_i| T_i, T_{i-1}) \rangle_{T_i, T_{i-1}}$ (see \cref{def:edc}). We now seek to show \cref{lem:pauli_channel}, 
which states that when $\set T= \set P_n$, the effective dressed cycle takes the form
\begin{align}
    \langle \adj(H_i,E_i, \vec{T}) \rangle_{\vec{T}}
    &= \phi(H_iE_i) \mc S_i~,
\end{align}
where $\mc S_i$ is stochastic.
Let's express the $i$th effective dressed cycle as
\begin{align}
    \eff(H_i,E_i)
    &=  \phi(H_iE_i) \Big\langle \phi(E_i^{-1})\phi(Q^{-1}) \phi(H_i^{-1})\nrep{H_i} \nrep{Q E_i 
    R}\phi(R^{-1}) \Big\rangle_{Q,R \in \set P_n}~.
\end{align}
We are hence interested in the error channel described by 
\begin{align}
    \mc E_i = \Big\langle \phi(E_i^{-1})\phi(Q^{-1}) \phi(H_i^{-1})\nrep{H_i} \nrep{Q E_i 
    R}\phi(R^{-1}) \Big\rangle_{Q,R \in \set P_n}~.
\end{align}
From \cref{lem:eff}, this expression captures precisely what error model 
randomized compiling induces (and hence what Cycle Benchmarking and 
other CB-based protocols measure), even in the presence of twirling 
operations whose implementations are cycle dependent.
We can explore this formula in three stages of successively more 
general approximations.

\subsection{If all easy cycles have the same error}

This case subsumes the case of perfectly implemented easy cycles.
Here, we assume $\nrep{E}=\Lambda \phi(E)$ for any easy cycle, including those which are not in the twirling group, which allows us to effectively pretend that the CPTP error map $\Lambda$ is part of the noise of the hard cycle, causing all easy gates to be perfect. This in turn collapses the expectation over $R$. Explicitly, we have
\begin{align}
    \mc E_i &= \Big\langle \phi(E_i^{-1})\phi(Q^{-1}) \phi(H_i^{-1})\nrep{H_i} \Lambda \rep{Q E_i 
    R}\phi(R^{-1}) \Big\rangle_{Q,R \in \set P_n} \tag{$\nrep{Q E_i 
    R}=\Lambda \rep{Q E_i 
    R}$} \\
    &= \phi(E_i^{-1}) \Big[ \phi(H_i^{-1})\nrep{H_i} \Lambda\Big]^{\set P_n} \rep{ E_i 
    } \label{eq:GD_Pauli}
\end{align}
which is purely stochastic.

\subsection{If the easy cycle is Clifford}

If the twirling group is $\set {P}_n$ and $H_i$ and $E_i$ are both Clifford (which includes the common benchmarking case where $E_i=I$), then 
we can define $Q'=E_i^{-1}Q E_i$ and $R'=Q'R$ which will 
both be Paulis:
\begin{align} \label{eq:easy_cliff}
    \mc E_i &= \Big\langle \phi({Q'}^{-1}) \phi(E^{-1}_i) \phi(H_i^{-1})\nrep{H_i} \nrep{E_i 
    R'}\phi({R'}^{-1}) \phi({Q'}) \Big\rangle_{Q',R' \in \set P_n} \notag \\
    &=  \Big[ \phi(E^{-1}_i) \phi(H_i^{-1})\nrep{H_i} \Big\langle \nrep{E_i 
    R'}\phi({R'}^{-1}) \Big\rangle_{R'\in \set P_n} \Big]^{\set P_n}~,
\end{align}
which is purely stochastic.

\subsection{If the hard cycle is a round of parallel small phase rotations}
Let the hard cycle be $H_i=Z_{\theta_1} \otimes \cdots \otimes Z_{\theta_n}$ and 
let the twirling group $\set T_i$ be generated by local Paulis and $H_i^2$ (notice that $E_i=I$). Each element of $\set T_i$ can be decomposed as $RP$ where $P \in \set P_n$ and $R \in \langle H_i^2\rangle$ (here $\langle G\rangle$ is the group generated by $G$). We can define $Q'=PQ$ which will be a Pauli. 
Then,
\begin{align}
    \mc E_i &= \Big\langle \phi(P^{-1})\phi(R^{-1}) \phi(H_i^{-1})\nrep{H_i} \nrep{RPQ}\phi(Q^{-1}) \Big\rangle_{\substack{R \in \langle H_i^2\rangle \\
    P,Q \in \set P_n}}~\notag \\
    & =\Big\langle \phi(P^{-1})\phi(R^{-1}) \phi(H_i^{-1})\nrep{H_i} \nrep{RQ'}\phi({Q'}^{-1}) \phi(P^{-1}) \Big\rangle_{\substack{R \in \langle H_i^2\rangle \\
    P,Q' \in \set P_n}}\notag \\
    & = \Big[ \Big\langle \phi(R^{-1}) \phi(H_i^{-1})\nrep{H_i} \nrep{RQ'}\phi({Q'}^{-1}) \Big\rangle_{\substack{R \in \langle H_i^2\rangle \\
    Q' \in \set P_n}} \Big]^{\set P_n} ~,
\end{align}
which is stochastic.

\section{The Pauli fidelity estimation (PIE) oracle}\label{sec:pie_oracle}
In this section we provide a brief overview of the PIE oracle used to estimate the Pauli fidelities for Clifford dressed cycles. We begin by presenting the original algorithm presented in Ref. \cite{FW2019}, then we show how this algorithm may be enhanced if either state-preparation or measurement errors are noiseless.

\subsection{The PIE oracle}
There can be a certain freedom in choosing the state 
preparation and measurement strategy for the PIE protocol. 
Here, we focus on local state preparation and measurement strategies. 
Furthermore, we will only consider Clifford hard cycles for the sake of simplicity.

Before we introduce the PIE oracle, we introduce some additional notation.
For $P \in \{X,Y,Z\}$ and $\bm s =(s_0, s_1,\cdots, s_{n-1}) \in \set Z_2^{n}$:
\begin{align}
    P^{\bm s} := \prod_{i \in \set{Z}_n} P^{s_i}_{i}~.
\end{align}
Furthermore, for some qubit support $\set A$, let $\bm s(\set A) \in \set Z_2^{n}$
be a string with $i$th entry $\bm s(\set A)_i = 1$ \emph{iff} $i \in \set A$.

\begin{protocol}[Pauli infidelity estimation (PIE)] \label{proto:pie}
    \begin{itemize}
        \item[]
        \item[] {\bf Input}: A $n$-qubit Clifford hard cycle $H$; a subset of Paulis $\set S\subseteq \bb{P}_n$; two positive integers $m_1,\: m_2>m_1$ such that $H^{m_2-m_1}=I$.
        
        \item[] {\bf Output}: A set ${\rm PIE}(\set S, H)=\{\hatfidorbit{P}{H} : P \in \set S\}$ of estimates of the orbital averages $\fidorbit{P}{H}$.
        \item[] {\bf Steps}: 
            \begin{itemize}
            \item[]For $P\in\set S$:
            \begin{itemize}
                \item[1.] For $m\in\{m_1,m_2\}$, construct as described
                a circuit with $m+1$ easy cycles and $m$ hard cycles:
                    \begin{itemize}
                        \item[1.1] Choose an initial Clifford easy cycle $E_0$ such that $E_0[P]= Z^{\bm s (\set A)}$ where $\set A$
                        is the support of $P$
                         (e.g. if $P=ZXX_{\{0,2,5\}}$, we could pick $E_0={\rm Had} \otimes {\rm Had}_{\{2,5\}}$). 
                         \item[1.2] Let $Q:=\phi(H^m)[P]$. Choose the last easy cycle such that $E_{m}[Q]= Z^{\bm s (\set B)}$ where $\set B$ is the support of Q.
                        \item[1.3] Initialize a base circuit $\mc{C}= E_{m} H E_{m-1} H \cdots E_1 H E_0$ where
                        $E_i=I$ for $i\notin\{0,m\}$, and randomly compile Paulis in the easy cycles:
                        \begin{align}
                            E_i \to T_{i}^cE_iT_{i-1}=:E_i'
                        \end{align}
                        where $T_{-1}=I$. Let $T_m^c:=X^{\bm x}Z^{\bm z}$ for $\bm x, \bm z \in \set Z_{2}^n$,and let 
                        $\mc{C}'= E_{m}' H E_{m-1}' H \cdots E_1' H E_0'$.
                        \item[1.4] Initialize a state $\rho_0 \approx \ketbra{0^n}$, perform the randomized circuit 
                        $\mc{C}'$ and measure in the computational basis. Let the outcome be the string $\bm s \in \set Z_{2}^n$.
                        \item[1.5]
                        If  $X^{\bm x+\bm s}$ commutes with  $Z^{\bm s (\set B)}$, count +1;\\
                        If  $X^{\bm x+\bm s}$ anti-commutes with  $Z^{\bm s (\set B)}$, count $-1$.\\

             \item[1.6] Get an average count $N_{P,m}$ from many shots and many circuit randomizations.
                    \end{itemize}
                \item[2.] Calculate $\hatfidorbit{P}{H}$ as \begin{equation} \label{eq:f_hat}\hatfidorbit{P}{H}=\left(\frac{N_{P,m_2}}{N_{P,m_1}}\right)^{1/(m_2-m_1)}\:. \end{equation}
            \end{itemize}
    \end{itemize}
    \end{itemize}
\end{protocol}

\subsection{Analysis of PIE}
Let's analyze \cref{proto:pie} in more detail. First, it follows
from \cref{eq:avg_rc} that
for a fixed $P \in \set S$ and fixed $T_m^c=X^{\bm x}Z^{\bm z}$ for $\bm x, \bm z \in \set Z_{2}^n$, 
the average circuit constructed in the first step of \cref{proto:pie} takes the form 
\begin{align}
    \left\langle \mc C_{\rm RC}(\vec{T}) \right\rangle_{\vec{T}} &\approx  \phi(X^{\bm x}Z^{\bm z}E_m)\mc S_{m} \phi(H) \mc S_H \cdots \phi(H)  \phi(E_0) \mc S_0 \\
                     & =   \phi(X^{\bm x}Z^{\bm z}) \phi(E_m)\mc S_{m} \left[\phi(H) \mc S_H\right]^m  \mc S_H^{-1}  \phi(E_0) \mc S_0~, \label{avg_pie_circuit}
\end{align}
where, from \cref{eq:easy_cliff}, 
\begin{subequations}
    \begin{align}
            \mc S_0 &:=  \Big[ \phi(E^{-1}_0) \phi(H^{-1})\nrep{H} \Big\langle \nrep{E_0 
    R}\phi({R}^{-1}) \Big\rangle_{R\in \set P_n} \Big]^{\set P_n}~, \\
        \mc S_H &:= \Big[ \phi(H^{-1})\nrep{H} \Big\langle \nrep{R}\phi({R}^{-1}) \Big\rangle_{R\in \set P_n} \Big]^{\set P_n}~,\\
        \mc S_{m} &:=  \Big[ \phi(E_m^{-1}) \Big\langle \nrep{E_m 
    R}\phi({R}^{-1}) \Big\rangle_{R\in \set P_n} \Big]^{\set P_n}~. 
    \end{align}
\end{subequations}
Starting from \cref{avg_pie_circuit}, the probability of observing the outcome $\bm s \in \set Z_2^n$ is given by
\begin{align}\label{eq:probability_s}
   \Pr(\bm s| \bm x,\bm z) &=  \dbra{\bm s} \Lambda_{\rm meas}\left\langle \mc C_{\rm RC}(\vec{T}) \right\rangle_{\vec{T}}  \dket{\rho_0}  \notag \\
   &=  \dbra{\bm s} \Lambda_{\rm meas} \phi(X^{\bm x}Z^{\bm z}) \phi(E_m)\mc S_{m} \left[\phi(H) \mc S_H\right]^m  \mc S_H^{-1}  \phi(E_0) \mc S_0  \dket{\rho_0} 
\end{align}
where $\Lambda_{\rm meas}^\dagger$ is a completely-positive unital error map acting on measurement
POVMs, $\dket{\rho_0}$ is the vectorized density matrix {of the input state}, and $\dbra{\bm s}$ is the vectorized (transposed) 
measurement projector $\ketbra{\bm s}$. In the step 1.5 of \cref{proto:pie}, we attribute to the outcome event 
$\bm s$ the value $\chi_{Z^{\bm s (\set B)}}(X^{\bm s+ \bm x})$ (the character function $\chi_{P}(Q)$ is defined in \cref{def:chi}); the expectation value of $\phi_{P,m}$ is
therefore:
\begin{align}
    \mbb E [N_{P,m}] &= \frac{1}{|\set P_n|}\sum_{\bm x, \bm z} \sum_{\bm s} \chi_{Z^{\bm s (\set B)}}(X^{\bm s+ \bm x}) \Pr(\bm s| \bm x,\bm z) \notag \\
    & = \frac{1}{|\set P_n|}\sum_{\bm x, \bm z,\bm s}\chi_{Z^{\bm s (\set B)}}(X^{\bm s+ \bm x})  \dbra{\bm s} \Lambda_{\rm meas} \phi(X^{\bm x}Z^{\bm z}) \phi(E_m)\mc S_{m} \left[\phi(H) \mc S_H\right]^m  \mc S_H^{-1}  \phi(E_0) \mc S_0 \dket{\rho_0 }  \notag \\
    &=\sum_{\bm x, \bm z,\bm s'}  \chi_{Z^{\bm s (\set B)}}(X^{\bm s'}) \dbra{\bm s'}  \frac{ \phi^\dagger(X^{\bm x}Z^{\bm z})\Lambda_{\rm meas}  \phi(X^{\bm x}Z^{\bm z})}{|\set P_n|} \phi(E_m)\mc S_{m} \left[\phi(H) \mc S_H\right]^m  \mc S_H^{-1}  \phi(E_0) \mc S_0 \dket{\rho_0}  \notag \\
    &=\sum_{\bm s'}   \chi_{Z^{\bm s (\set B)}}(X^{\bm s'}) \dbra{\bm s'}  \Lambda_{\rm meas}^{\set P_n}  \phi(E_m)\mc S_{m} \left[\phi(H) \mc S_H\right]^m  \mc S_H^{-1}  \phi(E_0) \mc S_0 \dket{\rho_0} \notag \\
    &= \dbra{Z^{\bm s (\set B)}}  \Lambda_{\rm meas}^{\set P_n}  \phi(E_m)\mc S_{m} \left[\phi(H) \mc S_H\right]^m  \mc S_H^{-1}  \phi(E_0) \mc S_0 \dket{\rho_0}~, \label{eq:exp_pie}
\end{align}
where on the last line we used elementary character theory to get
\begin{align}
    \sum_{\bm s'}   \chi_{Z^{\bm s (\set B)}}(X^{\bm s'}) \dbra{\bm s'}  =  \dbra{Z^{\bm s (\set B)}}~.
\end{align}
All the channels appearing in \cref{eq:exp_pie} have a simple action on Paulis; to translate 
\cref{eq:exp_pie} in terms of scalars, we define
channel-specific Pauli fidelities
\begin{align}
     f(Q, \mc S_i):= d^{-1}\dbra{Q} \mc S_{i} \dket{Q}~,
\end{align}
where $\mc S_i$ is a Pauli stochastic channel. Moreover, let's define $b:= m \mod |\orbit{P}{H}|$.
With these definitions at hand, we have 
\begin{align}
    \mbb E [N_{P,m}] & =f({Z^{\bm s (\set B)}},\Lambda_{\rm meas}^{\set P_n}) \dbra{Z^{\bm s (\set B)}} \phi(E_m)\mc S_{m} \left[\phi(H) \mc S_H\right]^m  \mc S_H^{-1}  \phi(E_0) \mc S_0 \dket{\rho_0} \notag \\
    & =f({Z^{\bm s (\set B)}},\Lambda_{\rm meas}^{\set P_n}) \dbra{Q}\mc S_{m} \left[\phi(H) \mc S_H\right]^m  \mc S_H^{-1}  \phi(E_0) \mc S_0 \dket{\rho_0}  \tag{$E_0[P]= Z^{\bm s (\set A)}$}  \\
    & = f({Z^{\bm s (\set B)}},\Lambda_{\rm meas}^{\set P_n})  f(Q,\mc S_{m})  f\left(Q,\left[\phi(H) \mc S_H\right]^{m-b}\right)   \dbra{Q} \left[\phi(H) \mc S_H\right]^{b} \dket{P}  f(P,S_H^{-1})  \dbra{Z^{\bm s (\set A)}}  \mc S_0 \dket{\rho_0} \notag \\
    & = f({Z^{\bm s (\set B)}},\Lambda_{\rm meas}^{\set P_n})  f(Q,\mc S_{m})  f\left(P,\left[\phi(H) \mc S_H\right]^{m-b}\right) f(P, \phi(H^{-b})\left[\phi(H) \mc S_H\right]^{b})  f(P,\mc S_H^{-1})     \dbra{Z^{\bm s (\set A)}}  \mc S_0 \dket{\rho_0}~. \label{eq:pie_signal}
\end{align}
It's straightforward to see that $f\left(P,\left[\phi(H) \mc S_H\right]^{m-b}\right)$ is of the form
\begin{align}
    f\left(P,\left[\phi(H) \mc S_H\right]^{m-b}\right) &= \left(\prod_{Q \in \orbit{ P}{H}} f(Q, \mc S_H)\right)^{(m-b)/|\orbit{ P}{H}|} \notag \\
    & = f^{m-b}(\orbit{P}{H}, \mc S_H)+ {\rm h.o.}~, \label{eq:decay_value}
\end{align}
where on the last line we used the fact that the geometric and arithmetic means are essentially equal since $\{f(Q, \mc S_H)\}_{Q \in \orbit{P}{H}}$ are guaranteed to be close one another. The reason to consider the arithmetic mean instead of the geometric mean is solely for pedagogical purposes.

By substituting \cref{eq:decay_value} back in \cref{eq:pie_signal}, we get
an exponential signal $\mbb E [N_{P,m}] =A f^m$ where
\begin{subequations}
    \begin{align}\label{eq:pie_param_A}
    A &:=  f({Z^{\bm s (\set B)}},\Lambda_{\rm meas}^{\set P_n})  f(Q,\mc S_{m}) \frac{f(P, \phi(H^{-b})\left[\phi(H) \mc S_H\right]^{b})}{f^b(\orbit{P}{H}, \mc S_H)}  f(P,S_H^{-1})     \dbra{Z^{\bm s (\set A)}}  \mc S_0 \dket{\rho_0}~, 
\\
    f & := f(\orbit{P}{H}, \mc S_H)~.\label{eq:pie_param_p}
    \end{align}
\end{subequations}
The choice of $b$ indirectly affects $Q$, since $\phi(H^b)[P]=Q$. This in turns affects
the support $\set B$. Since $f$ is estimated by looking at the curve $Af^m$, it can be estimated to a precision that greatly exceeds $1/\sqrt{N_{\rm runs}}$ \cite{Harper_2019}.

\subsection{Resolving the orbits}
As seen in \cref{eq:pie_param_A,eq:pie_param_p}, by fitting the signal of ${\rm{PIE}}(P,H)$ to an exponential, 
we can get a very precise estimate of the orbital average $f({\orbit{P}{H}})$.
Here, we provide a modification of the PIE algorithm which can resolve the individual fidelities $f(P)$ appearing in the orbital average $f({\orbit{P}{H}})$ under certain conditions. 

First, we assume that the state preparation error is negligible. 
This enforces the coefficient $\dbra{Z^{\bm s (\set A)}}  \mc S_0 \dket{\rho_0}$ appearing in 
\cref{eq:pie_param_A} to take the form:
\begin{align}\label{eq:cond_1}
\end{align}
We note here that \cref{eq:cond_1} is also obtainable via a gauge transformation 
which sets the state preparation error to the identity map. More generally, 
orbits can be resolved in any gauge where the state preparation error map is 
known \cite{Lin_2021,\learnability}.

Second, we assume that the easy cycles have a gate-independent error model, $\nrep{E}= \Lambda \phi(E)$. 
This assumption can also be thought of as an approximation, and the resolution of the individual
fidelities in the orbit $\orbit{\set P}{H}$ may contain systematic inaccuracies 
that scale as the standard deviation of the easy cycles' infidelity. 
This standard deviation is expected to be small since easy cycles have higher fidelities.
As seen through \cref{eq:GD_Pauli}, 
the assumption $\nrep{E}= \Lambda \phi(E)$ simplifies $\mc S_0$ and $\mc S_m$:
\begin{subequations}
    \begin{align}
    \mc S_0 &= \mc S_H, \\
    \mc S_m & = \phi(E_{m}^{-1}) \Lambda^{\set P_n} \phi(E_{m}),
    \end{align}
\end{subequations}
and the constant $A$ in the decay function (\cref{eq:pie_param_A}) takes the form
\begin{align}\label{eq:A_simple}
    A= f({Z^{\bm s (\set B)}},\Lambda_{\rm meas}^{\set P_n})  f({Z^{\bm s (\set B)}},\Lambda^{\set P_n}) \frac{f(P, \phi(H^{-b})\left[\phi(H) \mc S_H\right]^{b})}{f^b(\orbit{P}{H}, \mc S_H)}~,
\end{align}
where $\set B$ is the support of $Q=\phi(H^b)[P]$.

We emphasize that the individual fidelities are harder to resolve
since they are not individually amplified. For this reason, 
the estimates $\hat{f}(P)$ are limited by shot noise. 
Now, consider the following protocol to estimate $f(P)$:

\begin{protocol}[PIE with orbital resolution] \label{proto:pie_orbital_res}
    \begin{itemize}
        \item[]
        \item[] {\bf Input}: A $n$-qubit Clifford hard cycle $H$; a Pauli $P \subseteq \bb{P}_n$; three positive integers $m_1,\: m_2>m_1,\: m_3>0$ such that $m_1, m_2 \mod |\orbit{P}{H}| =1$, $m_3 \mod |\orbit{P}{H}| =0$.
        
        \item[] {\bf Output}: An estimate of $f(P)$ for the effective dressed cycle $\eff(H, I) $.
        \item[] {\bf Steps}: 
            \begin{itemize}
                \item[1.] Perform ${{\rm PIE}(P,H)}$ with sequence lengths $m_1$ and $m_2$, and from the decay signal
                $Ap^m$, get the estimates $\hat{A}$ and $\hat{p}$.
               
                \item[2.] Construct as described
                a circuit with $m_3+1$ easy cycles and $m_3$ hard cycles:
                    \begin{itemize}
                        \item[2.1] Choose an initial Clifford easy cycle $E_0$ such that $E_0[Q]= Z^{\bm s (\set B)}$ where $\set B$
                        is the support of $Q=\phi(H)[P]$. 
                         \item[2.2] Choose the last easy cycle such that $E_{m_3}[Q]=\set Z^{\bm s (\set B)}$ where $\set B$ is the support of Q.
                        \item[2.3] Initialize a base circuit $\mc{C}= E_{m_3} H E_{m_3-1} H \cdots E_1 H E_0$ where
                        $E_i=I$ for $i\notin\{0,m\}$, and randomly compile Paulis in the easy cycles:
                        \begin{align}
                            E_i \to T_iE_iT_{i-1}^c=:E_i'
                        \end{align}
                        where $T_{-1}^c=I$. Let $T_{m_3}=X^{\bm x}Z^{\bm z}$ for $\bm x, \bm z \in \set Z_{2}^n$,and let 
                        $\mc{C}'= E_{m_3}' H E_{m_3-1}' H \cdots E_1' H E_0'$.
                        \item[2.4] Initialize a state $\rho_0 = \ketbra{0^n}$, perform the randomized circuit 
                        $\mc{C}'$ and measure in the computational basis. Let the outcome be the string $\bm s \in \set Z_{2}^n$.
                        \item[2.5]
                        If  $X^{\bm x+\bm s}$ commutes with  $Z^{\bm s (\set B)}$, count +1;\\
                        If  $X^{\bm x+\bm s}$ anti-commutes with  $Z^{\bm s (\set B)}$, count $-1$.
                        \item[2.6] Get an average count $N_{Q,m_3}$ from many shots and many circuit randomizations.
                    \end{itemize}
                \item[3.] Calculate $\hat{f}(P, \mc S_H)$ as
                \begin{equation} 
                    \hat{f}(P, \mc S_H) = \frac{\hat{A}}{N_{Q,m_3} \hat{p}^{m_3-1}}~.
                \end{equation}
            \end{itemize}
    \end{itemize}
\end{protocol}
Notice indeed that from \cref{eq:A_simple}, 
\begin{align}
    \mbb E[\hat A] = f({Z^{\bm s (\set B)}},\Lambda_{\rm meas}^{\set P_n})  f({Z^{\bm s (\set B)}},\Lambda^{\set P_n}) \frac{f(P,  \mc S_H)}{f(\orbit{P}{H}, \mc S_H)}~,
\end{align}
where $\set B$ is the support of $Q=\phi(H)[P]$, and that 
\begin{align}
    \mbb E[N_{Q,m_3}] = f({Z^{\bm s (\set B)}},\Lambda_{\rm meas}^{\set P_n})  f({Z^{\bm s (\set B)}},\Lambda^{\set P_n}) f^{m_3}(\orbit{P}{H}, \mc S_H)~.
\end{align}
They are independent estimates, and their ratio is hence 
\begin{align}
    \mbb E \left[\frac{\hat{A}}{N_{Q,m_3}} \right]=   f(P,\mc S_H) f^{m_3-1}(\orbit{P}{H}, \mc S_H)~.
\end{align}

\end{document}